\documentclass[12pt]{amsart}
\usepackage{amscd,amssymb,amsmath,latexsym,enumerate}
\usepackage[mathscr]{euscript}
\usepackage{epsfig,graphics,graphicx}
\usepackage{fancybox}
\usepackage{verbatim}
\usepackage{tikz}
\usepackage{tikz-cd}

\definecolor{MyBlue}{cmyk}{1,0.13,0,0.63}
\definecolor{MyGreen}{cmyk}{0.91,0,0.88,0.52}
\newcommand{\mylinkcolor}{MyBlue}
\newcommand{\mycitecolor}{MyGreen}
\newcommand{\myurlcolor}{black}

\usepackage{hyperref}
\hypersetup{%
  bookmarksnumbered=true,bookmarksopen=false,%
  plainpages=false,
  linktocpage=true,%
  colorlinks=true,breaklinks=true,%
  linkcolor=\mylinkcolor,citecolor=\mycitecolor,urlcolor=\myurlcolor,%
  pdfpagelayout=OneColumn,%
  pageanchor=true,%
}
 
\usepackage{amsfonts}
\usepackage{color}
\usepackage{amsthm}
\usepackage{mathtools}
\usepackage{quiver}
\usepackage{dsfont}
\usepackage[kerning=true]{microtype}

\usepackage[margin=1in]{geometry} 

\title{Kubo Formulas and Bulk-Edge Correspondence for Curved Boundaries}

\author{Aren Martinian and Tom Stoiber}

\date{\today}

\newtheorem{theorem}{Theorem}[section]
\newtheorem{definition}[theorem]{Definition}
\newtheorem{proposition}[theorem]{Proposition}
\newtheorem{lemma}[theorem]{Lemma}
\newtheorem{corollary}[theorem]{Corollary}
\newtheorem{remark}[theorem]{Remark}
\newtheorem{example}[theorem]{Example}

\newcommand{\CM}{{\mathbb C}}
\newcommand{\NM}{{\mathbb N}}
\newcommand{\RM}{{\mathbb R}}
\newcommand{\SM}{{\mathbb S}}

\newcommand{\ZM}{{\mathbb Z}}
\newcommand{\PM}{{\mathbb P}}

\newcommand{\KM}{{\mathbb K}}
\newcommand{\HM}{{\mathbb H}}

\newcommand{\Aa}{{\mathcal A}}

\newcommand{\Pp}{{\mathcal P}}

\newcommand{\Bb}{{\mathcal B}}

\newcommand{\Ss}{{\mathcal S}}

\newcommand{\Tt}{{\mathcal T}}

\newcommand{\Cc}{{\mathcal C}}

\newcommand{\Ll}{{\mathcal L}}
\newcommand{\Qq}{{\mathcal Q}}
\newcommand{\Kk}{{\mathcal K}}
\newcommand{\Hh}{{\mathcal H}}
\newcommand{\Xx}{{\mathcal X}}
\newcommand{\Yy}{{\mathcal Y}}
\newcommand{\Zz}{{\mathcal Z}}

\newcommand{\one}{{\bf 1}}

\newcommand{\Tr}{\mbox{\rm Tr}}

\newcommand{\Ind}{{\rm Ind}} 
\newcommand{\Ker}{{\rm Ker}} 
 
\newcommand{\sgn}{{\rm sgn}}

\newcommand{\ep}{\varepsilon}
\newcommand{\pdv}{\partial}

\newcommand{\difd}{\textup{d}}

\providecommand{\abs}[1]{\left \lvert#1 \right \rvert} 
\providecommand{\norm}[1]{\left \lVert#1 \right \rVert}

\newcommand{\scalarmatrix}{H_-}

\DeclareMathOperator*{\wlim}{w-lim}

\DeclareMathOperator*{\Ext}{Ext}
\DeclareMathOperator*{\Id}{Id}
\DeclareMathOperator*{\Dom}{Dom}
\DeclareMathOperator{\diam}{diam}
\DeclareMathOperator{\supp}{supp}

\usetikzlibrary{calc,decorations.markings}

\usepackage[style=ieee,citestyle=numeric,backend=biber,sorting=nyt,url=false]{biblatex}

\AtEveryBibitem{\clearfield{issn}}	
\AtEveryCitekey{\clearfield{issn}}
\AtBeginBibliography{\small}		

\DeclareFieldFormat{postnote}{#1}

\DeclareFieldFormat[article,inbook,incollection,inproceedings,patent,thesis,unpublished]{title}{\textit{#1\isdot}} 

\DeclareFieldFormat{journaltitle}{#1}  

\DeclareFieldFormat[book,inbook,incollection,inproceedings]{series}{#1} 

\defbibheading{references}[\refname]{%
\section*{#1}%
\addcontentsline{toc}{section}{References} %
\markboth{#1}{#1}
}

\begin{document}

\begin{abstract}
Strong topological insulators are classified by integer invariants which admit different real-space expressions. Inspired by recent work studying these invariants for spaces with curved boundaries, we revisit the problem of equivalence of the various expressions, including Fredholm index pairings and Kubo formulas involving half-space projections. Using Roe algebras and a variant of $KK$-theory suitable for non-separable $C^*$-algebras, we prove a general form of bulk-edge correspondence for the index pairing with a general position-space Dirac operator in the presence of arbitrary boundaries. For spaces coarsely equivalent to $\RM^d$, we show that up to multiplicity, these pairings are equal to those obtained with the standard dual Dirac operator, with the multiplicity explicitly given as the topological degree of the symbol function. In addition, we prove a generalized Kubo formula which computes the general index pairing by a real-space formula.
\end{abstract}

\maketitle
\vspace{-0.75cm}

\tableofcontents

\newpage

\section{Introduction}
Topological phases of free-fermion systems are often characterized by $\ZM$- or $\ZM_2$-valued invariants which remain stable under perturbations that preserve the spectral gap. In periodic systems the integer invariants are often defined in momentum space, where they can be computed from vector bundles over the Brillouin zone and expressed by familiar integral formulas such as Chern numbers. However, there need not exist a useful momentum-space description in many situations of physical and mathematical interest, such as disordered systems, quasicrystals, and systems with curved or irregular boundaries. This makes it necessary to formulate the strong topological invariants directly in real space and to understand their relation to physically measurable quantities such as conductance and edge currents. Two closely related problems arise naturally in this context: first, to compare the various real-space formulas for the bulk invariant, and second, to relate these bulk quantities to boundary observables through bulk-edge correspondence.

In many settings, strong topological insulators in the complex symmetry classes of Altland and Zirnbauer \cite{Altland1997, Zirnbauer1996} are characterized by a single integer, called the strong topological invariant, which is computed from the band-flattening $\sgn(H)$ of a self-adjoint Hamiltonian $H$ which has a spectral gap around $0$. It is now well-understood that this topological invariant corresponds to the index pairing between either the K-theory class of the Fermi projection or that of the Fermi unitary and the K-homology class of a position-space Dirac operator.

In any dimension this invariant can be expressed as the index of a suitable Fredholm operator. In two dimensions one obtains the well-known expression \cite{BellissardVanElstSchulzBaldes1994}
\begin{equation}
\label{eq:index}
\Sigma_2=\Ind\left(p \frac{X_1+ \imath X_2}{|X_1+ \imath X_2|}p + 1-p\right),
\end{equation}
where $p=\chi(H < 0)$ is the Fermi projection, $X_j$ are the position operators on $\Xx = \ZM^2$ and one put $0/|0|=1$ at the origin. With the conventional physical units this integer corresponds to the quantum Hall conductance.

Beginning with the seminal works \cite{BellissardVanElstSchulzBaldes1994} and \cite{AvronSeilerSimon1994}, several real-space formulas for the strong invariant have been derived. The first is a rather direct real-space formulation of the momentum-space formula of \cite{TKNN}, which is the integral of the Berry curvature over the Brillouin torus. The above authors replaced the integral by the trace per unit volume and the momentum-space derivatives with commutators involving position operators. In two spatial dimensions under suitable hypotheses, this takes the form
\begin{equation}
\label{eq:kubo_tpuv_intro}
\Sigma_2 = -2\pi \imath\, \mathrm{Tr}_{\mathrm{av}}\big( p\,[[X_1,p]\,,[X_2,p]]\big).
\end{equation} 
A closely related, and in many respects better behaved, expression replaces the trace per unit volume by the usual Hilbert-space trace by introducing the projection operators for the two standard half-spaces $\HM_1, \HM_2$:
\begin{equation}
\label{eq:double_comm_intro}
\sigma_{\HM_1,\HM_2} = -2\pi \imath\, \mathrm{Tr}\big(p\,([[P_{\HM_1},p],[P_{\HM_2},p]])\big).
\end{equation} 
We will refer to the latter expression as the Kubo formula. Different formulas have historically been called by this name, and the equivalence between the two expressions $\Sigma_2$ and $\sigma_{\HM_1,\HM_2}$ above is often taken for granted.

This formulation is closely related to Kitaev's real-space formula~\cite{Kitaev06}, which is expressed in terms of a partition of $\mathbb{R}^2$ into three disjoint cone-like regions $A,B,C$,
$$\sigma_{A,B,C} = 12 \pi \imath \mathrm{Tr}\big(P_A p P_B p P_C p -  P_A p P_C p P_B p\big).$$
As pointed out by Ludewig and Thiang~\cite{LudewigThiang25}, the Kubo formula and Kitaev's formula are closely related and can be derived from each other if one allows a general enough notion of half-spaces. 

Another important motivation for discussing the Kubo formula for non-standard half-spaces comes from bulk-edge correspondence with curved boundaries. In the formulation of Drouot and Zhu~\cite{drouot2024bulk}, two transverse half-spaces $U$ and $V$ are fixed such that $U$ is the half-space which the system with boundary lives on, whereas $V$ is a fiducial half-space such that the edge conductance measures charge transport from $V$ to $V^c$. 
In this setting, the authors show that the edge conductance $\sigma_e^{U,V}$ is related to the Hall conductance $\sigma_{U,V}$ by
\[\sigma_e^{U,V}(H_e) = \sigma_{U,V}(p_+) -\sigma_{U,V}(p_-),\]
where $H_e$ is an edge Hamiltonian, $p_{\pm}$ are Fermi projections associated to the bulk Hamiltonians $H_{\pm}$, and the bulk Hall conductances are computed using \eqref{eq:double_comm_intro} with the non-standard half-spaces $U,V$. It is also shown that
\begin{equation}
\label{eq:sigmaUV}
\sigma_{U,V} = \chi(U,V)\, \sigma_{\mathbb{H}_1,\mathbb{H}_2}
\end{equation}
where the integer $\chi(U,V)$ is, under suitable assumptions on the half-spaces, the intersection number of the oriented boundaries of $U$ and $V$. 

In this paper we prove bulk-edge correspondence motivated by the formulation of \cite{drouot2024bulk}, but in higher dimensions. The bulk invariant generalizing the Hall conductivity is then computed by plugging the Fermi projection or Fermi unitary of a gapped Hamiltonian into a Kubo formula which uses a collection $\Yy_1,\ldots,\Yy_d\subset \Xx$ of generalized half-spaces of a space $\Xx$ with bounded geometry, for example $\Xx=\RM^d$ or $\Xx=\ZM^d$. 
After restricting that Hamiltonian to one of the half-spaces, say $\Yy_d$, one obtains an edge topological invariant which is computed in terms of the remaining half-spaces $\Yy_1,\ldots,\Yy_{d-1}$, and which is equal to the bulk invariant up to a sign. Moreover, for $\Xx=\RM^d$ we then determine the relative multiplicity, generalizing the intersection number, compared to the strong topological invariant which is defined in terms of the standard half-spaces $\HM_1,\ldots,\HM_d$.

There is a by now well-established analytic approach to bulk-edge correspondence which works directly with concrete numerical expressions for the bulk and edge invariants \cite{ElbauGraf, GrafPorta, ElgartGrafSchenker2005,GrafShapiro, drouot2024bulk}. A closely related analytic approach in the continuum setting for domain-wall models is to use pseudodifferential methods, such as the Fedosov-Hörmander index theorem (see e.g., \cite{DrouotMicrolocal2021, Bal2022TopologicalInvariants,BalContinuous2019}). We instead follow the K-theoretic approach, which goes back to \cite{KRS04}: Both bulk and edge topological invariants are encoded into K-theory classes over different $C^*$-algebras which are then related by boundary maps. 
From there, the numerical invariants are obtained by pairing with cyclic cocycles or K-homology classes. While initially introduced for ergodic discrete and continuous models on $\ZM^2$ \cite{KRS04} and $\RM^2$ \cite{KellendonkSchulzBaldes2004}, this approach has been generalized widely: to arbitrary dimensions, all Altland-Zirnbauer symmetry classes, crystalline (higher-order) topological insulators and more general aperiodic models, using either groupoid methods or coarse geometric methods \cite{ProdanSchulzBaldes2016,Kubota17,bourne2017k,BourneProdan2018,EwertMeyer,AlldridgeMaxZirnbauer2020, PoloOjitoProdanStoiber2025}.

More precisely, we will follow the KK-theoretic approach of \cite{BourneCareyRennie2015, bourne2017k}. In KK-theory, one can treat K-theory classes and K-homology classes on the same footing and the index pairing and boundary maps are just special cases of so-called Kasparov products.

In Section~\ref{sec:roe} we first develop the background necessary to define topological invariants via K-homology. 
We assume that $\Xx$ is a proper metric space with bounded geometry, which covers both discrete and continuous settings. Consider a collection of $d$ transverse subsets $\Yy_1,\dots, \Yy_d$ which we call half-spaces. As a generalization of the position operators we may use the signed distance functions $\delta(x) = (\delta_{\Yy_1},\dots, \delta_{\Yy_d})(x)$, with 
\[\delta_{\Yy_i}(x) = \begin{cases}
\mathrm{dist}(x,\Yy_i^c) & x\in \Yy_i\\
-\mathrm{dist}(x,\Yy_i) & x\in \Yy^c_i.
\end{cases}\]
Such a function $\delta$ is the prototypical example of what we will call a \emph{Dirac-type} function $f: \Xx \to \RM^d$, whose crucial property is that the position-space Dirac operator $D_f = \sum_i f_i \otimes \sigma_i$, for $\sigma_1,\ldots,\sigma_d$ the generators of the complex Clifford algebra $\CM_d$, yields a well-defined K-homology class $[D_f]$. 
Such a K-homology class produces integer-valued topological invariants through the usual index pairing 
$$\langle \cdot, \cdot \rangle: K_{d\bmod 2}(A) \times K^{d\bmod 2}(A)\to \ZM$$
with the K-theory class of a projection or unitary in a suitable $C^*$-algebra $A$. In our setting one can choose for $A$ the so-called uniform Roe algebra $C^*_u(\Xx)$, such that Fermi projections or Fermi unitaries of gapped short-range Hamiltonians define a bulk topological invariant by pairing with a K-homology class $[D_{\Yy_1,\ldots,\Yy_d}]$.
With the relative Roe algebra $C^*_u(\Zz_d\subset \Xx)$ of operators supported near the boundary $\Zz_d$ of the half-space $\Yy_d$, one can similarly define an edge index by pairing with an edge Dirac operator $D_{\Yy_1,\ldots,\Yy_{d-1}}$.

In Section~\ref{sec:bulk_edge_correspondence} we will prove that the bulk and edge indices are related by the following theorem:
\begin{theorem}
\label{th:bulk_edge_intro}
Let $\Yy_1,\ldots,\Yy_d$ be a polynomially coarsely transverse collection of half-spaces with boundaries $\Zz_1,\dots, \Zz_d$ in a space $\Xx$ with bounded geometry, and let $\pdv_{\Yy_d}: K_{d\bmod 2}(C^*_u(\Xx)) \to  K_{(d-1)\bmod 2}(C^*_u(\Zz_d\subset \Xx))$ be the Mayer-Vietoris boundary map associated to the decomposition $\Xx=\Yy_d\cup \Yy_d^c$. 
Then
$$\langle x, [D_{\Yy_{1},\ldots,\Yy_d}]\rangle = (-1)^{d} \langle \pdv_{\Yy_d}(x),  [D_{\Yy_{1},\ldots,\Yy_{d-1}}]\rangle, \qquad \forall x\in K_{d\bmod 2}(C^*_u(\Xx)).$$
\end{theorem}
As highlighted before \cite{EwertMeyer, Kubota17, LudewigThiang20, LudewigThiang25}, the Mayer-Vietoris boundary map gives the relation between bulk and edge topological invariants of Hamiltonians. However, it also applies more generally to bulk-interface correspondence, where it maps differences of bulk K-theory classes to interface invariants. The relative sign of $(-1)^{d}$ depends on conventions for the index pairings and boundary maps; since we use definitions compatible with \cite{ProdanSchulzBaldes2016} we obtain the same relative signs.

We will prove Theorem~\ref{th:bulk_edge_intro} by generalizing the  KK-theoretic approach of \cite{bourne2017k} which was developed for crossed-product $C^*$-algebras suitable for topological insulators on a lattice with straight boundaries. The first step will be to encode the Mayer-Vietoris boundary map into an extension class $\Ext_E$ of a suitable exact sequence of $C^*$-algebras. The statement of bulk-edge correspondence then reduces via associativity of the intersection product, which we will simply call the Kasparov product, to proving
$$[{\Ext}_E] \otimes [D_{\Yy_1,\ldots,\Yy_{d-1}}] = [D_{\Yy_1,\ldots,\Yy_{d}}]$$
where $\otimes$ denotes the Kasparov product. In the unbounded picture of KK-theory the Kasparov product corresponds to the anti-commuting sum of Dirac operators provided it is well-defined. Morally speaking, the relation therefore follows from the trivial decomposition
$$\delta_{\Yy_d} \otimes \sigma_d + \sum_{j=1}^{d-1} \delta_{\Yy_j} \otimes \sigma_j = \sum_{j=1}^d \delta_{\Yy_j} \otimes \sigma_j $$
where $[\delta_{\Yy_d}]$ will represent the extension class.

A major technical problem compared to \cite{bourne2017k}, however, is that the standard results about extension classes and existence of Kasparov products impose assumptions of separability on the $C^*$-algebras. Unlike the crossed-product algebras, the uniform Roe $C^*$-algebras that we consider are not separable, and in general not even $\sigma$-unital. These difficulties with using the usual KK-theoretic formalism in the coarse-geometric context have been pointed out before \cite[Remark~1.5]{SpakulaWillett2013} and mean that existence, uniqueness and associativity of Kasparov products are no longer guaranteed. To overcome this problem and to make the above blueprint work, we will therefore replace the usual KK-functor with the variant $KK^{\mathrm{sep}}$ introduced by Skandalis \cite{Skandalis1988}, which defines KK-theory for non-separable $C^*$-algebras as projective limits over their separable subalgebras. We will show that the Kasparov products and extension classes can then be computed on nets of carefully chosen separable subalgebras, which allows one to still use standard results such as Kucerovsky's criterion \cite{Kucerovsky}. 

Our second main result in Section~\ref{sec:euc_space} concerns the classification of Dirac-type functions in the case $\Xx=\RM^d$. We show that, up to equivalence in K-homology, all such functions are classified by a single integer.
Explicitly, this integer is the degree of the Dirac-type function $f$ on a sufficiently large sphere $R\SM^{d-1}$:
\begin{theorem}\label{thm:degree_intro}
Let $f:\RM^d\to \RM^d$ be a continuous Dirac-type function for $d > 1$. Define the asymptotic degree 
\[\deg(f)= \lim_{R\to \infty} \deg\left( \frac{f}{|f|}\bigg\vert_{R\SM^{d-1}}\right),\]
with the mapping degree in the sense of continuous maps $\SM^{d-1}\to \SM^{d-1}$. 
Then the K-homology class $[D_f]$ in $K^{d \bmod 2}(C^*_u(\RM^d))$ is related to the class $[D_{\mathrm{std}}]$ of the standard Dirac operator $D_{\mathrm{std}}=\sum_{i=1}^d X_i\otimes \sigma_i$ via $$[D_f]= \deg(f)[D_{\mathrm{std}}].$$
\end{theorem}

The proof of Theorem~\ref{thm:degree_intro} is based on the fact that the mapping degree is a complete homotopy invariant for maps between spheres.

A reader familiar with the Brouwer degree may notice that for a proper continuous function $f$, we have $\deg(f)=\deg_{\mathrm{B}}(f, \RM^d,0)$ with the Brouwer degree on the unbounded domain $\RM^d$ and reference point $0\in \RM^d$; we recall those notions in Section~\ref{sec:brouwer}. This gives a local way to compute the degree.

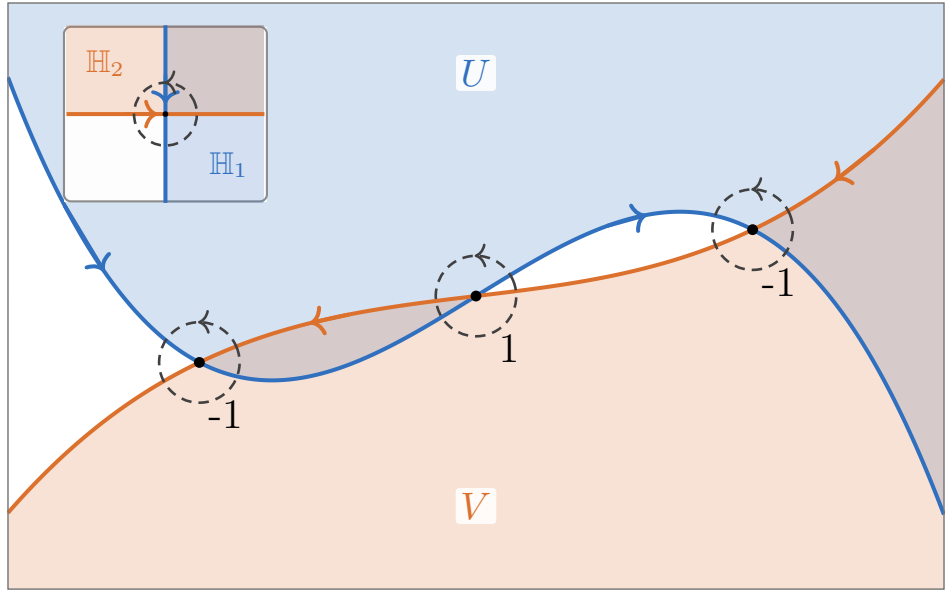
\begin{figure}[htbp]
\centering
\definecolor{Ublue}{RGB}{48,112,190}
\definecolor{Vorange}{RGB}{220,112,45}
\definecolor{Overlap}{RGB}{142,92,170}

\tikzset{
  boundary curve/.style={
    line width=1.15pt
  },
  local circle/.style={
    densely dashed,
    black!75,
    line width=.75pt,
    postaction={decorate},
    decoration={markings,
      mark=at position .20 with {\arrow{>}}
    }
  },
  sign label/.style={
    fill=white,
    fill opacity=.88,
    text opacity=1,
    rounded corners=1pt,
    inner sep=1.4pt,
    font=\small
  }
}

\resizebox{0.75\linewidth}{!}{%
\begin{tikzpicture}[x=0.88cm,y=0.88cm]

\begin{scope}[scale=1]
  \clip (-5.35,-3.35) rectangle (5.35,3.35);

  \path[fill=Ublue,fill opacity=.20]
    (-5.35,3.35) -- (-5.35,3.35) --
    plot[domain=-5.35:5.35,samples=160,smooth]
      (\x,{-0.038*(\x)^3+0.62*\x})
    -- (5.35,3.35) -- cycle;

  \path[fill=Vorange,fill opacity=.20]
    (-5.35,-3.35) -- (-5.35,0) --
    plot[domain=-5.35:5.35,samples=160,smooth]
      (\x,{0.012*(\x)^3+0.12*\x})
    -- (5.35,-3.35) -- cycle;

  \draw[Ublue,boundary curve]
    plot[domain=-5.35:5.35,samples=180,smooth]
      (\x,{-0.038*(\x)^3+0.62*\x});

  \draw[Vorange,boundary curve]
    plot[domain=-5.35:5.35,samples=180,smooth]
      (\x,{0.012*(\x)^3+0.12*\x});


  \draw[->,Ublue,line width=1.15pt]
    (-4.70,1.031) -- (-4.25,.282);

  \draw[->,Ublue,line width=1.15pt]
    (1.50,.802) -- (1.95,.927);

  \draw[->,Vorange,line width=1.15pt]
    (-1.45,-.211) -- (-1.90,-.310);

  \draw[->,Vorange,line width=1.15pt]
    (4.55,1.676) -- (4.10,1.319);

  \draw[black!55,thin] (-5.35,-3.35) rectangle (5.35,3.35);

  \node[sign label,text=Ublue] at (0,2.55) {$U$};
  \node[sign label,text=Vorange] at (0,-2.40) {$V$};

  \coordinate (p1) at (-3.1623,-0.7589);
  \coordinate (p2) at (0,0);
  \coordinate (p3) at (3.1623,0.7589);

  \foreach \p/\lab in {p1/1,p2/2,p3/3}{
    \draw[local circle]
      ($ (\p)+(20:.46) $)
      arc[start angle=20,end angle=380,radius=.46];
    \fill[black] (\p) circle[radius=1.55pt];
  }

  \node at (-3.1623+0.3,-0.7589-0.6) {-1};
  \node at (0+0.3,0-0.6) {\phantom{-}1};
  \node at (3.1623+0.3,0.7589-0.6) {-1};

\end{scope}

\begin{scope}[shift={(-3.55,2.08)},scale=.58]

  \fill[white,fill opacity=.96]
    (-2.,-1.72) rectangle (2.,1.72);

  \draw[black!45,rounded corners=2pt,line width=.6pt]
    (-2.,-1.72) rectangle (2.,1.72);

  \path[fill=Ublue,fill opacity=.20]
    (0,-1.75) rectangle (1.95,1.75);

  \path[fill=Vorange,fill opacity=.20]
    (-1.95,0) rectangle (1.95,1.75);

  \draw[Ublue,boundary curve]
    (0,-1.75) -- (0,1.75);

  \draw[Vorange,boundary curve]
    (-1.95,0) -- (1.95,0);

  \draw[->,Ublue,line width=1.05pt]
    (0,.78) -- (0,.20);

  \draw[->,Vorange,line width=1.05pt]
    (-.80,0) -- (-.15,0);

  \draw[local circle]
    (20:.62)
    arc[start angle=20,end angle=380,radius=.62];

  \fill[black] (0,0) circle[radius=1.4pt];

  \node[font=\scriptsize,text=Ublue]
    at (1.25,-1) {$\mathbb H_1$};

  \node[font=\scriptsize,text=Vorange]
    at (-1.18,1) {$\mathbb H_2$};

\end{scope}

\end{tikzpicture}%
}

\caption{Brouwer degree and intersection number in an example. The arrows on
$\pdv U$ and $\pdv V$ indicate their orientations, with the shaded
half-space lying to the left. Inset: A counter-clockwise loop enters first
$\mathbb H_1$, then $\mathbb H_2$. This normalizes the local contributions to
the Brouwer degree.}
\label{fig:brouwer-degree-intersection}
\end{figure}

\begin{example}
\label{ex:intersections}
{\rm Fix $d=2$ and consider two transverse half-spaces $U,V$ whose boundaries are smooth curves with only finitely many transverse intersection points. For $\delta=(\delta_U,\delta_V)$, the Brouwer degree $\deg_{\mathrm B}(\delta, \RM^d, 0)$ is the sum of local contributions from small circles around each preimage of $(0,0)$, i.e., the intersections of the boundary curves. The contribution is $1$ if the orientation matches the standard half-spaces and $-1$ if it is opposite, see Figure~\ref{fig:brouwer-degree-intersection}. 
We show in Subsection~\ref{sec:deg_hs} that the Brouwer degree in two dimensions is, up to sign conventions, the intersection number as in \cite{drouot2024bulk}.}
\end{example}

In higher dimensions, one can similarly try to compute the degree of $\delta=(\delta_{\Yy_1},\ldots,\delta_{\Yy_d})$ by determining all points for which the boundaries of all half-spaces $\Yy_1,\ldots,\Yy_d$ intersect simultaneously. In the special case where boundaries are smooth and those intersection points are isolated and transverse, we then again read off the degree immediately from the local orientations. 
Again in Subsection~\ref{sec:deg_hs}, we prove additivity and construct examples of any degree $n$.

Our last main result, in Section~\ref{sec:kubo}, is a proof of the bulk and edge Kubo formulas for $\Xx = \RM^d$. Here the reason why we work with uniform Roe $C^*$-algebras, rather than their non-uniform counterparts, becomes apparent. Either of those versions can be used for the purpose of bulk-edge correspondence on the level of K-theory and the non-uniform Roe algebras have better coarse-geometric properties and their $K$-theory groups are simpler (see, for example, \cite{Roe2003,Roe1996}). However, the \textit{Schwartz} subalgebra $\Ss^1_u(\Xx)$ of locally trace-class rapidly decaying elements is dense and spectrally invariant only for the uniform Roe algebras, hence it can be used to define numerical invariants via the Kubo formula. 

The generalized Kubo cocycle associated to half-spaces $\Yy_1,\ldots,\Yy_d$ is a multilinear functional which takes as input operators $a_0,\ldots,a_d$ in the Schwartz algebra.
Adopting the convention that a permutation is identified with its parity, we write
\[\eta_{\Yy_1,\ldots,\Yy_d}(a_0, \dots, a_{d}) = \sum_{\sigma \in S_d} (-1)^{\sigma} \Tr\left(a_0 \prod_{i=1}^d [F_{\sigma(i)}, a_i]\right),\]
where it is convenient to use the switch functions $F_{i} = 2 P_{\Yy_i}-1$. Similarly one can define an edge Kubo cocycle for the first $d-1$ half-spaces.

\begin{theorem}[Kubo Formula]\label{thm:kubo_intro}
Suppose $\Xx = \RM^d$, $d>1$, and $\Yy_1,\dots, \Yy_d$ are polynomially coarsely transverse half-spaces.

In even dimension $d$, for every projection $p\in \Ss^1_u(\Xx)^+$  in the unitization of the Schwartz algebra, we have
\begin{equation}
\label{eq:kubo_bulk_index_intro}
 \langle [p]_0-[s(p)]_0,  [D_{\Yy_1,\ldots,\Yy_d}]\rangle= \Xi_d \eta_{\Yy_1,\ldots,\Yy_d}(p,\ldots,p) \in \ZM
\end{equation}
where $\Xi_d$ is the normalization constant of Theorem~\ref{th:kubozd} and $s(p)$ denotes the scalar part of the projection $p$. 

If $u\in \Ss^1_u(\Zz_d\subset \Xx)^+$ is a unitary with $[u]_1=\pdv_{\Yy_d} ([p]_0)$ then in addition
$$\langle [u]_1, [D_{\Yy_1,\ldots,\Yy_{d-1}}]\rangle=\Xi_{d-1} \eta_{\Yy_1,\ldots,\Yy_{d-1}}(u^*,u,u^*,\dots,u)\in \ZM.$$
The analogous formulas hold for odd $d$ with the roles of projections and unitaries reversed.
\end{theorem}
As these numerical formulas reproduce the bulk and edge topological invariants, using the Kubo formulas for non-trivial half-spaces one obtains a relative factor 
$$\eta_{\Yy_1,\ldots,\Yy_d}(p,\ldots,p)=\deg(\delta_{\Yy_1},\ldots,\delta_{\Yy_d}) \,\eta_{\mathbb{H}_1,\ldots,\mathbb{H}_d}(p,\ldots,p)$$
compared to the standard half-spaces $\mathbb{H}_i$. In two dimensions we therefore recover the relation with the intersection number \eqref{eq:sigmaUV} via Example~\ref{ex:intersections}.

\begin{example}
{\rm In two dimensions the relation between bulk and edge conductance as shown in \cite{drouot2024bulk} is then a consequence of the two separate numerical formulas. For lattice models, the Mayer-Vietoris boundary map always sends the K-theory class of the Fermi projection $[p]_0$ to that of an edge unitary $[u_\Delta]_1$, and the two invariants related by Theorem~\ref{th:bulk_edge_intro} are
$$\langle [p]_0, [D_{\Yy_1,\Yy_2}] \rangle = \sigma_{\Yy_1,\Yy_2} = -2\pi \imath \Tr(p[[P_{\Yy_1},p],[P_{\Yy_2},p]])$$
and
$$\langle [u_\Delta]_1, [D_{\Yy_1}] \rangle = -\Tr(u_\Delta^* [P_{\Yy_1}, u_\Delta]).$$
Using the functional calculus of a half-space Hamiltonian $\hat{H}$, we recover the latter expression, yielding the usual expression for the edge conductance of the form \cite[Section 6]{LudewigThiang20}
$$\langle [u_\Delta]_1, [D_{\Yy_1}] \rangle = \imath \pi \Tr(\rho'_\Delta(\hat{H}) [P_{\Yy_1}, \hat{H}]),$$
where $\rho$ is a smooth switch function.
Some more examples of bulk-edge correspondence and bulk-interface correspondence including unbounded Hamiltonians are given in Section~\ref{ssec:bec_examples}.}$\diamond$
\end{example}

Our Theorem~\ref{thm:kubo_intro} has substantial overlap with  \cite{LudewigThiang25}, where it is shown using techniques from coarse cohomology that the Kubo formula produces an integer in the case of finite-propagation locally trace-class operators on a discrete space $\Xx$. However, as pointed out by the authors, the Kubo formula in the physically relevant setting of projections with faster-than-polynomial decay cannot be derived from their results immediately, as K-theory classes do not generally admit representatives with finite-propagation. While we do also use some notions of coarse cohomology for convenience, our proof follows a rather elementary approach: The Kubo formula  for the case of the standard half-spaces of $\ZM^d$ can be shown to hold via a relatively simple modification of the arguments used for the trace-per-unit-volume version \cite{Kubota17}. For completeness and to fix normalizations, we give a proof in Appendix~\ref{sec:kubozd}. We then show that in the physically relevant case of subsets of $\RM^d$, one can reduce the Kubo formula for arbitrary collections of half-spaces to this special case. In particular, using homotopy it is enough to prove the Kubo formula for one example of any given degree. 

This strategy may also be useful in a more general setting: for example, to prove formulas for the $\ZM_2$-invariants of the real Altland-Zirnbauer classes  \cite{SchulzBaldes2015Z2, Drouot2026GeometricBulkEdge, Kitaev09}, which cannot be computed by Kubo-style cocycles. Indeed, our KK-theoretic approach to bulk-edge correspondence would extend to real symmetry classes with minimal modifications (compare e.g., \cite{AlldridgeMaxZirnbauer2020}). 
Another possible application would be to study mobility-gapped topological insulators (see e.g., \cite{BellissardVanElstSchulzBaldes1994, ProdanSchulzBaldes2016,ProdanSchulzBaldes2016Odd,Shapiro2020}). In that case, one deals with spectral projections with only semi-uniform exponential decay, which means they are outside the scope of Roe $C^*$-algebras. However, one can still use index theory and cyclic cohomology more directly. In particular, the part of our argument which deforms arbitrary Dirac-type functions to standard representatives will be applicable in that setting, as the index pairing still has the required homotopy invariance \cite{RichterSchulzBaldes2001, BolsSchenkerShapiro2023, Stoiber2025SpectralLocalizer}.

\section{Roe Algebras and KK-theory Preliminaries}
\label{sec:roe}

In this section, we develop the machinery of uniform Roe $C^*$-algebras in a way that handles discrete and continuous metric spaces simultaneously.

\subsection{Roe Algebras and Coarse Geometry}

Throughout this paper, $\Xx$ will denote a proper metric space with bounded geometry. Later, we will sometimes need to assume polynomial growth as well.
\begin{definition}
\label{def:poly_growth}
Let $\Xx$ be a proper metric space with a coarsely dense uniformly discrete subset $\Lambda\subset \Xx$.
We say $\Xx$ has bounded geometry if for any $r>0$, there exists some $N_r>0$ such that
\begin{equation}\label{eqn:bounded_geometry}
\sup_{x \in \Lambda} \#|B_r(x)\cap \Lambda| \leq N_r.
\end{equation}

We say $\Xx$ has polynomial growth if there are $C, \nu > 0$ such that
\begin{equation}\label{eqn:poly_growth}
\sup_{x \in \Lambda} \#|B_r(x)\cap \Lambda| \leq C (1+r^\nu), \qquad \forall r>0.
\end{equation}
\end{definition}
Note that these properties are independent of the choice of $\Lambda$ taken, so we take them as properties of $\Xx$.

\begin{definition}
A geometric module $\Hh_\Xx$ is a separable Hilbert space with a nondegenerate representation $\rho_0: C_0(\Xx) \to \Bb(\Hh_\Xx)$.
We use the same notation for functions and images under $\rho_0$.
\end{definition}
The representation of $C_0(\Xx)$ extends to a projection-valued measure. Consequently, one can also consider multiplication by arbitrary Borel functions. We will say that $\Hh_\Xx$ is \textit{coarsely faithful} if the support of this measure is coarsely dense. We consider general geometric modules for considerable flexibility; the key examples to keep in mind throughout this paper are $\Xx=\RM^d$ with the two coarsely faithful geometric modules $L^2(\RM^d)$ and $\ell^2(\ZM^d)$ with the latter representation given by restriction to $\ZM^d$.

Let us recall the definitions of the Roe algebra and uniform Roe algebra associated to a geometric module:
\begin{definition}
\label{def:roe_algebra}
\leavevmode\newline
\vspace{-0.5cm}
\begin{enumerate}
\item[(i)] 
An operator $T\in \Bb(\Hh_{\Xx})$ is said to have finite propagation if there exists some constant $\rho>0$ such that 
\[f T g =0\]
for all $f,g\in C_0(\Xx)$ with $\mathrm{dist}(\mathrm{supp}(f),\mathrm{supp}(g))> \rho$. The infimum of all $\rho$ is called the propagation range $\mathrm{prop}(T)$.
\item[(ii)] 
We say that  $T\in \Bb(\Hh_{\Xx})$ is locally compact if $f T, T f$ are compact for every $f\in C_0(\Xx)$.
\item[(iii)]
The Roe algebra $C^*(\Xx, \Hh_\Xx)$ is the norm-closure of the locally compact operators of finite propagation on $\Hh_\Xx$.    
\end{enumerate}
\end{definition}

\begin{definition}[{\cite{Spakula2009}}]
\label{def:uniform_roe}
\leavevmode\newline
\vspace{-0.5cm}
\begin{enumerate}
\item[(i)]
A family of operators $(T_i)_{i\in I}$ on $\Hh_\Xx$ is called uniformly approximable by finite-rank operators if for every $\epsilon>0$ there exists some $N\in \NM$ such that \[\sup_{i\in I}\inf_{\substack{S\in \Bb(\Hh_\Xx)\\\mathrm{rank}(S)\leq N}} \|T_i-S\| \leq \epsilon.\]
\item[(ii)]
We say that $T\in \Bb(\Hh_{\Xx})$ is uniformly locally compact  if for every $r>0$, the family
\[\{\chi_{B_r(x)} T, T\chi_{B_r(x)}: \, x\in \Xx\}\]
is uniformly approximable by finite-rank operators.
\item[(iii)] Let $\CM_u(\Xx,\Hh_\Xx)$ be the algebra of uniformly locally compact operators of finite propagation on $\Hh_\Xx$. Then the uniform Roe algebra $C^*_u(\Xx, \Hh_\Xx)$ is the norm-closure of $\CM_u(\Xx,\Hh_\Xx)$.
\end{enumerate}
\end{definition}
As explained in the introduction, we will prefer to use the uniform Roe algebra, which is typically a proper subalgebra of the Roe algebra, as the non-uniform Roe algebra permits arbitrarily large local (on-site) degrees of freedom.

\begin{definition}\label{defn:relative_roe}
Let $\Xx$ be a proper metric space.
\begin{enumerate}
\item[(i)] The $R$-thickening $(A)_R$ of a set $A\subset \Xx$ is defined as
\[(A)_R = \bigcup_{x\in A} B_R(x).\]

\item[(ii)] A non-empty closed subset $\Zz \subseteq \Xx$ is called a coarse boundary of a set $\Yy\subset \Xx$ if there exists some constant $C>0$ such that \[\Zz \subseteq (\Yy)_C \cap (\Yy^c)_C\] and for all $R > 0$ there exists $S(R)$ such that 
\[(\Yy)_R \cap (\Yy^c)_R \subseteq (\Zz)_{S(R)}.\]
\item[(iii)] A half-space $\Yy\subset \Xx$ shall be a set which has a coarse boundary and for which $\Yy$ and $\Yy^c$ both have non-empty interior. 
\end{enumerate}
\end{definition}
Note that our definition of half-space is very broad. In particular, $\Yy$ is permitted to be any proper closed subset of $\RM^d$ with nonempty interior, since its topological boundary is also a coarse boundary.
\begin{definition}
We say that $T$ is supported near $\Zz\subset \Xx$ if there exists some $R \geq 0$ such that for every $f\in C_0(\Xx)$ with $\supp(f) \cap (\Zz)_R=\varnothing$, we have 
\[f T = 0 = Tf.\]
Let $\CM_u(\Zz\subset \Xx, \Hh_\Xx)$ denote the set of finite propagation, uniformly locally compact operators supported near $\Zz$. The relative Roe algebra of a subset $\Zz$, denoted $C^*_u(\Zz \subset \Xx, \Hh_\Xx)$ is the operator norm closure of $\CM_u(\Zz\subset \Xx,\Hh_\Xx)$. 
\end{definition}
The relative Roe algebras are closed two-sided ideals in $C^*_u(\Xx, \Hh_\Xx)$.
If a set $\Zz$ is a coarse boundary of a half-space $\Yy$ then the relative Roe algebras satisfy the coarse excision condition
\begin{equation}
\label{eq:coarse_excision}
C^*_u(\Zz\subset \Xx, \Hh_\Xx)=C^*_u(\Yy^c \subset \Xx, \Hh_\Xx)\cap C^*_u(\Yy\subset \Xx, \Hh_\Xx)
\end{equation}
for any geometric module $\Hh_\Xx$.
Indeed, by construction, an operator is supported near both $\Yy$ and $\Yy^c$ if and only if it is supported near $\Zz$.
Hence, $C^*_u(\Yy \subset \Xx,\Hh_\Xx) C^*_u(\Yy^c \subset \Xx,\Hh_\Xx) \subset C^*_u(\Zz \subset \Xx,\Hh_\Xx)$, which gives the desired containment.

\subsection{KK-theory for Non-Separable Algebras}
For the purposes of bulk-edge correspondence, we will need to  express boundary maps in $K$-theory as Kasparov products. For separable $C^*$-algebras this is standard \cite[Section 3]{jensen2012elements}, as in that case any semisplit extension $E$ provides the data to define an extension class $[\Ext_E]\in KK_1(A, B)$ whose Kasparov product with $K_i(A)\simeq KK_i(\CM, A)$ coincides up to sign with the boundary maps. 

However, the uniform Roe algebra $C^*_u(\Xx, \Hh_\Xx)$ for infinite $\Xx$ with coarsely faithful Hilbert module is not separable and in general not even $\sigma$-unital.
A separable algebra is one with a countable norm-dense subset and a $\sigma$‑unital algebra is one that contains a countable approximate unit.
As noted before (e.g., \cite[Remark 1.5]{SpakulaWillett2013}) the usual Kasparov product need not be well-defined and the same problem extends to other constructions of KK-theory such as extension classes. In this subsection, we recall some standard notions of KK-theory and describe how to handle the non-separable case.

We take all $C^*$‑algebras to be graded and complex, ideals to be closed and two‑sided, and homomorphisms to preserve gradings.
With these conventions in mind, let us define our main object of study:

\begin{definition}
Fix two graded $C^*$-algebras $A,B$. 
For a Hilbert $B$-module $E$, let $\Ll_B(E)$ denote the algebra of $B$-linear adjointable operators on $E$, and $\KM_B(E)$ denote the algebra of compact operators on $E$ as a right Hilbert $B$-module.

A $KK$‑cycle $(E,\phi,F)$ for $A, B$ consists of
\begin{enumerate}
    \item[(i)] a countably generated graded Hilbert $B$‑module $E$;
    \item[(ii)] a graded $*$‑homomorphism $\phi:A\to \Ll_B(E)$;
    \item[(iii)] an odd adjointable operator $F\in \Ll_B(E)$ 
   satisfying the compactness relations
   \[\{[F,\phi(a)],\ (F^2-1)\phi(a),\ (F-F^*)\phi(a)\} \subset \KM_B(E) \quad \hbox{ for all } a\in A\]
   with the graded commutator $[F,\phi(a)]$.
\end{enumerate}
A $KK$-cycle is \emph{degenerate} if the expressions in (iii) are not just in $\KM_B(E)$ but vanish identically for every $a\in A$.

\noindent Two $KK(A,B)$-cycles are \emph{homotopic} if there exists a $KK(A,C([0,1],B))$-cycle whose evaluations at $0$ and $1$ are unitarily equivalent to the respective cycles.
\noindent With the direct sum, homotopy classes of $KK$-cycles form an abelian group $KK(A,B)$.
\end{definition}
The degenerate cycles are the neutral element, as every degenerate cycle is homotopic to the $0$-module \cite[Lemma 17.2.3]{blackadar1998k}. Let us also note that KK-cycles which are the same up to unitary equivalence, compact perturbation, operator homotopy or homology are also homotopic, possibly after adding degenerate cycles \cite[Section 17.2]{blackadar1998k}. 

One does not need to impose any separability or $\sigma$-unitality conditions to obtain a group. We will nevertheless want to use a variant of KK-theory due to Skandalis \cite{Skandalis1988}, which is better behaved for non-separable algebras. For this, we first recall the functoriality of KK-groups:
\begin{remark}({\cite{Skandalis1984}})
\label{rmk:pullback}
{\rm 
Let $A,B,C$ and $D$ be graded $C^*$-algebras.
For a graded $*$-homomorphism $\phi:A\to B$, the pullback is given by
\[\phi^*:KK(B,C)\longrightarrow KK(A,C),\qquad [(\Hh,\pi,F)]\longmapsto[(\Hh,\pi\circ\phi,F)],\] 
while the pushforward is given by
\[\phi_*:KK(C,A)\longrightarrow KK(C,B),\qquad [(\Hh,\pi,F)]\longmapsto[(\Hh\hat{\otimes}_\phi B,\pi\hat{\otimes}_\phi \one,F\hat{\otimes} \one)].\]
These operations are well-defined on the level of $KK$-cycles and preserve homotopies; no additional separability or $\sigma$-unitality assumptions on $A,B,C$ are needed.

If in addition $D$ is a $\sigma$-unital, then one has the exterior product
\[\tau_D:KK(A,B)\longrightarrow KK(A\hat{\otimes}D,B\hat{\otimes}D),\qquad [(\Hh,\pi,F)]\longmapsto[(\Hh\hat{\otimes} D,\pi\hat{\otimes} \one,F\hat{\otimes} \one)].\]
Here, the $\sigma$-unitality of $D$ is used to ensure that the external tensor product $\Hh\hat{\otimes} D$ is countably generated. 
}
$\diamond$
\end{remark}

Since we can take pullbacks and pushforwards under inclusion maps, for general graded $C^*$-algebras $A,B$ one can then define \cite{Skandalis1988,AntoniniAzzaliSkandalis2020}
\[KK^{\mathrm{sep}}(A,B):= \varprojlim_{\substack{ A'\subset A \\ A' \text{ separable}}} KK(A', B) \simeq \varprojlim_{\substack{ A'\subset A \\ A' \text{ separable}}} \varinjlim_{\substack{ B'\subset B \\ B' \text{ separable}}}KK(A', B')\]
with the inverse limit over graded separable subalgebras of $A$ and direct limit over all separable graded subalgebras of $B$. 
The first equality is the primary definition; in practice, one uses that for separable $A'$, the direct limit satisfies $KK(A',B)\simeq \varinjlim_{B'\subset B}KK(A', B')$ via a canonical isomorphism, which is important to define the product. 

A class in $KK^{\mathrm{sep}}(A,B)$ is represented by a compatible family of classes in $KK(A',B)$. By functoriality, there is a canonical map $KK(A,B)\to KK^{\mathrm{sep}}(A,B)$ via pullback under the compatible family of inclusions. 

\begin{theorem}[{\cite{Skandalis1988}}]
There exists a product
\[\otimes_B: KK^{\mathrm{sep}}(A,B) \times KK^{\mathrm{sep}}(B,C) \to KK^{\mathrm{sep}}(A,C)\]
which coincides with the usual Kasparov product in the separable case and satisfies associativity, bilinearity,
naturality in both variables, and compatibility with pullbacks and pushforwards.
\end{theorem}

\begin{remark}{\rm
We briefly recall the construction of the product of $x\in KK^{\mathrm{sep}}(A,B)$ and $y\in KK^{\mathrm{sep}}(B,C)$: By the definition of inverse limit, $x$ pulls back to a class $x'\in KK(A',B)$ for any separable $A'\subset A$. 
By definition of the direct limit, for any fixed $A'$ and $x$ there exists some $B'$ such that $x'$ is the image of some class $x''\in KK(A',B')$ under pushforward. 

By an identical construction, $y$ pulls back to a class $y' \in KK(B', C)$, which is the image of some $y''\in KK(B', C')$. 

Finally, one can take the product $x''\otimes_{B'} y'' \in KK(A',C')$. The pushforwards of those products to $KK(A', C)$ give a compatible family in terms of $A'$ and hence an element of the projective limit $KK^{\mathrm{sep}}(A,C)$.

}$\diamond$
\end{remark}

For $n\geq 0$, the higher KK-groups are defined by 
\begin{equation}\label{eqn:higher_kk_def}
KK_n(A,B):=KK(A,B\hat{\otimes}\CM_n). \end{equation}
Here, the graded tensor product is taken with the complex Clifford algebra $\CM_n$, and likewise for $KK^{\mathrm{sep}}$. We take $\CM_0=\CM$, while
$\CM_1=\CM\one\oplus\CM\ep$ is the universal unital graded
$C^*$-algebra generated by an odd self-adjoint unitary
$\ep=\ep^*=\ep^{-1}$. 
The $n$-dimensional Clifford algebra is then given by
\[\CM_n= \CM_1\hat{\otimes}\cdots\hat{\otimes}\CM_1 = \langle \sigma_1,\dots, \sigma_n \rangle; \quad \sigma_i = 1^{\hat{\otimes}(i-1)} \hat{\otimes} \ep \hat{\otimes}  1^{\hat{\otimes}(n-i)}, \quad 1 \leq i \leq n \]
with $n$ factors, i.e. $\CM_n$ is the algebra generated by the anti-commuting odd self-adjoint
unitaries $\sigma_1,\ldots,\sigma_n$.

There are special elements (see Appendix~\ref{sec:index}) $[\hat{\rho}_d]\in KK(\CM_d, \CM_{d\bmod 2})$ which are KK-equivalences. The formal Bott periodicity isomorphism $KK_d^{\mathrm{sep}}(A,B) \simeq  KK_{d\bmod 2}^{\mathrm{sep}}(A,B)$ is given by the Kasparov product product with the class of the exterior product $1_B \hat{\otimes} [\hat{\rho}_d] \in KK(B \hat{\otimes}\CM_d, B \hat{\otimes} \CM_{d \bmod 2})$, and we may use it to tacitly identify these groups.

To define a general product
\begin{equation}
\label{eq:KK_product}
KK^{\mathrm{sep}}_i(A,B) \times KK_j^{\mathrm{sep}}(B,C) \to KK^{\mathrm{sep}}_{j+i}(A,C)
\end{equation}
one first applies $\tau_{\CM_i}: KK^{\mathrm{sep}}(B,C)\to KK^{\mathrm{sep}}(B\hat{\otimes} \CM_i, C\hat{\otimes} \CM_i)$ on the second factor to obtain a Kasparov product taking values in $KK^{\mathrm{sep}}(A, (C\hat{\otimes}\CM_j)\hat{\otimes}\CM_i)$. The maps $\tau_{\CM_i}$ are natural with respect to Kasparov products and are compatible with Bott periodicity, 
therefore they are isomorphisms and one still obtains an associative product even if the degrees in \eqref{eq:KK_product} are considered in $\ZM/2\ZM$ using Bott periodicity.

In this paper, we will take Kasparov products with classes defined by unbounded Dirac-type operators. We recall the unbounded picture that will be used below.

\begin{definition}[Definition 6 in \cite{Kucerovsky}]
\label{def:unbddkk}
An unbounded $KK(A,B)$-cycle is a triple $(E,\phi,D)$ consisting of a
countably generated graded Hilbert $B$-module $E$, a graded
$*$-homomorphism
\[\phi:A\longrightarrow\Ll_B(E),\]
and an odd, regular, densely defined self-adjoint operator $D$ on $E$,
such that
\begin{enumerate}
\item[(i)]
$\phi(a)(D+\imath)^{-1}\in\KM_B(E)$ for every $a\in A$;
\item[(ii)]
there is a dense $*$-subalgebra $\mathcal{A} \subset A$ for such that $\phi(a)$ maps $\Dom(D)$ to itself and the graded commutator $[D,\phi(a)]$ extends to an operator in
$\Ll_B(E)$ for every $a\in\mathcal{A}$.
\end{enumerate}
\end{definition}

Such a cycle defines a bounded $KK$-cycle and therefore a class in $KK(A,B)$ by applying a normalizing function: If $\Theta:\RM\to[-1,1]$ is an increasing odd continuous function with
\[\lim_{t \to \pm \infty} \Theta(t)=\pm 1,\]
then the $KK(A,B)$-class of $(E,\phi,\Theta(D))$ is independent of $\Theta$. 

For an ungraded $C^*$-algebra $A$, its K-homology groups $K^d(A)=KK_d(A, \CM)$  are defined on Hilbert-$\CM_d$-modules. In our case, we only need to deal with Hilbert modules of the form $\Hh\hat{\otimes}\CM_d$ for $\Hh$ an ungraded separable Hilbert space. 
As a vector space, this is equal to $\Hh\otimes \CM_d$ but with a $\CM_d$-valued inner product, a canonical left and right multiplication by $\CM_d$, and a grading operator inherited from the grading of $\CM_d$. 

\begin{definition}
\label{def:index}
For $i=d\bmod 2$, the index pairing $\langle \cdot , \cdot\rangle: K_{i}(A)\times KK_d(A , \CM) \to \ZM $ is defined by the Kasparov product
\[\langle x, y \rangle = \kappa_i^A(x) \otimes y\]
for fixed isomorphisms $\kappa_i^A: K_{i}(A)\to KK_i(\CM, A)$ and $KK(\CM,\CM_{d+i})\simeq KK(\CM,\CM)\simeq \ZM$.
\end{definition}
Although the product itself is technically computed on $KK^{\mathrm{sep}}$, since we have the isomorphisms $KK(A,\CM)\simeq KK^{\mathrm{sep}}(A,\CM)$ \cite{Skandalis1985ExtensionsSemifiniteFactor} and $KK(\CM, A)\simeq KK^{\mathrm{sep}}(\CM, A)$ by separability of the first factor, we can drop the superscript for the K-theory and K-homology groups. 
The isomorphisms in the index pairing are only canonical up to choices of orientations; we make concrete choices in Appendix~\ref{sec:index}.

\subsection{Dirac-Type Functions}
\label{sec:dirac_ops}
Let $\Xx$ be a proper metric space with bounded geometry and let $\Hh_\Xx$ be a geometric module.
\begin{definition}
\label{def:edge_dirac_type}
Let $\varnothing \neq \Zz\subset \Xx$ be a closed subset.
A Borel function $f:\Xx \to \RM^d$ is of $\Zz$-Dirac-type if 
\begin{enumerate}
   \item[(i)] $f$ is large-scale Lipschitz-continuous, i.e. there is a constant $L$ such that
   \[\norm{f(x)-f(y)}\leq L(1+ d(x,y))\]
   for all $x,y \in \Xx$.
    \item[(ii)] there exists a positive function $\rho\in C_0(\Xx)$ such that \[(1+\abs{f}^2)^{-\frac{1}{2}} (1+ d_\Zz^2)^{-\frac12} \leq  \rho,\]
    with $d_\Zz(x)=\mathrm{dist}(x, \Zz)$.
\end{enumerate}
We say that $f$ is $\Zz$-polynomially proper if one has following the stronger condition.
\begin{enumerate}
    \item[(ii')] For some $x_*\in \Xx$ and $c,C,s>0$ one has $$|f(x)|^2 + d^2_\Zz(x) \geq  C d(x_*,  x)^{2s}$$
for all $x\in \Xx$ with $d(x_*,x)>c$.
\end{enumerate}
We define the associated Dirac operator \[D_f:= \sum_{i=1}^d f_i\hat{\otimes} \sigma_i,\]
with $f_i = M_{f_i}$ identified with an (unbounded) multiplication operator.

\end{definition}

The bulk case is recovered by taking $\Zz=\Xx$ and we will refer to such functions as \emph{(polynomially proper) Dirac-type functions}.
\begin{remark}
{\rm Note that $\Zz$-properness implies that, for any compact set $K$, the preimage of any compact set $f^{-1}(K)$ intersected with any thickening $(\Zz)_R$ of $\Zz$ is bounded. 
If $f$ is continuous, then the notion of $\Xx$-properness is therefore the same as the topological notion, i.e. preimages of compact sets are compact.} $\diamond$
\end{remark}
\begin{lemma}
\label{lemma:dirac_op}
If $f:\Xx\to\RM^d$ is of $\Zz$-Dirac-type, then
\[
 (\Hh_{\Xx} \hat{\otimes} \CM_d,\pi,D_f)
\]
defines an unbounded $KK_d(C^*_u(\Zz\subset \Xx, \Hh_\Xx),\CM)$-cycle where $\pi: a\mapsto a\hat{\otimes}\one_{\CM_d}$. We denote its class by $[D_f]$. 
\end{lemma}
\begin{proof}
We check conditions (i) and (ii) of Definition~\ref{def:unbddkk}. Since $f_i$ are self-adjoint scalar functions and the Clifford generators anticommute, $D_f$ is self-adjoint with 
\[\Dom(D_f) = \bigcap_{i=1}^d \Dom( {f_i} \hat{\otimes} \sigma_i)=\left(\bigcap_{i=1}^d \Dom( {f_i})\right) \hat{\otimes} \CM_d,\]
and regularity follows from $D_f^2 = \sum_{i=1}^d f_i^2 \hat{\otimes} \one$.
The algebra $C^*_{u}(\Zz\subset \Xx, \Hh_\Xx)$ has the dense subalgebra $\CM_{u}(\Zz\subset \Xx, \Hh_\Xx)$ consisting of elements which are supported in a thickening of $\Zz$. By Lemma~\ref{lemma:lipschitz}, the commutators $[D_f,\pi(a)]$ are bounded for any $a\in \CM_u(\Zz\subset\Xx, \Hh_\Xx)$ and the domain of $D_f$ is preserved.
Moreover, as $a$ is supported near $\Zz$, one can write $a= g \cdot a$ for some bounded function $g$ supported in a thickening of $\Zz$. Therefore, we can express
\[a = (1+d_\Zz^{2})^{-\frac12} \left( (1+d^2_\Zz)^\frac12 g\right) a\]
with a product of two bounded Borel functions. To see local compactness of the resolvent we factorize
\begin{equation}
\label{eq:resolvent_with_boundary_stuff_factored}
(1+D_f^2)^{-\frac12}a= (1+\abs{f}^2)^{-\frac12} a = \left((1+\abs{f}^2)^{-\frac12}(1+d_\Zz^{2})^{-\frac12}\frac{1}{\rho}\right) \left( (1+d^2_\Zz)^\frac12 g\right) (\rho a)
\end{equation}
which is a compact operator since $\rho a$ is a compact operator on $\Hh_\Xx$ ($a$ is locally compact) and the other factors are bounded operators.
\end{proof}

It seems clear from the definition of KK-theory that the class $[D_f]$ depends on $f$ only up to homotopy. For further use, let us state precise sufficient conditions:

\begin{lemma}
\label{lemma:homotopy}
Let $f: [0,1] \times \Xx \to \RM^d$ be a path of Dirac-type functions such that for each $t\in [0,1]$, the slice $f^{t}: \Xx \to \RM^d$ is a $\Zz$-Dirac-type function as in Definition~\ref{def:edge_dirac_type} with $L$ and $\rho$ independent of $t$. If $t\mapsto f^t\rvert_{K}$ is continuous in sup-norm for every compact set $K\subset \Xx$, then the class $[D_{f^{t}}]$ in $KK_{d}(C^*_u(\Zz\subset \Xx, \Hh_\Xx), \CM)$ does not depend on $t$.
\end{lemma}
\begin{proof}
Denote the Hilbert module by $E\coloneqq \Hh_\Xx\hat{\otimes} \CM_d$ for brevity. Define an unbounded operator $\tilde{D}$ on the Hilbert module $C([0,1], E)$ via pointwise multiplication by $\tilde{D}(t)=D_{f^{t}}$. We show that this is an unbounded $KK_{d}(C^*_{u}(\Zz\subset \Xx, \Hh_\Xx),C([0,1]))$-cycle, i.e. a homotopy of unbounded KK-cycles. Note that by local boundedness of $f^t$ and properness of $\Xx$, all $\tilde{D}(t)$ have the common core $E_c=C_c(\Xx)\cdot E$. In addition, since the maps $t\mapsto f^{t}$ become norm-continuous when restricted to any compact subset of $\Xx$, the map $t \mapsto \tilde{D}(t)\psi$ is continuous for all $\psi \in E_c$. By \cite[Lemma~1.15]{vanDenDungenMesland2020}, $\tilde{D}$ is therefore a regular self-adjoint operator on the Hilbert $C([0,1])$-module $C([0,1], E)$. Moreover, the proof of the cited lemma shows that the algebraic tensor product $C([0,1])\otimes_{\mathrm{alg}}E_c$ is a core for $\tilde{D}$.

It remains to show that for every $a \in \CM_u(\Zz\subset \Xx)$ one has
\begin{equation}\label{eqn:commutator_homotopy}
[\tilde{D}, a]\in \mathcal{L}(C([0,1], E))
\end{equation}
and that $\tilde{D}$ has relatively compact resolvent, i.e.,
\begin{equation}\label{eqn:relative_compact_resolvent}
(\tilde{D}+\imath)^{-1} a \in C([0,1], \KM(E)).
\end{equation}
Since $a$ preserves compact support, it maps $C([0,1])\otimes_{\mathrm{alg}}E_c$ into itself, and by Lemma~\ref{lemma:lipschitz} the operator norm of the commutator $[\tilde{D}(t),a]$ is bounded uniformly in $t$. This gives \eqref{eqn:commutator_homotopy}.

For \eqref{eqn:relative_compact_resolvent}, we can factor again as in \eqref{eq:resolvent_with_boundary_stuff_factored} for any $a\in \CM_u(\Zz\subset\Xx,\Hh_\Xx)$, which shows that each fiber $(\tilde{D}(t)+\imath)^{-1} a$ is compact. By an $\frac{\epsilon}{3}$-argument, the locally uniform continuity implies norm-continuity of $t\mapsto (\tilde{D}(t)+\imath)^{-1}g$ for any function $g$ supported in a thickening of $\Zz$. We conclude that $(\tilde{D}(t)+\imath)^{-1} a$ is a norm-continuous path of compact operators.
\end{proof}

\subsection{Half-Spaces and Dirac-Type Functions}
As explained in the introduction, the main motivation for Dirac-type functions is to study Dirac operators defined by signed distance functions. 
The goal of this subsection is to show that the signed distance functions on transverse sets are equivalent to Dirac-type functions under mild assumptions.
We begin with our definition of transversality:
\begin{definition}
\label{def:half-space_dirac}  
The half-spaces $\Yy_1,\dots, \Yy_d\subset \Xx$ are called coarsely transverse relative to some closed set $\Zz\subset \Xx$ if
$$(1+\delta_{\Yy_1}^2+\cdots+\delta_{\Yy_d}^2+d_\Zz^2)^{-1}\in C_0(\Xx)$$
for the signed distance functions
\[\delta_{\Yy_i}(x) =\mathrm{dist}(x, \Xx \setminus \Yy_i)- \mathrm{dist}(x, \Yy_i).\]
They define a  $\Zz$-Dirac-type function $x\in \Xx\mapsto (\delta_{\Yy_1}(x), \dots, \delta_{\Yy_d}(x)) \in \RM^d$ which in turn defines the Dirac operator 
\[D_{\Yy_1,\dots, \Yy_d} = D_{(\delta_{\Yy_1},\ldots,\delta_{\Yy_d})}.\]

We call the half-spaces polynomially coarsely transverse relative to $\Zz$ if $(\delta_{\Yy_1},\ldots,\delta_{\Yy_d})$ is $\Zz$-polynomially proper.
\end{definition}
If $\Zz=\Xx$ we simply speak of (polynomially) coarsely transverse half-spaces.
\begin{remark}
{\rm 
Our condition can be shown to be equivalent to the one given in \cite{LudewigThiang2023}: half-spaces are coarsely transverse if and only if
$$(\Zz)_R\cap (\Yy_1)_R \cap (\Yy^c_1)_R \cap \dots \cap (\Yy_d^c)_R \cap (\Yy_d)_R = \{x\in \Xx:|\delta_{\Yy_1}(x)|,\ldots,|\delta_{\Yy_d}(x)|,d_\Zz(x) < R\}$$
is a bounded set for every $R$, and they are polynomially coarsely transverse if and only if the diameters of these sets grow at most polynomially in $R$.}$\diamond$
\end{remark}

We next show that, under a mild nondegeneracy assumption, every Dirac-type class can be represented by signed distances to half-spaces:

\begin{proposition}
Let $\Zz \subset \Xx$ be closed and $f$ a (polynomially proper) $\Zz$-Dirac-type function. Suppose the sets $\Yy_i=\{x\in \Xx: f_i(x)\geq 0\}$ are half-spaces and let $\delta=(\delta_{\Yy_1},\ldots,\delta_{\Yy_d})$ be the corresponding signed distance functions. Then $\Yy_1,\dots, \Yy_d$ are (polynomially) coarsely transverse relative to $\Zz$, $\delta$ is a  (polynomially proper) $\Zz$-Dirac-type function, and $[D_f]=[D_\delta]$.
\end{proposition}
\begin{proof}
We show that $\delta$ is (polynomially) proper and hence that $\Yy_1,\dots, \Yy_d$ are (polynomially) coarsely transverse relative to $\Zz$. By the large-scale Lipschitz property one has
$$|f_i(x)|\leq L(1+|\delta_{\Yy_i}(x)|).$$
Indeed, if $x \in \Yy_i^c$, then for every $y\in \Yy_i$ one has $-f_i(x)\leq f_i(y)-f_i(x)\leq L(1+d(x,y))$, and if $x \in \Yy_i$ then for any $y \in \Yy_i^c$, $f_i(x)\leq f_i(x)-f_i(y)\leq L(1+d(x,y))$.
In either case, taking the infimum over $y$ gives the desired bound.

If $f$ is proper with some function $\rho\in C_0(\Xx)$ then 
\[(1+|\delta|^2)^{-1/2}(1+d_{\Zz}^2)^{-1/2}\leq C_{c,d,L}(1+|f|^2)^{-1/2}(1+d_{\Zz}^2)^{-1/2}\leq c_1 \rho\]
with some constant depending on $c,d$ and $L$. The polynomial properness is also immediate from 
\[Cd(x_*,x)^{2s} \leq |f(x)|^2 + |d_{\Zz}(x)|^2 \leq c_2 L^2 (1+|\delta(x)|^2) + |d_{\Zz}(x)|^2\]
as one can divide through by $c_2L^2$, choose $c$ so large that the additive constant $1$ is negligible compared to $d(x_*,x)^{2s}$, and enlarge the constants on the right-hand-side if necessary.
Note in addition that if $(1+|\delta|^2)^{-1/2}(1+d_\Zz^2)^{-1/2} \in C_0(\Xx)$, then $(1+|\delta|^2 + d_\Zz^2)^{-1} \in C_0(\Xx)$.

For the last statement, we can then consider the path
$$f^t(x)= t f(x) + (1-t)\delta(x),$$
which is continuous in the local supremum norm in $t$.
As $f_i$ and $\delta_{\Yy_i}$ have the same sign, we compute
$$f_i^t(x)^2 = t^2 f^2_i(x) +2 t(1-t) f_i(x)\delta_{\Yy_i}(x) + (1-t)^2\delta_{\Yy_i}(x)^2\geq t^2 f_i^2(x) + (1-t)^2\delta^2_{\Yy_i}(x)$$
which implies the path is uniformly proper with the same $\rho$ by a similar argument as above. 
The path is also uniformly large-scale Lipschitz, hence the conditions of Lemma~\ref{lemma:homotopy} are satisfied.

\end{proof}

\section{Bulk-edge correspondence}\label{sec:bulk_edge_correspondence}

In this section, we again let $\Xx$ be a proper metric space with bounded geometry and fix a geometric module $\Hh_\Xx$. As explained in the introduction, the K-theoretic side of bulk-edge correspondence is expressed through the Mayer-Vietoris boundary map which maps the K-theory groups of the uniform Roe $C^*$-algebra $C^*_u(\Xx)$, the \textit{bulk} algebra, to those of a relative Roe algebra $C^*_u(\Zz_d \subset \Xx)$, the \textit{edge} algebra. We first realize this map through an exact sequence of $C^*$-algebras. For notational convenience we will suppress the module $\Hh_\Xx$ and simply write $C^*_u(\Xx)=C^*_u(\Xx, \Hh_\Xx)$, but it is understood that we consider the Roe $C^*$-algebras on $\Hh_\Xx$ throughout.

\begin{definition}
\label{def:pullback}
Let $\Yy_d\subset \Xx$ be a half-space with coarse boundary $\Zz_d$.
We define the $C^*$-algebra $E_{\Yy_d}$ as the pullback of the following diagram:
\begin{equation*}
	\begin{tikzcd}
		 E_{\Yy_d} \arrow[r] \arrow[d] & C^*_u(\Xx) \arrow[d] \\
		 C^*_u(\Yy_d \subset \Xx)  \arrow[r] & C^*_u(\Xx)/C^*_u(\Yy_d^c \subset \Xx)
	\end{tikzcd}
\end{equation*}
More concretely, we have an explicit description of the pullback as
$$E_{\Yy_d} = \{(\hat{a},a) \in C^*_u(\Yy_d \subset \Xx) \oplus C^*_u(\Xx) \hspace{0.1cm} \vert \hspace{0.1cm} \hat{a}-a \in C^*_u(\Yy_d^c \subset \Xx)\}.$$
\end{definition}
An element of the $C^*$-algebra $E_{\Yy_d}$ consists of a pair of operators $(\hat{a},a)$ on $\Hh_\Xx$ with $\hat{a}$ supported near $\Yy_d$ and $\hat{a}-a$ near $\Yy_d^c$. Far away from $\Yy_d^c$, the two elements $\hat{a}$ and $a$ must therefore coincide asymptotically. 

\begin{proposition}
We have an exact sequence of $C^*$-algebras
\begin{equation}
\label{eq:bulk_edge_exact}
	\begin{tikzcd}
		0 \arrow[r] & C^*_u(\Zz_d\subset \Xx) \arrow["\iota", r]  & E_{\Yy_d} \arrow["q", r]  & C^*_u(\Xx)   \arrow[r] & 0 
    \end{tikzcd}
\end{equation}
with $\iota(b)=(b,0), q(\hat{a},a)=a$, and the boundary map $\pdv_{\Yy_d}: K_i(C^*_u(\Xx)) \to K_{1-i}(C^*_u(\Zz_d\subset \Xx))$ coincides with the Mayer-Vietoris boundary map obtained from the decomposition $\Xx=\Yy_d \cup \Yy_d^c$.
\end{proposition}
\begin{proof}
By \cite[Lemma 2]{HigsonRoeYu1993} which applies to uniform Roe algebras,
\[C^*_u(\Xx)=C^*_u(\Yy_d\subset\Xx)+C^*_u(\Yy_d^c\subset\Xx), \qquad C_u^*(\Zz_d\subset \Xx)= C^*_u(\Yy_d\subset \Xx)\cap C^*_u(\Yy^c_d\subset \Xx).\]
Surjectivity of $q$ is immediate from the former equality, and exactness from the latter: An element $(\hat{a},a)$ is in the kernel of $q$ if and only if $a=0$, hence $\hat{a} \in C^*_u(\Yy^c_d\subset \Xx)$. Together with $\hat{a}\in C^*_u(\Yy_d\subset \Xx)$ this implies $\hat{a}\in C_u^*(\Zz_d\subset \Xx)$. As shown in \cite{EwertMeyer}, the Mayer-Vietoris boundary map is given by the composition of the map $K_i(C^*_u(\Xx))\to K_i(C^*_u(\Xx)/C^*_u(\Yy^c_d\subset \Xx))$ with the boundary map of the exact sequence
\begin{equation*}
	\begin{tikzcd}
		0 \arrow[r] & C^*_u(\Zz_d\subset \Xx) \arrow[r]  &  C^*_u(\Yy_d\subset \Xx) \arrow[r]  & C^*_u(\Xx)/C^*_u(\Yy^c_d\subset \Xx)   \arrow[r] & 0 
    \end{tikzcd}
\end{equation*}
and the result follows by definition of $E_{\Yy_d}$ and naturality of the boundary maps.
\end{proof} 

The topological invariants of bulk and edge K-theory classes are given by pairings with suitable K-homology classes. We assume that both are derived from a Dirac-type function whose $d$-th component behaves roughly like the signed distance function to the half-space $\Yy_d$:

\begin{definition}\label{def:dirac_type_adapted}
Let $\Yy_d\subset \Xx$ be a half-space with coarse boundary $\Zz_d$. We say that a Dirac-type function $f = (\underline{f},f_d):\Xx\to \RM^d$ is adapted to $\Yy_d$ if
\begin{enumerate}
    \item[(i)] $\underline{f}$ is of $\Zz_d$-Dirac-type.
    \item[(ii)] For every $R>0$ there exists some $S>0$ such that
    $$(\Zz_d)_R\subset \{|f_d|\leq S\}, \qquad \{|f_d|\leq R\}\subset (\Zz_d)_S $$

    \item[(iii)] $\{f_d \geq 0\}$ and $\Yy_d$ have finite Hausdorff distance from each other.
\end{enumerate}
\end{definition}
By Lemma~\ref{lemma:dirac_op}, these functions define the classes $$[D_f]=[\Hh_\Xx \hat{\otimes} \CM_d, \pi, D_{f}]\in KK_{d}(C^*_u(\Xx), \CM)$$ and $$[D_{\underline{f}}]=[\Hh_\Xx \hat{\otimes} \CM_{d-1}, \pi, D_{\underline{f}}]\in KK_{d-1}(C^*_u(\Zz_d\subset \Xx),\CM)$$ for bulk and edge, respectively. Our main result is a more general form of Theorem~\ref{th:bulk_edge_intro}:
\begin{theorem}
\label{th:bulk_edge}
Let $\Yy_d\subset \Xx$ be a half-space with coarse boundary $\Zz_d$, and $f:\Xx\to \RM^d$ be a $\Yy_d$-adapted Dirac-type function.
For any class $x\in K_j(C^*_u(\Xx))$ with $j=d\bmod 2$ we have
\[\langle \pdv_{\Yy_d}(x), [D_{\underline{f}}]\rangle = (-1)^{d}\langle x, [D_{f}]\rangle.\]
\end{theorem}
The proof of this theorem will span subsections~\ref{ssec:extensionclass} and \ref{ssec:kkproduct}. The first step is to express the boundary map through an extension class in $KK_1^{\mathrm{sep}}$. As these algebras need not be separable, we will need to approximate all algebras appearing in the short exact sequence \eqref{eq:bulk_edge_exact} by directed nets of separable subalgebras. We will treat this construction abstractly.

\subsection{Extension classes in KK-theory}
\label{ssec:extensionclass}
A short exact sequence
$$0\;\rightarrow\;B\xrightarrow{\iota}\;E\xrightarrow{q}\;A\;\rightarrow0$$
is called semisplit if there exists a completely positive contractive section $\sigma:A\to E$ with $q\circ \sigma=\mathrm{id}_A$.

Let $M(B)$ denote the multiplier algebra of $B$. The Busby invariant is the unique $*$-homomorphism $\tau: A\to M(B)/B$ which fits into the commutative diagram
\begin{equation*}
	\begin{tikzcd}
		0 \arrow[r] & B \arrow[r] \arrow[d] & E \arrow[r] \arrow[d] & A \arrow[d, "\tau"]  \arrow[r] & 0 \\
		 0 \arrow[r]  & B \arrow[r] &  M(B) \arrow[r] & M(B)/B \arrow[r] & 0
	\end{tikzcd}
\end{equation*}

Let us recall how one associates a class in $KK_1(A,B)$ to a semisplit extension:

\begin{proposition}\label{prop:check_extension}
Suppose that for a semisplit short exact sequence of separable, ungraded $C^*$-algebras as above, one can write the Busby invariant $\tau: A \to M(B)/B \cong \Qq(B)$ in the form
$$\tau(a)= P \pi(a) P + B$$
with some $\ast$-homomorphism $\pi: A \to M(B)$ and an operator $P\in M(B)$ which satisfies 
\begin{equation*} 
\{[P,\pi(a)],\, (P^*-P)\pi(a),\, (P^2-P)\pi(a)\} \subset \KM_B(B).
\end{equation*}
We associate to the extension the class $[{\Ext}_E]\in KK_1(A,B)$ represented by the $KK(A, B\hat{\otimes}\CM_1)$-cycle $(B\hat{\otimes}\CM_1, \pi \hat{\otimes} \one_{\CM_1}, F)$  with $F=(2P-1)\hat{\otimes}\ep$. It is compatible with the K-theoretic boundary maps $\pdv_E:K_j(A)\to K_{1-j}(B)$ in the sense that 
\begin{equation}
\label{eq:bulk_edge_kk}
\kappa_{1-j}^B(\pdv_{E}(x)) = (-1)^{j}\kappa^A_j(x) \otimes [{\Ext}_{E}] \in KK_{1-j}(\CM, B), \qquad \forall x\in K_j(A)
\end{equation}
under the isomorphisms $\kappa_i^D: K_i(D)\to KK_i(\CM, D)$.

\end{proposition}
\begin{proof}
The definition of the extension class is standard (see e.g., \cite[17.6.4]{blackadar1998k}). For the boundary maps, we use the conventions of \cite{RordamLarsenLaustsen2000}.
Comparing them to Proposition~\ref{prop:basic_products}, these coincide for the map $K_0(A)\to K_1(B)$, while they give the opposite sign for the map $K_1(A)\to K_0(B)$. 
\end{proof}

For each separable approximation of an exact sequence, the boundary maps in K-theory are therefore given by Kasparov products by an extension class. If those extension classes form a compatible family then we get the following:
\begin{proposition}
\label{prop:ext_boundarymap}
Assume that for the semisplit extension $0 \to B \to E \to A \to 0$ of $C^*$-algebras there exist directed families $(A_i)_{i\in I}$, $(B_i)_{i\in I}$, $(E_i)_{i\in I}$ of separable subalgebras indexed by the same directed set $i\in I$ such that
\begin{enumerate}
    \item[(i)] the nets are cofinal: for every pair $(A',B')$ of separable subalgebras $A'\subset A$, $B'\subset B$ there is some $i\in I$ with $A'\subset A_i$ and $B'\subset B_i$, and
    \item[(ii)] for each $i\prec j$ we have the commutative diagram 
\begin{equation}
\label{eq:comm_diag_i_to_j}
	\begin{tikzcd}
		0 \arrow[r] & B_i \arrow[r] \arrow[d, "\beta_{i,j}"] & E_i \arrow[r] \arrow[d, "\epsilon_{i,j}"] & A_i \arrow[bend right=33, "\sigma_i"]{l} \arrow[d,"\alpha_{i,j}"]  \arrow[r] & 0 \\
		 0 \arrow[r]  & B_j \arrow[r] & E_j \arrow[r] & A_j \arrow[bend right=33, "\sigma_j"]{l} \arrow[r] & 0
	\end{tikzcd}
\end{equation}
with exact semisplit rows and $\alpha_{i,j}, \beta_{i,j}$, and $\epsilon_{i,j}$ the inclusion maps.
\end{enumerate}  
Then there exists a class $[{\Ext}_{E}]_{\mathrm{sep}}\in  KK^{\rm sep}_1(A,B)$ such that
$$\kappa_{1-l}^B(\pdv_{E}(x)) = (-1)^{l}\kappa_l^A(x)\otimes [{\Ext}_{E}]_{\mathrm{sep}} \in KK_{1-l}(\CM, B)$$
for every $x\in K_l(A)$ and where $\pdv_{E}: K_l(A)\to K_{1-l}(B)$ is the boundary map.
\end{proposition}
\begin{proof}
The diagram \eqref{eq:comm_diag_i_to_j} corresponds to a morphism of semisplit extensions. By naturality,
\begin{equation}
\label{eq:extension_compatibility}
(\beta_{i,j})_*[{\Ext}_{E_i}] =(\alpha_{i,j})^*[{\Ext}_{E_j}]
\end{equation}
as classes in $KK_1(A_i,B_j)$. Define $e_i= (b_i)_*[{\Ext}_{E_i}]$ for $b_i: B_i\to B$. By (i), for any $A'\subset A$ we can find some $i\in I$ with $A'\subset A_i$ and set $$e_{A'}= (A' \to A_i)^*(e_i) \in KK_1(A', B)$$
as a pullback along the inclusion. By  \eqref{eq:extension_compatibility} and passing to a common upper bound of indices, $e_{A'}$ does not depend on the choice of $i$ as long as $A'\subset A_i$ and the classes $e_{A'}$ form a compatible family. The extension class $$[{\Ext}_{E}]_{\mathrm{sep}}:= (e_{A'})_{A'\subset A}$$ is therefore well-defined.

By continuity of $K$-theory, any $x \in K_l(A)$ can be written as the image of some $x_i\in K_l(A_i)$ for some $i$. Using first the definition of the product in $KK^{\mathrm{sep}}$, then \eqref{eq:bulk_edge_kk}, and finally naturality of the identifications $\kappa$, we have
\begin{align*}
\kappa_l^A(x)\otimes_{A} [{\Ext}_E]_{\mathrm{sep}} &= (b_i)_*( \kappa_l^{A_{i}}(x_i)\otimes_{A_i} [{\Ext}_{E_i}])\\
&=(-1)^{l} (b_i)_*(\kappa_{1-l}^{B_i}(\pdv_{E_i}(x_i)))=(-1)^{l} \kappa_{1-l}^B((b_i)_*(\pdv_{E_i}(x_i)))
\end{align*}
as elements of $KK_{1-l}^{\mathrm{sep}}(\CM, B)=KK_{1-l}(\CM, B)$. Naturality of the boundary maps $(b_i)_*\circ \pdv_{E_i}=\pdv_E \circ (a_i)_*$ for $a_i:A_i\to A$ finishes the proof.
\end{proof}

\subsection{Proof of the Main Result}
\label{ssec:kkproduct}

For compatibility with the previous section let us abbreviate the algebras in the bulk-edge exact sequence \eqref{eq:bulk_edge_exact} as follows:
$$B := C^*_u (\Zz_d \subset \Xx), \qquad E:= E_{\Yy_d}, \qquad A := C^*_u(\Xx).$$

The main technical result needed for the proof of Theorem~\ref{th:bulk_edge} is the computation of the dual of the boundary map, which describes how the extension class acts on K-homology classes:
\begin{proposition}
\label{prop:bulkedge_kk}
There exists a directed set $I$ of subextensions as in Proposition~\ref{prop:ext_boundarymap} such that the pullbacks $[D^{i}_{f}]\in KK_{d}(A_i,\CM)$ and $[D^{i}_{\underline{f}}]\in KK_{d-1}(B_i,\CM)$ along the inclusion maps $a_i: A_i \to A$ and $b_i: B_i \to B$ satisfy
\begin{equation}
\label{eq:kk_product_ext_i}
[{\Ext}_{E_i}] \,\otimes\, [D_{\underline{f}}^{i}] =[D_{f}^{i}]
\end{equation}
as an equality in $KK_{d}(A_i, \CM).$
\end{proposition}

Before proving this, let us first show how this completes the proof of the bulk-edge correspondence:
\begin{proof} (of Theorem~\ref{th:bulk_edge}) 
Since the nets $A_i\to A$ and $B_i\to B$ are cofinal, Proposition~\ref{prop:bulkedge_kk} implies
$$[{\Ext}_E]_{\mathrm{sep}} \otimes [D_{\underline{f}}] = [D_f]$$
as classes in $KK^{\mathrm{sep}}$ since the Kasparov product is uniquely determined by \eqref{eq:kk_product_ext_i} and naturality, following a similar argument to the proof of Proposition~\ref{prop:ext_boundarymap}.

For $\tilde{x}=\kappa_j^A(x)$ and $\tilde{y}=\kappa_{1-j}^B(\pdv_{E}(x))$ we therefore have
\begin{align*}
\langle \pdv_{E}(x), [D_{\underline{f}}]\rangle = \tilde{y} \otimes [D_{\underline{f}}] &= (-1)^{j} (\tilde{x} \otimes [{\Ext}_E]_{\mathrm{sep}}) \otimes [D_{\underline{f}}]\\
&=(-1)^{j}\tilde{x} \otimes ([{\Ext}_E]_{\mathrm{sep}} \otimes [D_{\underline{f}}]) =(-1)^{j}\langle x, [{\Ext}_E]_{\mathrm{sep}}\otimes [D_{\underline{f}}]\rangle
\end{align*}
where the first and last equality hold by Definition~\ref{def:index} and Proposition~\ref{prop:ext_boundarymap} respectively, whereas the middle equality is just associativity of the product.

\end{proof}
We now proceed with the proof of Proposition~\ref{prop:bulkedge_kk}. The first step is to construct the net of subextensions. Observe that the extension classes in KK-theory do not depend on the semisplitting at all, hence we are free to choose one which is adapted to the function $f_d$. Define a smoothed version of the half-space projection as the multiplication operator
$$\Pp_d := \rho(f_d)$$
with $\rho: \RM\to \RM$ a smooth increasing function such that $\supp(\rho')$ is compact, $\lim_{\lambda\to \infty}\rho(\lambda)=1$, $\lim_{\lambda\to -\infty}\rho(\lambda)=0$ and $\rho(-\lambda)=1-\rho(\lambda)$. 
By Definition~\ref{def:dirac_type_adapted}, the completely positive contraction
\[\sigma: a \mapsto (\Pp_d a \Pp_d,a)\]
is a semisplitting for the exact sequence: For finite-propagation $a\in \CM_u^*(\Xx)$ one clearly has $\Pp_d a \Pp_d\in C^*_u(\Yy_d\subset \Xx)$ and the two terms of $\Pp_d a \Pp_d-a= \Pp_d[a,\Pp_d] + (1-\Pp_d^2)a$ are supported near $\Zz_d$ and $\Yy_d^c$ respectively.

The following elementary fact will be important:
\begin{lemma}
\label{lem:Mfd_multiplier}
For $A,B$ as above, the map $f_d$ acting via left multiplication $b \mapsto {f_d}\, b$ yields a densely defined regular self-adjoint operator such that $A({f_d}+\imath)^{-1}\subset B$.
\end{lemma}
\begin{proof}
Observe that ${f_d}$ is already self-adjoint as a densely defined Hilbert-space operator on $\Hh_\Xx$.
Hence, it is enough to show that $({f_d}+\imath)^{-1}B\subset B$ is norm-dense \cite[Lemma 2.6]{KSt}.

For any $\epsilon>0$ we can approximate any $b\in B$ by some $\tilde{b}\in \CM_u(\Zz_d\subset \Xx)$ supported in some $R$-thickening of $\Zz_d$ such that $\|b-\tilde{b}\|<\epsilon$. One can pick a continuous cutoff function $g:\Xx\to \RM$ which is equal to one on $(\Zz_d)_R$ and zero outside $(\Zz_d)_{2R}$. Note now that $g f_d$ is a bounded function by Definition~\ref{def:dirac_type_adapted}(ii). We therefore get
$$\left(1+ \frac{\imath}{n}{f_d}\right)^{-1} \tilde{b} - \left(1+ \frac{\imath}{n}{g f_d}\right)^{-1} \tilde{b} \longrightarrow 0$$
in operator norm, with uniform convergence giving $\left(1+ \frac{\imath}{n}{g f_d}\right)^{-1} \tilde{b} \longrightarrow \tilde{b}$. By an $\frac{\epsilon}{2}$-argument we have
$$\lim_{n\to \infty} \left(1+ \frac{\imath}{n}{f_d}\right)^{-1} b = b$$
which implies $(f_d + \imath)^{-1}B \subset B$ is norm-dense. Similarly one has $a({f_d}+\imath)^{-1}\subset B$ for every $a\in \CM_u^*(\Xx)$, as $a g(f_d)$ is supported near $\Zz_d$ for every $g\in C_c(\RM)$ and $({f_d}+\imath)^{-1}$ can be approximated in norm by such functions.
\end{proof}

We now construct the net for Proposition~\ref{prop:ext_boundarymap}. Note that it needs to be large enough to define the unbounded representatives for the Kasparov product.

\begin{proposition}\label{prop:sep_approx}
We write the short exact sequence in the form
\[\begin{tikzcd}
	0 & B & E & A & 0
	\arrow[from=1-1, to=1-2]
	\arrow["\iota", from=1-2, to=1-3]
	\arrow["q", from=1-3, to=1-4]
	\arrow["\sigma"', shift right=3, curve={height=6pt}, from=1-4, to=1-3]
	\arrow[from=1-4, to=1-5]
\end{tikzcd}\]
and recall that $B$ is a subalgebra of $A$.

There exists a directed set $I$ and separable algebras $(A_i)_{i\in I}$, $(B_i)_{i\in I}$, $(E_i)_{i\in I}$ as in Proposition~\ref{prop:ext_boundarymap} with the following properties:
\begin{enumerate}
\item[(i)] Identify $f_k$ with the multiplication operators by $f_k$ on $\Hh_\Xx$, $k=1,\ldots,d$. Then both $A_i$ and $B_i$ are invariant under $\RM^d$-actions \begin{equation}
\label{eq:r_actions}
\alpha: (a,t) \mapsto e^{-\imath f\cdot t} a e^{\imath f\cdot t}.
\end{equation}
\item[(ii)] Let $\Cc:= C_0(f_d)$ be the algebra generated by the $C_0$-functional calculus of $f_d$. We have $\Cc B_i \subset B_i$, $A_i \Cc \subset B_i$ and $A_i B_i \subset B_i$. Moreover, $\Cc B_i$ is norm-dense in $B_i$.
\end{enumerate}
\end{proposition}

\begin{proof}
As the directed set $I$ we take the set of pairs $(A',B')$ of separable subalgebras of $A$ and $B$ respectively with $(A',B') \preceq (A'', B'')$ if $A'\subseteq A''$ and $B'\subseteq B''$. 

Fix any pair $i=(A',B') \in I$. The left action by $\Cc$ and the $\RM^d$-action $\alpha$ are norm-continuous and commute, hence the $C^*$-algebraic span
$$A_i= C^*\left(\{(\lambda \one+c)\,\alpha_{t}(a): a\in A', \lambda\in \CM, c\in \Cc, t\in \mathbb{R}^d\}\right),$$
is a separable $C^*$-algebra which is $\alpha$-invariant and closed under left and right multiplication by $\Cc$. By the same argument one can construct a separable algebra $\hat{B}_i$ as the smallest $C^*$-algebra which contains $B'$, $A_i \Cc$ and is invariant under $\alpha$ and $\Cc$. Define then $E_i=C^*(\sigma(A_i), \iota(\hat{B}_i)) \subset E$, which is also separable.

We set $B_i = \iota^{-1}(\ker(q\rvert_{E_i}))$ to obtain the exact sequence
\[\begin{tikzcd}
	0 & {B_i} & {E_i} & {A_i} & 0
	\arrow[from=1-1, to=1-2]
	\arrow["{\iota}", from=1-2, to=1-3]
	\arrow["{q}", from=1-3, to=1-4]
	\arrow[from=1-4, to=1-5]
\end{tikzcd}\]
and note that $B'\subset \hat{B}_i\subset B_i$ by definition. 
By construction, the conditions of Proposition~\ref{prop:ext_boundarymap} hold as for $i \preceq j$, $A_i \subset A_j, E_i \subset E_j$, and hence $B_i \subset B_j$.
Thus, the nets $(A_i)_{i\in I}, (B_i)_{i\in I}$ are increasing and cofinal families and the diagrams \eqref{eq:comm_diag_i_to_j} commute and have exact rows. The restriction of the semisplitting $\sigma: A_i\to E_i$ is also well-defined and gives a semisplitting for the subextension.

With the diagonal actions on $E_i$ the maps $\iota$ and $q$ are equivariant with respect to $\alpha$ as well as left and right multiplication by $\Cc$; this uses that $\sigma$ is also equivariant, hence the generating set for $E_i$ is invariant. 
The kernel $B_i$ is therefore also invariant, which gives $\Cc B_i\subset B_i$. Norm-density of $\Cc B_i$ in $B_i$ follows from Lemma~\ref{lem:Mfd_multiplier}: 
As the $C_0(f_d)$ representation on $B$ is nondegenerate, $\Cc$ contains an approximate unit for $B$, and hence $b\in \overline{\Cc b}$ for every $b\in B$. 
Finally, we have $A_i\Cc \subset \hat{B}_i\subset B_i$ by definition and then also $\mathrm{span} \{ab: a\in A_i,b\in B_i\} = A_i \overline{\Cc B_i} \subset B_i$, i.e., $A_i$ is also contained in the multiplier algebra of $B_i$.

\end{proof}
With condition (i), we have the following consequence of Lemma~\ref{lem:normcont} with $\RM^d$ in place of $\RM$:
\begin{corollary}
\label{cor:commutators_smooth}
$A_i$ and $B_i$ have the norm-dense subalgebras $\Aa_i$ and $\Bb_i$ of elements which are smooth with respect to the $\RM^d$-actions \eqref{eq:r_actions}. For each $a\in \Aa_i$ and $b\in \Bb_i$,
$$[D_f, a]\in \Aa_i \hat{\otimes}\CM_{d}, \quad [f_d, a]\in \Aa_i, \quad [D_{\underline{f}}, b]\in \Bb_i\hat{\otimes}\CM_{d-1}.$$
\end{corollary}
This gives us the dense subalgebras needed to define unbounded KK-classes without having to make sure $\CM_u(\Xx)$ or $\CM_u(\Zz_d\subset \Xx)$ have norm-dense intersections with $A_i$ or $B_i$ respectively.

With those preliminaries handled, we can now turn our attention to the actual computation of the Kasparov product with the extension class for fixed $i$. A sufficient condition which identifies the Kasparov product of two classes is given by Kucerovsky's criterion \cite[Theorem 13]{Kucerovsky}:

\begin{theorem}[Lemma 10 and Theorem 13 in \cite{Kucerovsky}]\label{kucerovsky}
Suppose that $(X \hat{\otimes}_B Y, \phi_1 \hat{\otimes} \Id , D)$, $(X, \phi_1, S_X)$ and $(Y, \phi_2, T_Y)$ are unbounded cycles representing classes in $KK(A,C)$, $KK(A,B)$ and $KK(B,C)$ respectively. Suppose that
\begin{enumerate}
    \item[(i)] For all $x$ in some dense subset of $\phi_1(A)X$, the operator
    \begin{align*}
        \bigg[
        \begin{pmatrix}
        D & 0 \\
        0 & T_Y \\
        \end{pmatrix},
        \begin{pmatrix}
        0 & L_x \\
        L_x^* & 0 \\
        \end{pmatrix}
        \bigg]
    \end{align*}
    is bounded on $\Dom(D \oplus T_Y)$ with the map $L_x \colon Y \to X \hat{\otimes}_BY$ that sends $e \mapsto x \hat{\otimes} e$,
    \item[(ii)] $\Dom D \subset \Dom S$, where $S = S_X \hat{\otimes} \one_Y$,
    \item[(iii)] $\langle S x, Dx \rangle + \langle Dx,  S x \rangle \geq -\kappa \langle x, x \rangle$ for all homogeneous $x$ in the domain of $D$.
\end{enumerate}
Then $(X \hat{\otimes}_B Y, \phi_1 \hat{\otimes} \Id , D)$ represents the Kasparov product of $[(X, \phi_1, S_X)] \in KK(A,B)$ and $[(Y, \phi_2, T_Y)] \in KK(B,C)$.
\end{theorem}

The first factor in the product is the extension class, whose unbounded representative is as follows:

\begin{proposition}
\label{prop:ext_class_unbounded}
Let $\pi: A \to M(B_i)$ be the canonical representation.
Let $(X, \phi_1, S_X)$ be the triple with the Hilbert $A_i$-$(B_i \hat{\otimes} \CM_1)$-bimodule $X= B_i \hat{\otimes} \CM_1$, the representation $\phi_1 \colon A_i \to \Ll_{B_i \hat{\otimes} \CM_1}(X)$ with $\phi_1=\pi \hat{\otimes} \one_{\CM_1}$ and the unbounded Dirac operator
$$S_X = f_d \hat{\otimes} \ep.$$
Then the unbounded Kasparov cycle $(X, \phi_1, S_X)$ is well-defined and represents the class of the extension $E_i$ in $KK_1(A_i, B_i)=KK(A_i, B_i \hat{\otimes}\CM_1)$.
\end{proposition}

\begin{proof}
First, observe that the compact and adjointable operators on $B_i\hat{\otimes}\CM_1$ are $B_i\hat{\otimes} \CM_1$, and $M(B_i)\hat{\otimes}\CM_1$, respectively. 
By Lemma~\ref{lem:Mfd_multiplier}, nondegeneracy of the representation, and the fact that the domain of $S_X$ is norm-dense by Proposition~\ref{prop:sep_approx}(ii), it is a densely defined regular self-adjoint operator on $B_i\hat{\otimes} \CM_1$. 

We still need to check the two conditions of Definition~\ref{def:unbddkk}:
\begin{enumerate}
    \item[(i)] For all $a \in \Aa_i$, the domain of $S_X$ is stable under $\phi_1(a)$ and the graded commutator $[S_X, \phi_1(a)]$ is in $M(B_i) \hat{\otimes} \CM_1$.
    \item[(ii)] $\phi_1(a)(S_X+\imath)^{-1} \in B_i\hat{\otimes}\CM_1$ for all $a\in A_i$.
\end{enumerate}
Condition (ii) is immediate from Proposition~\ref{prop:sep_approx}(ii).

For condition (i), we get boundedness of the commutator from Corollary~\ref{cor:commutators_smooth}, and since $A_i$ is contained in the multiplier algebra of $B_i$ by Proposition~\ref{prop:sep_approx}(ii), the commutator is also in $M(B_i)\hat{\otimes}\CM_1$.
It remains to check that $\Aa_i$ maps the domain 
$$\Dom(S_X)= (S_X + \imath)^{-1} (B_i\hat{\otimes}\CM_1)$$
of $S_X$ to itself. Computing as operators on $\Hh_\Xx\hat{\otimes}\CM_1$,
\begin{align*}
a (S_X + \imath)^{-1} &= (S_X + \imath)^{-1} a + [a,(S_X + \imath)^{-1}]\\
&= (S_X + \imath)^{-1} a + (S_X + \imath)^{-1}[S_X,a](S_X + \imath)^{-1}
\end{align*}
which implies $a \Dom(S_X)\subset \Dom(S_X)$.

Finally, we verify that this $KK$-cycle corresponds to the extension class. Since the Busby invariant is
\[\tau_i(a) = \Pp_{d} \pi(a) \Pp_{d} + B_i,\]
the extension class is represented by the Hilbert module $(X, \phi_1, F_X)$ with 
\[F_X=(2\Pp_{d}-1)\hat{\otimes}\ep\]
and $\phi_1(a)= a \hat{\otimes} 1 \in M(B_i\hat{\otimes} \CM_1)$. Since $\Pp_d=\rho(f_d)$ with a switch function $\rho$, we have $F_X=\Theta(S_X)$ with $\Theta = 2\rho-1$.
It is then easy to check the three conditions of Proposition~\ref{prop:check_extension}.
Hence, the extension class has the unbounded representative $(X, \phi_1, S_X)$.
\end{proof}

The second factor of the product is given by the class $[D_{\underline{f}}]$, which one amplifies with $\tau_{\CM_1}$ before taking the product to match the Clifford dimensions. 
\begin{proposition}
Let $[D_{\underline{f}}^i]\in KK(B_i, \CM_{d-1})$ be the pullback of $[D_{\underline{f}}]$ along $B_i\to B$. Then $\tau_{\CM_1}([D^i_{\underline{f}}])$ is represented by the $KK(B_i\hat{\otimes}\CM_1, \CM_d)$-cycle $(Y, \phi_2, T_Y)$ with  the Hilbert $(B_i\hat{\otimes}\CM_1)$-$\CM_d$-bimodule $Y= \Hh_\Xx \hat{\otimes} \CM_{d}$, the representation $\phi_2 \colon B_i \hat{\otimes} \CM_1 \to \Ll_{\CM_d}(Y)$ given by
    $$\phi_2(b\hat{\otimes} \sigma) = b \hat{\otimes}  \one_{\CM_{d-1}} \hat{\otimes} \sigma, \qquad b\in B_i, \sigma\in \CM_1,$$
and the unbounded regular self-adjoint operator $T_Y  = \sum_{k=1}^{d-1} {f_k} \hat{\otimes} \sigma_k= \tilde{T}_Y\hat{\otimes} 1_{\CM_1}$. 
\end{proposition}
\begin{proof}
The KK-cycle is obtained by pullback of $[D_{\underline{f}}]$ to $B_i$ and subsequent exterior product with $\CM_1$. It is well-defined since Corollary~\ref{cor:commutators_smooth} provides a dense smooth subalgebra on which the commutators are bounded, and satisfies the properties of regularity, domain invariance, and relative compactness of the resolvent analogously to Proposition~\ref{prop:ext_class_unbounded}.
\end{proof}

\begin{proposition}\label{prop:kasparov_sum}
On $Y=\Hh_\Xx\hat{\otimes}\CM_d$ define the regular self-adjoint operators $S = {f_d} \hat{\otimes} \one_{\CM_{d-1}}\hat{\otimes} \ep$ and $T = T_Y = \sum_{k=1}^{d-1} {f_k} \hat{\otimes} \sigma_k \hat{\otimes} 1_{\CM_1}$. Then $(Y, \pi, S+T)$ represents the Kasparov product of $(X, \phi_1, S_X)$ and $(Y, \phi_2, T_Y)$.
\end{proposition}

\begin{proof}

Since $X=B_i\hat{\otimes}\CM_1$, the internal tensor product satisfies $X \hat{\otimes}_{B_i\hat{\otimes}\CM_1} Y\simeq Y$ under the natural isomorphism $\eta(b\hat{\otimes}\sigma\hat{\otimes}\psi)= \phi_2(b\hat{\otimes}\sigma)\psi$. Note that $S$ corresponds to $S_X\hat{\otimes} \one_Y$ under this isomorphism, as its action on a triple $\psi \hat{\otimes} v \hat{\otimes} w \in \Hh_\Xx \hat{\otimes} \CM_{d-1} \hat{\otimes} \CM_1$ is given by
$$
S \phi_2(b \hat{\otimes} \sigma)(\psi \hat{\otimes} v \hat{\otimes} w) =(-1)^{\pdv v (\pdv \sigma+1)} ({f_d}b \psi) \hat{\otimes} v \hat{\otimes} (\ep \sigma w)=\phi_2(S_X(b\hat{\otimes}\sigma))(\psi \hat{\otimes} v \hat{\otimes} w)$$
with $\pdv$ the degree in $\ZM/2\ZM$. 

It is clear that $S+T=\sum_{i=1}^d f_i \hat{\otimes}\sigma_i=D_f$ is well-defined as an unbounded $KK(A_i,\CM_d)$-cycle; the proof follows as in Lemma~\ref{lemma:dirac_op} with the dense subalgebra $\Aa_i$ from Corollary~\ref{cor:commutators_smooth} in place of $\CM_u(\Xx)$. 

We now check the conditions of Theorem \ref{kucerovsky} with $D=S+T$.
For Condition (i), we note that by the isomorphism it suffices to check that
for all $x$ in some dense subset of $\phi_1(\mathcal{A}_i)X$, the operator
    \begin{equation}
    \label{eq:comm_kucerovsky} 
        \bigg[
        \begin{pmatrix}
        S + T & 0 \\
        0 & T_Y \\
        \end{pmatrix},
        \begin{pmatrix}
        0 & L_x \\
        L_{x^*} & 0 \\
        \end{pmatrix}
        \bigg]
    \end{equation}
    is bounded on $\Dom((S+T) \oplus T_Y)$.
    For the dense subset of $\phi_1(\Aa_i)X$ take 
    \[X_0 = \phi_1(\Aa_i)(S_X+\imath)^{-1} (\Bb_i\hat{\otimes}\CM_1)=\phi_1(\Aa_i)\Dom(S_X),\]
    which is dense since $\Dom(S_X)$ is dense in $X$. One has $X_0 \subseteq \Dom(S_X)$ as seen in the proof of Proposition~\ref{prop:ext_class_unbounded}.
    To check the connection condition, we note that under $\eta$ the operators $L_x$ correspond to the natural action of $x$ on $Y=\Hh_\Xx\hat{\otimes}\CM_{d}$. In that sense
    $$SL_x + TL_x - (-1)^{|x|}L_x T_Y = Sx + Tx - (-1)^{|x|}x T_Y= Sx+ [T, x]$$ as $T=T_Y$ under this identification, and moreover this is bounded as $x \in \Dom(S_X)$ and $x$ is smooth. The same argument for $x^*$ shows the commutator \eqref{eq:comm_kucerovsky} is bounded.
    
Using $ST+TS=0$ on $\Dom(S) \cap \Dom(T)$ one has $\|(S+T)\psi\|^2=\|S\psi\|^2+\|T\psi\|^2$ and $\langle S \psi, (S+T)\psi \rangle + \langle (T+S)\psi,  S \psi \rangle = 2\langle S \psi, S\psi\rangle\geq 0$. The latter is Condition (iii) and the former implies Condition (ii) since it means $\Dom(S+T)=\Dom(S)\cap\Dom(T)
\subseteq\Dom(S)$. 

\end{proof}

\begin{proof}[Proof of Proposition~\ref{prop:bulkedge_kk}]
By Proposition~\ref{prop:kasparov_sum}, we have that
\[[S+T]= [S_X]\otimes [T_Y] = [{\Ext}_{E_i}]\otimes [T_Y]\]
as KK-classes. 
Since $[S+T]=[D_f^{i}]$ and $[T_Y]=\tau_{\CM_1}([D_{\underline{f}}^{i}])$, by definition of the Kasparov product for higher KK-groups \eqref{eq:KK_product} this is the same as $[D_f^i]= [{\Ext}_{E_i}]\otimes [D^i_{\underline{f}}]$.
\end{proof}

\subsection{Examples}\label{ssec:bec_examples}

In this section, we describe physically relevant examples for which bulk-edge correspondence applies. 

\begin{definition}
\label{def:bulk_invariants}
Let $H$ be a (possibly unbounded) self-adjoint operator on a geometric module $\Hh_\Xx \otimes \CM^N$, $N$ even (i.e., with $N$ local degrees of freedom). Assume that the bulk Hamiltonian $H$ has a spectral gap $\Delta\subset \RM$ which contains $0$. 

The Fermi projection is the spectral projection $p = \chi(H<0)$. If there is an additional chiral symmetry
$$J H J = - H, \qquad J = \begin{pmatrix}
\one_{N/2} & 0 \\ 0 & -\one_{N/2}
\end{pmatrix}$$
if $J\Dom(H) \subset \Dom(H)$, then the Fermi projection has the structure
$$p = \frac{1}{2}(1-\sgn(H))=  \frac{1}{2} \one_N - \frac{1}{2} \begin{pmatrix}
    0 & u^* \\ u & 0
\end{pmatrix}$$
with a unitary $u$, called the Fermi unitary.
\end{definition}
In situations where the Fermi projection and Fermi unitary are in $M_N(C^*_u(\Xx,\Hh_\Xx)^+)$ and $M_{N/2}(C^*_u(\Xx,\Hh_\Xx)^+)$ respectively, $[p]-[s(p)]$ and $[u]$ taken to have scalar part $1$ immediately define K-theory classes.
\begin{definition}
\label{def:edge_invariants}
Let $\hat{H}$ be a (possibly unbounded) self-adjoint operators on a geometric module $\Hh_\Xx \otimes \CM^N$ which is invertible modulo the ideal $C^*_u(\Zz_d\subset \Xx,\Hh_\Xx)$, i.e., there exists an interval $\Delta$ which contains $0$ such that
$$f(\hat{H})\in M_N(C^*_u(\Zz_d \subset \Xx,\Hh_\Xx))$$
for every function $f\in C_c(\Delta)$. 

Let $\rho_\Delta$ be any smooth increasing function  which is $0$ to the left of $\Delta$ and $1$ to the right. We define the edge unitary
$$u_\Delta = e^{2\pi \imath \rho_\Delta(\hat{H})}$$
which is a matrix over the unitization of $C^*_u(\Zz_d\subset \Xx)$.
    
If there is an additional chiral symmetry
$$J \hat{H} J = - \hat{H}, \qquad J = \begin{pmatrix}
\one_{N/2} & 0 \\ 0 & -\one_{N/2}
\end{pmatrix}$$
we can define an edge projection as
$p_\Delta = \frac{1}{2}\left(\one_N + Je^{\imath \pi f_\Delta(\hat{H})}\right)$
for a function $f_\Delta$ which is anti-symmetric, $-1$ below $\Delta$ and $+1$ above.
\end{definition}
Note that we do not impose any relation between bulk and edge Hamiltonian at this point, since the precise setup depends on whether one is modeling bulk-edge or bulk-interface correspondence.
For lattice models, Hamiltonians are bounded operators and therefore elements of some uniform Roe algebra. Going back to the computation of boundary maps in \cite{ProdanSchulzBaldes2016}, which uses the conventions of \cite{RordamLarsenLaustsen2000} for the boundary maps, the following is standard as a consequence of the definition of bulk and edge classes:
\begin{proposition}
\label{prop:boundary_maps}
Suppose $\Xx$ is uniformly discrete and consider an interface model of the form
$$\hat{H}= P_{\Yy_d} H_+ P_{\Yy_d}  +  P_{\Yy^c_d} H_- P_{\Yy^c_d} + z$$
with two Hamiltonians $H_\pm \in C^*_u(\Xx,\ell^2(\Xx))\otimes M_N(\CM)$ with a common spectral gap $\Delta$, and some edge term $z=z^*\in C^*_u(\Zz_d\subset \Xx, \ell^2(\Xx))\otimes M_N(\CM)$.

Then the Fermi projections $p_\pm = \chi(H_\pm<0)$ define classes in $K_0(C^*_u(\Xx,\ell^2(\Xx)))$ and 
$$[u_\Delta]_1 = \pdv_{\Yy_d} \left( [p_+]_0- [p_-]_0\right) \in K_1(C^*_u(\Zz_d\subset \Xx, \ell^2(\Xx))).$$

If $\hat{H}$, $H_+$ and $H_-$ have an additional chiral symmetry one similarly has
\[[p_\Delta]_0-[s(p_\Delta)]_0 = \pdv_{\Yy_d} \left( [u_+]_1- [u_-]_1\right) \in K_0(C^*_u(\Zz_d\subset \Xx, \ell^2(\Xx)))\]
where $s[p_{\Delta}]_0$ is the scalar part of $[p_{\Delta}]$.
\end{proposition}
If $H_-$ is a scalar matrix then it represents empty space and consequently the setting reduces to bulk-edge correspondence. Note that $H_-$ cannot be zero, since its resolvent set needs to include $\Delta$. 

In the continuous case, one will want to consider unbounded Hamiltonians and this becomes much more difficult. One problem is (as happens e.g., for Dirac-type operators) that already the Fermi projection can fail to define a class in K-theory, as it is generally not a matrix over the unitization of $C^*_u(\Xx)$ \cite{KSt}. A useful regularization is therefore to use relatively defined bulk invariants which compare pairs of bulk Hamiltonians. This leads to concepts like relative bulk-edge correspondence \cite{KSt} or bulk-difference-interface correspondence (in the terminology of \cite{Bal2022TopologicalInvariants}). Likewise, edge invariants can be ill-defined in principle, especially for physically unrealistic boundary conditions. Even if both bulk and edge invariants are well-defined, they still need not be related by the boundary map because the bulk-edge relation can depend on the choice of boundary conditions \cite{KSt,JudTauber2026,GrafTarantola2026}.

The bulk-interface correspondence for domain-wall models, however, can be handled very well since there are no problems with boundary conditions:
\begin{proposition}\label{prop:dense_resolvent_example}
Let $D$ be a densely defined self-adjoint operator on $\Hh_\Xx$ with resolvent in $C^*_u(\Xx,\Hh_\Xx)$, and consider three bounded self-adjoint operators $V_+,V_-, \hat{V}$ in the multiplier algebra $M(C^*_u(\Xx,\Hh_\Xx))$ such that
\begin{equation}
\label{eq:relatively_compact_interface}
\begin{aligned}
(D+\imath)^{-1}(\hat{V}- V_+)\in C^*_u(\Yy_d^c\subset \Xx, \Hh_\Xx) \hbox{ and }
(D+\imath)^{-1}(\hat{V}- V_-)\in C^*_u(\Yy_d\subset \Xx, \Hh_\Xx).
\end{aligned}
\end{equation}
Assume that $H_\pm=D+ V_\pm$ have a common spectral gap $\Delta$. 
Then the relative K-theory class $[p_+,p_-]_0 \in K_0(C^*_u(\Xx,\Hh_\Xx))$ defined by the Fermi projections $p_\pm=\chi(H_\pm<0)$ satisfies
$$[u_\Delta]_1 = \pdv_{\Yy_d}[p_+,p_-]_0 \in K_1(C^*_u(\Zz_d\subset \Xx,\Hh_\Xx))$$
with the edge unitary $u_\Delta$ defined by the interface model $\hat{H}=D+ \hat{V}$
\end{proposition}
\begin{proof}
The relative K-theory class is defined in \cite[Definition 3.11]{KSt}. If the $p_\pm$ do define K-theory classes in the Roe algebra individually, then $[p_+,p_-]_0= [p_+]_0-[p_-]_0$ but this need not be the case.

The main technical difficulty is that the assumptions are not strong enough to imply  that $D$ is affiliated to the Roe algebra, otherwise this would follow immediately from \cite[Theorem 3.15]{KSt}. 
Loosely following ideas of \cite{Trout2000}, we can circumvent this problem by using the largest subalgebra that $D$ is affiliated to. Abbreviating $A=C^*_u(\Xx,\Hh_\Xx)$ the assumption is $r:=(D+\imath)^{-1}\in A$. Therefore, $D$ is always affiliated to the hereditary subalgebra $A_D= \overline{r C^*_u(\Xx)r^*}$ (e.g., this is an immediate consequence of \cite[Lemma 2.6]{KSt}). 
Abbreviating $I_+=C^*_u(\Yy_d\subset \Xx)$ and $I_-=C^*_u(\Yy_d^c\subset \Xx)$, the closed hereditary subalgebras $I_{\pm,D}:= A_D\cap I_\pm$ are ideals of $A_D$. There exists a pullback algebra $E_D\subset E$ defined as in Definition~\ref{def:pullback}, but using the smaller algebras $A_D$ and $I_{+,D}$ which gives a commutative diagram with exact rows
\begin{equation*}
	\begin{tikzcd}
		0 \arrow[r] & B_D \arrow[r] \arrow[d, "j_B"] & E_D \arrow[r] \arrow[d] & A_D  \arrow[d, "j_A"]  \arrow[r] & 0 \\
		 0 \arrow[r]  & B \arrow[r] & E \arrow[r] & A  \arrow[r] & 0
	\end{tikzcd}
\end{equation*}
where we set $B_D= I_{+,D}\cap I_{-,D}$ and $j_A,j_B$ are the natural inclusion maps. The resolvents of $\hat{H}$, $H_\pm$ lie in $A_D$ as one has for $|\lambda|>\|v\|$ the norm-convergent expansion $$(D+v+\lambda\imath)^{-1}=\sum_{k=0}^\infty (-1)^k (D+\lambda\imath)^{-1} (v (D+\lambda\imath)^{-1})^k$$
and the right-hand side for $v\in \{\hat{V},V_+,V_-\}$ is in $A_D$. 
By perturbation theory for the bounded transforms $F(\lambda)=\lambda(1+\lambda^2)^{-\frac12}$, the difference $F(h+v)-F(h)$ for $h$ self-adjoint with $(h+\imath)^{-1}\in A_D$ and $v\in M(A_D)$ bounded can be written with a norm-convergent integral formula \cite[Lemma 2.7]{CareyPhillips1998} which implies $F(h+v)-F(h)\in I$ for some $C^*$-ideal $I\subset A_D$. 
Whenever the resolvent-difference is in $I$,
$$(h+\lambda\imath)^{-1}-(h+v+\lambda\imath)^{-1}=(h+\lambda\imath)^{-1} v (h+v+\lambda\imath)^{-1} \in I.$$
With these observations in hand, one concludes that $H_+$ and $H_-$ are $A_D$-comparable and from \eqref{eq:relatively_compact_interface} also that $(\hat{H},H_+)$ and $(H_-,H_-)$ are $E_D$-comparable. Using pair algebras as in \cite[Proposition 3.10, Theorem 3.15]{KSt}, the boundary maps can be computed to conclude that
$$[(\hat{H},H_+), (H_-,H_-)]_1=\pdv_{E_D} [p_+,p_-]_0$$
with a relative edge invariant in $K_1(B_D)$ \cite[Definition 3.11]{KSt}. Since $(H_-,H_-)$ has a spectral gap the left-hand side is exactly $[u_\Delta]_1$ which finishes the proof together with naturality of the boundary maps.
\end{proof}
For potentials, the above conditions can be checked in terms of the geometric support:

\begin{lemma}\label{lem:geometric_support}
In the setting of the previous Proposition, if $V_\pm,\hat{V}:\Xx \to M_N(\CM)$ are three bounded self-adjoint Borel functions such that $\hat{V}-V_+$ and $\hat{V}-V_-$ are supported near $\Yy_d^c$ and $\Yy_d$ respectively, then \eqref{eq:relatively_compact_interface} holds.
\end{lemma}

\begin{proof}
Clearly, $(D+i)^{-1}(V-V_{\pm})f =0$ if $f \in C_0(\Xx)$ with support away from $\Yy_d$ or $\Yy_d^c$.
For left multiplication, approximate $(D+i)^{-1}$ by finite propagation operators.
\end{proof}

As special cases, differential operators which are resolvent-affiliated to a crossed-product algebra in the sense of \cite{KSt} satisfy the resolvent condition in Proposition~\ref{prop:dense_resolvent_example}. This includes (magnetically) translation-invariant elliptic differential operators:
\begin{example}
Fix $\Yy$ a half-space and consider a bounded potential $\hat{V}:\RM^2\to \RM$ of the form $$\hat{V}=\begin{cases}
        V_+(x) & \text{if }x\in \Yy \\
        V_-(x) & \text{if }x\notin \Yy
    \end{cases}$$
    with $V_+, V_-$ bounded Borel functions. 
\begin{enumerate}
    \item[(i)] On $L^2(\RM^2)$ consider a magnetic Schrödinger operator
    $$\hat{H}= p_1^2 + (p_2 - B x_1)^2 + \hat{V} =: H_0 + \hat{V}$$
    with a constant magnetic field $B \in \RM$. If $H_\pm = H_0 + V_\pm$ have a common spectral gap, then 
    $$[u_\Delta]_1 = \pdv_{{\Yy_d}}([p_+]_0 - [p_-]_0).$$
    \item[(ii)] On $L^2(\RM^2)\otimes \CM^2$ consider a domain-wall Dirac operator 
    $$\hat{H}= p_1\sigma_1 + p_2\sigma_2 + \hat{V}\sigma_3$$
    with the Pauli matrices. If $H_\pm = p_1\sigma_1 + p_2\sigma_2  + V_\pm \otimes \sigma_3$ have a common spectral gap, then $$[u_\Delta]_1 = \pdv_{\Yy_d}([p_+,p_-]_0).$$
\end{enumerate}
\end{example}
The difference between the two examples is that the Schrödinger operators are bounded from below, therefore their spectral projections define K-theory classes immediately. 

More generally, as is already the case for crossed product algebras on straight half-spaces, the uniform Roe algebras are also large enough to describe models with sharp boundaries. For example, the so-called strongly affiliated differential operators discussed in \cite{KSt} for the standard half-spaces also satisfy the K-theoretic bulk-edge relation as elements of the uniform Roe algebras. 
However, for e.g., elliptic differential operators in the presence of curved Dirichlet boundaries, a proof that their resolvents are in the uniform Roe algebras does not seem to be readily available in the literature and, unfortunately are out of scope for this paper. We will have to leave this to future work.

\section{Homotopy and Degree in Euclidean Space}\label{sec:euc_space}
\label{sec:degree}

In this section, assume that $\Xx = \RM^d$, $d > 1$, and therefore that $\Lambda$ is a Delone set. Let us also fix as before some coarsely faithful geometric module $\Hh_\Xx$. As explained in the introduction, the geometric module need not be $L^2(\RM^d)$, but for example could be $\ell^2(\Lambda)$ which would reduce to the discrete setting. 
This flexibility in constructing the geometric module essentially yields the freedom to think of any space $\Xx$ coarsely equivalent to $\RM^d$ by taking $\RM^d$ equipped with the geometric module $L^2(\Xx)$.

\begin{remark}
\label{rem:flipschitz}
{\rm We will from now on often assume that Dirac-type functions $f$ are Lipschitz-continuous or smooth. We can assume this without loss of generality in $\RM^d$, by the large-scale Lipschitz assumption replacing $f$ by its convolution $f*\psi$ with a mollifier is a bounded perturbation and hence does not change the K-homology class.}$\diamond$
\end{remark}

\begin{definition}
For any continuous function $f: \RM^d \to \RM^d$ which is non-zero on the sphere $\pdv B_R$, define the degree
$$\deg_R(f) = \deg\left(\frac{f}{\abs{f}}\bigg\rvert_{\pdv B_R}\right)$$
where the right-hand side is the mapping degree of a function $\pdv B_R\simeq \SM^{d-1}\to \SM^{d-1}$ using the radial orientation-preserving isomorphism.
\end{definition}
For any continuous Dirac-type function $f$ the degree is well-defined by properness for any large enough radius $R$ and the limit
$$\deg(f) = \lim_{R\to \infty}\deg_R(f)$$
exists since the right-hand side is eventually constant by homotopy invariance.

Our main theorem of this section is the following:

\begin{theorem}\label{thm:degree_multiplicity}
Let the standard Dirac operator $D_{\mathrm{std}}$ be defined by the symbol function $f_{\mathrm{std}}: x\mapsto x$ which has degree $1$ by definition.

If $f$ is a continuous Dirac-type function, then 
$$[D_f] = \deg(f)[D_{\mathrm{std}}].$$
In particular, for the same class $x\in K_{d \bmod 2}(C^*_u(\Xx, \Hh_\Xx))$, the index pairings are related by
$$\langle x, [D_f] \rangle =\deg(f) \langle x, [D_{\mathrm{std}}]\rangle.$$
\end{theorem}

In conclusion, the degree is a complete invariant for Dirac-type functions with respect to K-homology: 
As the module $\Hh_\Xx$ is coarsely faithful, there exists a K-theory class whose pairing with the standard Dirac operator is equal to $1$. The range of values for the index pairing is therefore precisely $\deg(f)\ZM$. 

To prove this we first construct a homotopy which relates any Dirac-type function to a simple standard model:

\begin{definition}\label{definition:homogeneous_homotopy}
We will say a Dirac-type function $f: \RM^d \to \RM^d$ is homogeneous if it is of the form
$$ f(x) =
\begin{cases}
   \abs{x} u\bigl(\frac{x}{\abs{x}}\bigr) & x\neq 0\\
    0 & x=0
\end{cases}
$$
for some Lipschitz sphere-valued function $u: \SM^{d-1}\to \SM^{d-1}$.
\end{definition}

We will now show that any continuous Dirac-type function $f$ can be deformed to any homogeneous function of the same degree. Since the mapping degree is a complete homotopy-invariant for maps $\SM^{d-1}\to \SM^{d-1}$, it is clear that any two homogeneous Dirac-type functions define the same K-homology class if they have the same degree.

\begin{lemma}
\label{lem:homogeneous}
For any smooth Lipschitz Dirac-type function $f^0:\RM^d \to \RM^d$ and any homogeneous Lipschitz Dirac-type function $f^0$ with the same degree, then there exists a continuous path of Dirac-type functions $(f^{t})_{t\in [0,2]}$ between $f^0$ and $f^2$ which satisfies the conditions of Lemma~\ref{lemma:homotopy}, i.e. $[D_{f^0}]=[D_{f^2}]$.
\end{lemma}
\begin{proof}
Choose a radius $R$ such that $\inf_{|x| \geq R} |f(x)|$ is bounded from below for some $C>0$, which implies that $\deg_R(f^0)=\deg(f^0)$. The first part of the homotopy for $t\in [0,1]$ deforms $f^0$ inside the ball $B_{2R}$ to obtain a Lipschitz function $f^1$ with $$f^1\rvert_{B_R}=f^2 \rvert_{B_R}, \qquad f^1\rvert_{B_{2R}^c}=f^0\rvert_{B_{2R}^c}$$
and for which $\abs{f^1} > 0$ on $B_{2R} \setminus B_R$. 
Such a homotopy $(f^{t})_{t\in [0,1]}$ exists, since the degree is a complete homotopy-invariant for functions $\SM^{d-1}\to \SM^{d-1}$  and it can be chosen to be smooth.
The first part of the homotopy is the linear interpolation
$$f^t= (1-t) f^0 + t f^1, \qquad t\in [0,1]$$ which is certainly uniformly Lipschitz and uniformly proper as it is only a bounded perturbation.

The second part of the homotopy is the rescaling
$$f^t(x) = \frac{1}{2-t} f^1\left((2-t)x\right), \qquad t\in [1,2)$$
Note that $f^t$ coincides with $f^2$ on the ball of radius $\frac{R}{2-t}$, hence we obtain a locally uniformly continuous function on $[1,2]\times \RM^d$. Since we are only rescaling the function $f^1$, all $f^t$ are Lipschitz continuous with uniformly bounded Lipschitz constant. 

For uniform properness we claim that the continuous function
$$\rho(x) := \max_{t\in [1,2]} (1+|f^t(x)|^2)^{-\frac12}$$
is in $C_0(\RM^d)$. 
Suppose instead that $\rho$ does not tend to $0$. 
Then there exist sequences $x_n\to \infty$ and $t_n \to 2$ such that $|f^{t_n}(x_n)|$ does not tend to $\infty$. Put $y_n= (2-t_n)x_n$. If $y_n\to \infty$ then eventually 
\[|f^{t_n}(x_n)|=(2-t_n)^{-1} |f^1(y_n)|\geq |f^1(y_n)|,\]
but $|f^1(y_n)| \to \infty$ due to properness of $f^1$. Hence after passing to a subsequence, $y_n$ must be bounded and converge to some $y\in \RM^d$.

In particular, $(2-t_n)=\frac{|y_n|}{|x_n|}\to 0$ since $x_n\to \infty$. As $f^1$ vanishes only at the origin, $y\neq 0$ would imply $f^1(y)\neq 0$.
Hence, $|(2-t_n)^{-1}f^1(y_n)| \to \infty$ which implies $|f^{t_n}(x_n)|\to \infty$. 

Now suppose that $y_n\to 0$. In that case, $$|f^{t_n}(x_n)|=(2-t_n)^{-1}|f^1(y_n)|= (2-t_n)^{-1}|y_n| = |x_n|$$ for large $n$ as eventually $f^1(y_n)=f^\infty(y_n)$ and $f^\infty$ is homogeneous. This finishes the proof by contradiction.
\end{proof}

To prove Theorem~\ref{thm:degree_multiplicity}, it remains to relate the classes to multiples of the standard Dirac operator. This cannot be done while staying in the class of Dirac-type functions as directs sums of the standard Dirac operator act on a larger Hilbert space. 
Nevertheless, it can be done at the level of K-homology, as the equivalence relation is stable homotopy.

\begin{proof}[Proof of Theorem~\ref{thm:degree_multiplicity}]
We may assume $f$ is smooth. By Lemma~\ref{lem:homogeneous}, we can then also assume without loss of generality that $f$ is homogeneous, i.e., determined by a function  $u: \SM^{d-1}\to \SM^{d-1}$. 
We may further assume $\deg(f) = n \geq 0$ by replacing $f_1$ with $-f_1$ which inverts both the degree and the K-homology class. We will show that $[D_f]=n[D_{\mathrm{std}}]$.

 Given a Lipschitz-continuous self-adjoint unitary matrix-valued function $U\in C(\SM^{d-1})\otimes M_{N}(\CM)$, define a $KK_1(C^*_u(\Xx),\CM)$-class by the odd Dirac operator
$$D_U = |x|U\left(\frac{x}{|x|}\right)\hat{\otimes}\ep$$
on $(\Hh_\Xx\otimes \CM^N)\hat{\otimes}\CM_1$, with $D_U(0)=0$. Similarly a unitary $V\in C(\SM^{d-1})\otimes M_{N/2}(\CM)$ determines a class in $KK_0(C^*_u(\Xx),\CM)$ on $\Hh_\Xx\otimes \CM^N$ with odd Dirac and grading operators given by
$$D_V = |x| \begin{pmatrix}
    0 & V\bigl(\frac{x}{|x|}\bigr)\\
    V^*\bigl(\frac{x}{|x|}\bigr) & 0
\end{pmatrix}, \qquad \gamma_0 =  \begin{pmatrix}
    \one_{N/2} & 0\\
    0 &  -\one_{N/2}
\end{pmatrix}$$
and $D_V(0)=0$ and $N$ even.
With the irreducible representation $\gamma_1,\ldots,\gamma_d$ of the Clifford-algebra as in Appendix~\ref{sec:index}, a homogeneous Dirac-type function determines the self-adjoint unitary matrix
$$U_f(x) = \sum_{i=1}^d  u_i(x) \gamma_i.$$
If $d$ is even, this has the off-diagonal structure
$$U_f(x) = \begin{pmatrix}
0 & V_f(x) \\ V_f(x)^* & 0
\end{pmatrix}$$
with a unitary matrix $V$. Then $[D_{U_f}]$ and $[D_{V_f}]$ in odd and even dimension respectively are precisely the classes corresponding to $[D_f]$ under Bott periodicity $KK_d \simeq KK_{d\bmod 2}$. 
We write $U_{\mathrm{std}}$ and $V_{\mathrm{std}}$ for the corresponding respective unitaries for $D_{\mathrm{std}}$.

We now use the K-theory of the $(d-1)$-dimensional sphere for $i \in \{0,1\}$,
$$K^i(\SM^{d-1}) =\begin{cases}
    \ZM \oplus \ZM & \text{if }d \text{ odd, }i=0 \\
    \ZM  & \text{if }d \text{ even, }i \in \{0,1\} \\
    0 & \text{else}
\end{cases}$$
to construct a stable homotopy of $KK_{d\bmod 2}$-classes. 

First, consider the case of odd $d$. As in \cite{SchulzBaldesStoiber2023}, we may use a picture of $K$-theory where $K_0$ is represented by self-adjoint unitary matrices and the equivalence relation is stable homotopy. 
In this case, $K_0(C(\SM^{d-1}))=\ZM\oplus \ZM$.
In terms of complex vector bundles on $\SM^{d-1}$, the first integer is the top degree of the Chern character, while the second is the rank of the vector bundle. 
The first integer therefore corresponds to the mapping degree, whereas the second is generated just by constant self-adjoint unitaries. If we assume without loss of generality that $n$ is positive, then there exist constant self-adjoint unitary matrices $M_0, M_1$ such that we have a norm-continuous path of self-adjoint unitary matrices $U_t \in C(\SM^{d-1})\otimes M_{\tilde{N}}(\CM)$ with $U_0=U_f \oplus M_0$ and $U_1= (U_{\mathrm{std}})^{\oplus n}\oplus M_1$. 
We may further assume the path is smooth and takes values in the smooth functions. As in Lemma~\ref{lemma:homotopy}, the Dirac operators defined by the homogeneous extension $$\hat{D}_t = D_{U_t}$$ then give a homotopy of $KK$-cycles in $KK_1(C^*_u(\Xx, \Hh_\Xx), \CM)$. 
Consequently,
$$[D_{U_f}]+[D_{M_0}]= [\hat{D}_0] = [\hat{D}_1] = n [D_{U_\mathrm{std}}]+ [D_{M_1}] = n[D_{U_\mathrm{std}}]$$
which becomes $[D_f]=n [D_{\mathrm{std}}]$ under Bott periodicity.
Here we used that
$$[D_{M_0}]=0=[D_{M_1}]$$
and claim that for any constant matrix $M_i$, the corresponding bounded $KK$-cycle $(E,\pi, F_i)$ with $F_i=D_{M_i}(1+D^2_{M_i})^{-\frac12}$ is equivalent to a degenerate cycle.
To show this, we note that $(F_{i}-{M_i}\hat{\otimes}\ep)\pi(a)$ is compact since $|x|(1+|x|^2)^{-\frac12}-1$ is a $C_0$-function. Since $M_i$ is a constant self-adjoint unitary,
\[[M_i \hat{\otimes}\ep,\pi(a)]=0, \quad M_i^2 = 1, \quad M_i \hat{\otimes}\ep = (M_i \hat{\otimes}\ep)^*\]
hence the cycle is degenerate as desired.

For even $d$, the argument works similarly but with the class $[V]_1\in K_1(C(\SM^{d-1}))\simeq \ZM$ corresponding to the mapping degree instead. 
After stabilization with constant matrices, there exists a smooth path of unitary matrix-valued functions on $\SM^{d-1}$ which connects $V$ to $n$ copies of a reference unitary. Extending that path homogeneously then gives the desired homotopy in $KK_0(C^*_u(\Xx, \Hh_\Xx), \CM)$.
\end{proof}

\subsection{Brouwer Degree and Asymptotic Degree}
\label{sec:brouwer}

As sketched in the introduction, for explicit computations it is useful to
identify the asymptotic degree with the Brouwer degree. We briefly recall a
convenient construction of the Brouwer degree in unbounded regions and the properties that will be
used below; see
\cite[Definition~3.7, Corollary~3.8, Theorem~3.17, and
Definition~4.1]{BenevieriFuriPeraSpadini2023} and
\cite[Chapter~1]{DincaMawhin2021}.

\begin{remark}[Brouwer degree]\label{rem:brouwer_degree}
{\rm
Let $\Omega\subset\RM^d$ be open, let $f:\Omega\to\RM^d$ be continuous,
and let $y\in\RM^d$. A triple $(f,\Omega,y)$ is called weakly admissible if $K:=f^{-1}(y)\cap\Omega$ is compact.

Choose a bounded open set $V$ such that $K\subset V$ and $\overline V\subset\Omega$, then
\[
    \delta:=\mathrm{dist}\bigl(y,f(\pdv V)\bigr)>0.
\]
By smooth approximation and Sard's theorem, one may choose a smooth map
$g$, defined on a neighborhood of $\overline V$, such that $y$ is a
regular value of $g$ and 
\[ \|g-f\|_{L^\infty(V)}<\delta.\]
The Brouwer degree of the weakly admissible triple is then given
by the Jacobian formula
\begin{equation}\label{eq:brouwer_jacobian_formula}
    \deg_{\mathrm B}(f,\Omega,y)
    :=
    \sum_{x\in g^{-1}(y)\cap V}
    \sgn\bigl(\det Dg(x)\bigr).
\end{equation}
The sum is finite, and its value is independent of the choices of $V$
and $g$.}$\diamond$
\end{remark}

We recall the following properties of the Brouwer degree
\cite[Theorem~3.10 and Theorem~4.2]
{BenevieriFuriPeraSpadini2023}:

\begin{enumerate}
\item[(i)] \emph{Normalization.} For the identity map
$I:\RM^d\to\RM^d$,
\[
    \deg_{\mathrm B}(I,\RM^d,0)=1.
\]

\item[(ii)] \emph{Additivity.} If $\Omega_1,\Omega_2\subset\Omega$ are
disjoint open sets and
\[
    f^{-1}(y)\cap\Omega\subset\Omega_1\cup\Omega_2,
\]
then
\[
    \deg_{\mathrm B}(f,\Omega,y)
    =
    \deg_{\mathrm B}(f,\Omega_1,y)
    +
    \deg_{\mathrm B}(f,\Omega_2,y).
\]

\item[(iii)] \emph{Homotopy invariance.} Let
$H:\Omega\times[0,1]\to\RM^d$ and
$\alpha:[0,1]\to\RM^d$ be continuous. If
\[
    \bigl\{(x,t)\in\Omega\times[0,1]:
    H(x,t)=\alpha(t)\bigr\}
\]
is compact, then
\[
    \deg_{\mathrm B}(H(\,\cdot\,,0),\Omega,\alpha(0))
    =
    \deg_{\mathrm B}(H(\,\cdot\,,1),\Omega,\alpha(1)).
\]

\item[(iv)] \emph{Excision.} If $\Omega'\subset\Omega$ is open and
\[
    f^{-1}(y)\cap\Omega\subset\Omega',
\]
then
\[
    \deg_{\mathrm B}(f,\Omega',y)
    =
    \deg_{\mathrm B}(f,\Omega,y).
\]
\end{enumerate}

We will also use the relation between the Brouwer degree and the mapping
degree of the normalized boundary map. Suppose that
$h:\overline{B_R}\to\RM^d$ is continuous and that
$y\notin h(\pdv B_R)$, then
\begin{equation}\label{eq:brouwer_boundary_formula}
    \deg_{\mathrm B}(h,B_R,y)
    =
    \deg\left(
        \frac{h-y}{|h-y|}
        \bigg|_{\pdv B_R}
    \right),
\end{equation}
where $\pdv B_R$ is identified with $\SM^{d-1}$, see
\cite[Section IV.4]{OutereloRuiz2009}. Indeed, the left-hand side depends only on the boundary values of $h$, hence it can be used as the primary definition of the mapping degree by extending a given normalized function on $\pdv B_R$ continuously to the interior.

\begin{lemma}\label{lem:brouwer_asymptotic_degree}
Let $f:\RM^d\to\RM^d$ be continuous and proper. For any $y\in \RM^d$ one has
\[
    \lim_{R\to \infty}\deg_{\mathrm B}(f,B_R,y)=\deg(f).
\]
\end{lemma}

\begin{proof}
Properness of $f$ is equivalent to $|f(x)|\longrightarrow\infty$ for $|x|\to \infty$. 
Fix $y\in\RM^d$ and choose $R$ sufficiently large that $\deg_R(f)$ has stabilized, that one has $f^{-1}(y)\subset B_R$ and $|f(x)|>|y|$ for every $|x|>R$. Then $\deg_{\mathrm B}(f,B_R,y)=\deg_{R}(f)$ because
$$t\in [0,1]\mapsto \frac{f-ty}{|f-ty|}\bigg\rvert_{\pdv B_R}$$
is a continuous homotopy.
\end{proof}

Our notion of degree is therefore nothing but the Brouwer degree of $f:\RM^d\to \RM^d$, which is well-defined since $f$ is proper.

\subsection{Degree for Half-Spaces}
\label{sec:deg_hs}
Let us discuss the special case of Dirac operators $D_{\Yy_1,\dots,\Yy_d}$ corresponding to the symbol function
$$f_{\Yy_1,\dots,\Yy_d}(x)= \left(\delta_{\Yy_1}(x),\dots,\delta_{\Yy_d}(x)\right)$$
with transverse half-spaces $\Yy_1,\dots,\Yy_d \subset \RM^d$.
Every class of transverse half-spaces $\Yy_i$ which can be suitably deformed to the standard half-spaces has degree $1$. 
To obtain collections of half-spaces with higher degrees, take $\Yy_i$ with distinct connected components. 

\begin{example}\label{ex:sph_caps}
    \rm{There exist half-spaces $\Yy_1,\dots, \Yy_d$ with any global degree $n \in \ZM$. Indeed, note that if each half-space is a disjoint union $\Yy_i=\Yy_i^{(1)}\cup \dots \cup \Yy_i^{(N)}$ such that all $\Yy_1^{(j)},\dots \Yy_d^{(j)}$ occupy disjoint angular sectors $\Omega_j\subset \SM^{d-1}$, then by additivity and excision of the Brouwer degree,
    \[\deg(f_{\Yy_1,\dots,\Yy_d}, \RM^d,0)=\sum_{j=1}^N \deg(f_{\Yy^{(j)}_1,\dots,\Yy^{j)}_d}, \RM_+\Omega_j,0)=\sum_{j=1}^N \deg(f_{\Yy^{(j)}_1,\dots,\Yy^{(j)}_d}, \RM^d,0).\]
    Hence, it is enough to construct examples of half-spaces with degree $\pm 1$ which are supported in small angular sectors. Let $\HM_1,\dots,\HM_d$ be the standard half-spaces and let $\Omega_1 \subset \SM^{d-1}$ be an open set. There is a Möbius transformation $M: \SM^{d-1} \to \SM^{d-1}$ which, after radial extension, maps $\HM_1,\dots,\HM_d$ into cones $\Yy_1^{(1)},\dots,\Yy_d^{(1)}$ over $\Omega_1$. If the transformation is orientation-preserving then $f_{\Yy_1,\dots,\Yy_d}$ is degree $1$, while if it is orientation-reversing then $f_{\Yy_1,\dots,\Yy_d}$ is degree $-1$. Rotating or reflecting this example generates disjoint configurations which add up to any degree $n$.}
\end{example}

Finally, we show that the notion of intersection number defined in \cite{drouot2024bulk} is equivalent to the Brouwer degree, hence they are equivalent.

\begin{figure}
    \centering
    \definecolor{Ugreen}{RGB}{40,165,90}
    \definecolor{Vpurple}{RGB}{126,67,171}
    \definecolor{dashred}{RGB}{250,82,65}
    \begin{tikzpicture}[scale=0.85]
  \clip (-7,-5.35) rectangle (7,5.35);
        \def\R{20}
        \def\rword{1.52}
        \draw[step=1, gray!25, very thin]
            (-7.1,-7.1) grid (7.1,7.1);
        \draw[gray!55, thin] (-6.1,0) -- (6.1,0);
        \draw[gray!55, thin] (0,-6.1) -- (0,6.1);
        \foreach \x in {-7,...,7}{
            \foreach \y in {-7,...,7}{
                \fill[black!75] (\x,\y) circle[radius=1.25pt];
            }
        }
        \foreach \a/\b in {20/70,130/170,220/260}{
            \path[
                fill=Ugreen,
                fill opacity=.32
            ]
                (0,0)
                -- (\a:\R)
                arc[start angle=\a,end angle=\b,radius=\R]
                -- cycle;

            \draw[Ugreen,line width=.9pt]
                (\a:\R) -- (0,0) -- (\b:\R);
        }
        \foreach \a/\b in {50/110,190/230,310/360}{
            \path[
                fill=Vpurple,
                fill opacity=.30
            ]
                (0,0)
                -- (\a:\R)
                arc[start angle=\a,end angle=\b,radius=\R]
                -- cycle;
            \draw[Vpurple,line width=.9pt]
                (\a:\R) -- (0,0) -- (\b:\R);
        }
        \draw[black,line width=1.2pt]
            (0,0) circle[radius=\R];
        \tikzset{
            region label/.style={
                fill=white,
                fill opacity=.70,
                text opacity=1,
                rounded corners=1.5pt,
                inner sep=2pt,
                font=\Large
            }
        }
        \node[region label,text=Ugreen]
            at ( 2.35, 2.45) {$U$};
        \node[region label,text=Ugreen]
            at (-2.95, 1.65) {$U$};
        \node[region label,text=Ugreen]
            at (-1.70,-2.85) {$U$};
        \node[region label,text=Vpurple]
            at ( .75, 4.45) {$V$};
        \node[region label,text=Vpurple]
            at (-3.95,-2.35) {$V$};
        \node[region label,text=Vpurple]
            at ( 4.10,-2.20) {$V$};
        \draw[
            dashred,
            densely dashed,
            line width=1.15pt
        ]
            (0,0) circle[radius=\rword];
        \tikzset{
            cross label/.style={
                fill=white,
                fill opacity=.72,
                text opacity=1,
                rounded corners=1pt,
                inner sep=1.2pt,
                font=\small
            }
        }
        \fill[black] (0,0) circle[radius=2.2pt];
    \end{tikzpicture}
    \caption{
        A pair of half-spaces $U,V$ with Brouwer degree $0$. Since the crossings at the origin are on neither simple nor transverse it needs some regularization before one can compute the degree as the intersection number, but the topological definition still applies.
        The Brouwer degree agrees with the intersection number taken by counting oriented intersections along a path traversing the boundaries $\pdv V$.
    }
    \label{fig:radial}
\end{figure}
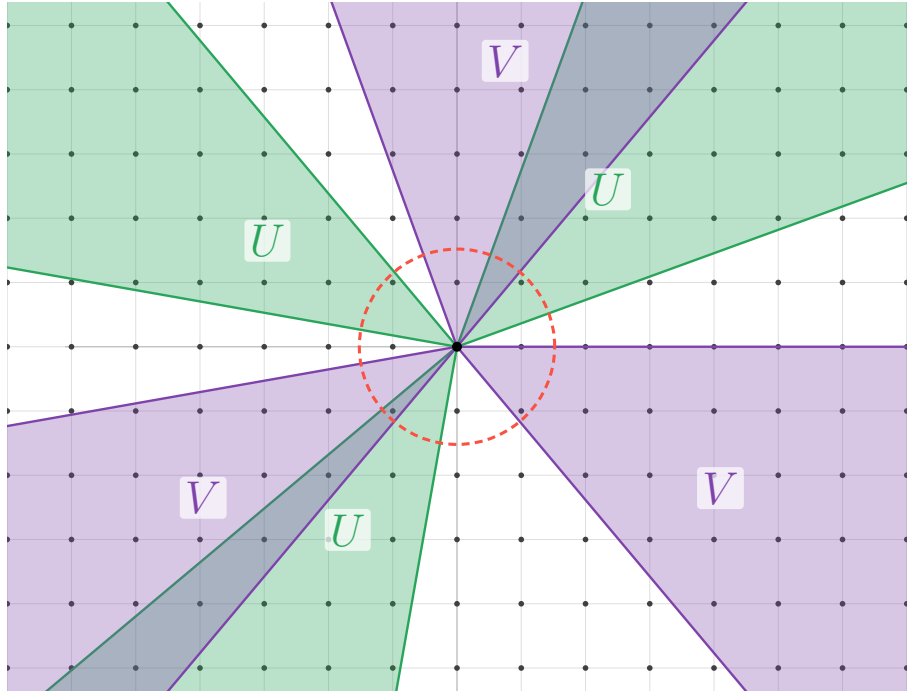

\begin{proposition}\label{prop:intersection}
Let $\Xx=\RM^2$, $\gamma_U,\gamma_V:\RM^2 \to \RM$ be two simple piecewise smooth curves for which $U,V$ are the half-spaces to the left of $\gamma_U$ and $\gamma_V$ respectively. We assume that $U,V$ are coarsely transverse, that $\gamma_U$ and $\gamma_V$ have finitely many intersections and that all intersections are transverse. Define the intersection number as in \cite{drouot2024bulk} as
\[\chi_{U,V} = \chi_{+}(U,V) - \chi_{-}(U,V); \qquad \chi_{\pm}(U,V) = \lim_{t \to \pm \infty} \mathds{1}_V \circ \gamma(t)\]
for one connected component, and a sum over connected components in general.
Then 
\begin{equation}\label{eqn:deg_intersection}
    \chi_{U,V} = -\deg(\delta_U, \delta_V).
\end{equation}

\end{proposition}

\begin{proof}
Let us first consider the contribution to the Brouwer degree by a single local intersection at point $x \in \RM^2$.
Choosing $R$ large enough so that $\pdv B_0(R)$ does not contain this intersection, $\delta \neq (0,0)$ and hence the map $\frac{\delta}{|\delta|}: \SM_R \to \SM^1$ is well-defined, with degree given by the local Brouwer degree $\deg_{R}((\delta_U,\delta_V), \SM_R, x)$.
The two transverse curves divide $\pdv B_0(R)$ into four arcs, on which the signs of $\delta_U$ and $\delta_V$ are respectively
\[(+,+),(+,-),(-,+),(-,-).\]
Consequently, as $\pdv B_0(R)$ is traversed once counterclockwise, the image of $\frac{\delta}{|\delta|}$ passes through the corresponding four quadrants of $\RM^2$. 
There are only two possible cyclic orders:
in one case the image winds once counterclockwise around the origin, and in the other it winds once clockwise. The the local degree is $\pm 1$ depending on whether the orientation matches that of the standard half-spaces.

Finally, there are four cases to check for \eqref{eqn:deg_intersection}. If $\gamma$ starts in $V$ and ends in $V^c$ then the first intersection exits $V$ and has degree $1$. If $\gamma$ enters $V$ again, it must eventually leave it again, hence there can only be any even number of additional intersections which alternate in sign and hence cancel. The remaining cases are analogous.
\end{proof}

\section{Kubo formula}
\label{sec:kubo}
In this section we prove the bulk and edge Kubo formulas, numerical formulas for the bulk and edge invariants, in a broad setting.
To begin, we define the central algebra on which our Kubo cocycles are defined.

\subsection{Schwartz algebras}
\label{sec:poly}

Given $\Xx$ with bounded geometry, there always exists a coarsely dense uniformly discrete subset $\Lambda\subset \Xx$ and a corresponding uniformly bounded Borel partition $\Xx= \bigsqcup_{\lambda \in \Lambda} Q_\lambda$ \cite[Remark 8.2]{Spakula2009}, i.e., one can construct partitions with $\lambda\in Q_\lambda$ and bounded diameter $\sup_{\lambda}\diam(Q_{\lambda}) < \infty$. 

Corresponding to the Borel partition, we get projections  
$$\chi_\lambda := \chi_{Q_\lambda}, \qquad \Hh_\lambda = \chi_\lambda \Hh_\Xx$$
which yields a direct sum decomposition
$$\Hh_\Xx = \bigoplus_{\lambda\in \Lambda} \chi_{\lambda}\Hh_\Xx= \bigoplus_{\lambda\in \Lambda} \Hh_\lambda.$$
Without loss of generality, for a coarsely faithful geometric module, we assume that all $\Hh_{\lambda}$ are nonzero. This can always be achieved by adapting a given partition, for example by combining cells. 
Since none of the main results will explicitly depend on the choice of $\Lambda, Q_\lambda$ we omit them from the notation and consider them fixed for any given module $\Hh_\Xx$.

An important property of the uniform Roe algebra, and the main reason why we prefer it over its non-uniform counterpart, is that in the polynomial-growth case it has a dense Fr\'echet subalgebra of polynomially controlled locally trace-class operators. 
\begin{definition}
Let $\Xx$ be a proper metric space with polynomial growth and a Borel partition as above.
\begin{enumerate}
\item[(i)]
The uniform Schwartz algebra $\Ss^p_u(\Xx, \Hh_\Xx)$, $1\leq p\leq \infty$, is the Fr\'echet subalgebra of $C^*_{u}(\Xx,\Hh_\Xx)$ defined by locally $p$-Schatten operators with matrix elements whose off-diagonal decay is faster than every polynomial.
Precisely, $a\in C^*_u(\Xx,\Hh_\Xx)$ is in $\Ss_u^p(\Xx,\Hh_\Xx)$ if all of the seminorms
$$\|a\|_{k,p} := \sup_{\lambda,\mu \in \Lambda} \|\chi_\mu a\chi_\lambda \|_{S^p(\Hh_\lambda, \Hh_\mu)} \langle d(\lambda,\mu)\rangle^k <\infty$$
are finite for each $k \in \NM$, with $\lVert\cdot\rVert_{S^p(\Hh_{\lambda},\Hh_{\mu})}$ the $p$-Schatten-norm for operators from $\Hh_{\lambda}$ to $\Hh_{\mu}$.
\item[(ii)] 
For $\Zz\subset \Xx$ a closed subset, define the relative Schwartz algebra $\Ss_u^p(\Zz\subset \Xx,\Hh_\Xx)$ as the Fr\'echet subalgebra of $C^*_{u}(\Zz \subset \Xx,\Hh_\Xx)$ consisting of elements $a$ such that all of the seminorms
$$\|a\|_{\pdv,k,p} := \sup_{\lambda,\mu \in \Lambda} \|\chi_\mu a\chi_\lambda\|_{S^p(\Hh_\lambda, \Hh_\mu)} \langle d(\lambda,\mu)\rangle^k \langle d(\lambda,\Zz)\rangle^k \langle d(\mu,\Zz)\rangle^k$$
are finite for each $k$. 
\end{enumerate}
\end{definition}

We note that different Borel partitions give rise to equivalent seminorms; we omit the simple proof since we never use this fact.
For $p=1$, we simply write $\|a\|_{\pdv,k} \coloneqq \|a\|_{\pdv,k,1}$.

In Appendix~\ref{sec:schwartz}, we prove that $\Ss_u^p(\Xx,\Hh_\Xx)$, $\Ss_u^p(\Zz\subset \Xx,\Hh_\Xx)$, and their stabilizations are dense Fr\'echet $\ast$-algebras which are closed under holomorphic functional calculus of their $C^*$-closures. This will be essential for defining the pairings with cyclic cocycles as we recall in Section~\ref{subsec:cocycle_definitions}.

In the following, we also need stabilizations of $C^*$-algebras given by tensoring the compact operators $\KM=\KM(\ell^2(\NM))$ on $\ell^2(\NM)$. The corner embedding of an algebra into its stabilization is a strong Morita-equivalence and therefore induces isomorphisms in both K-theory and $K$-homology. 
More generally, these isomorphisms also pass to limits over separable subalgebras, hence $KK^{\mathrm{sep}}(A,B)\simeq KK^{\mathrm{sep}}(A\otimes \KM, B \otimes \KM)$. It will be convenient that stabilization takes uniform Roe $C^*$-algebras to uniform Roe $C^*$-algebras: Note that $\Hh_\Xx\otimes \ell^2(\NM)$ is a geometric $\Xx\times\NM$-module, where we use on $\Xx\times \NM$ the product metric.
Then we can identify the stabilized uniform Roe $C^*$-algebras via
\begin{align*}
    C^*_{us}(\Xx, \Hh_\Xx)&= C^*_u(\Xx, \Hh_\Xx)\otimes \KM \\
    &\simeq C^*_u( (\Xx \times \{0\})\subset (\Xx\times \NM), \Hh_\Xx \otimes \ell^2(\NM))
\end{align*}
and
\begin{align*}
C^*_{us}(\Zz\subset \Xx, \Hh_\Xx) &= C^*_{u}(\Zz\subset \Xx, \Hh_\Xx)\otimes \KM \\
&\simeq C^*_{u}( (\Zz\times \{0\})\subset (\Xx\times \NM), \Hh_\Xx \otimes \ell^2(\NM)).
\end{align*}
The Schwartz subalgebras $\Ss^p_{us}(\Xx, \Hh_\Xx)$ and $\Ss^p_{us}(\Zz\subset \Xx, \Hh_\Xx)$ are then defined as the subalgebras corresponding to the relative Schwartz algebras of the right hand-sides. For the identification, we think of $\Xx\times \NM$ as infinitely many identical sheets with the original copy of $\Xx$ embedded as $\Xx\times \{0\}$.

\subsection{Kubo cocycles}\label{subsec:cocycle_definitions}

In this section, we lay the necessary groundwork for the Bulk and Edge Kubo formulas.
In particular, we define (polynomially bounded) coarse and cyclic cochains, explain their identification, and demonstrate that they are robust against arbitrary bounded perturbations in the defining functions $F$ in the sense that they do not change the cohomology class.

\begin{definition}
A coarse $d$-cochain $\varphi: \Xx^{d+1}\to \CM$ is a locally bounded Borel function such that for each $R>0$, the set
$$\mathrm{supp}(\varphi)\cap (\Delta)_R$$
is bounded where $(\Delta)_R$ is an $R$-thickening of the diagonal of $\Xx^{d+1}$.

For $\Zz\subset \Xx$ a closed subset, a coarse $d$-cochain relative to $\Zz$ is a locally bounded Borel function $\varphi: \Xx^{d+1}\to \CM$ such that for each $R>0$, the set
$$\mathrm{supp}(\varphi)\cap (\Delta)_R \cap (\Zz^{d+1})_R$$
is bounded.
\end{definition}
Coarse cohomology is the cohomology of the vector space of coarse cochains with respect to the Alexander-Spanier coboundary map
$$(\delta\varphi)(x_0,\dots,x_{d+1})=\sum_{i=0}^{d+1} (-1)^i \varphi(x_0,\dots,\widehat{x_i},\dots,x_{d+1}).$$
For cochains constructed by anti-symmetrization of functions $f_i:\Xx\to \CM$, this takes the form
$$\delta(f_0\wedge \dots \wedge f_d)= 1\wedge f_0\wedge \dots \wedge f_d.$$
One may restrict to antisymmetric cocycles without changing the cohomology theory, as antisymmetrization yields a quasi-isomorphic subcomplex.

\begin{definition}\cite[Section 4.2]{Roe1993}
\label{def:cyclic_cocycle_def}
Let $a_0,\dots,a_d \in \CM^1_u(\Zz\subset \Xx, \Hh_\Xx):= \CM_u(\Zz\subset \Xx, \Hh_\Xx)\cap \Ss_u^1(\Zz\subset \Xx, \Hh_\Xx)$, the algebra of locally trace-class operators with finite propagation supported near $\Zz$. There is a unique locally finite complex Radon measure $\mu_a$ on $\Xx^{d+1}$ with
$$\int (f_0\otimes \dots \otimes f_d) d\mu_a = \Tr(f_0 a_0 \dots f_d a_d)$$
for all $f_0,\dots,f_d\in C_c(\Xx)$.

For $\varphi$ an antisymmetric relative coarse $d$-cocycle a cyclic cocycle $\eta_{\varphi}: \CM_u^1(\Zz\subset \Xx)^{d+1} \to \CM$ is defined by
\begin{equation}
\label{eq:cyclic_cocycle_def}
\eta_{\varphi}(a_0,\dots,a_d)= \int_{\Xx^{d+1}} \varphi(x) d \mu_a(x).
\end{equation}
\end{definition}
The integral \eqref{eq:cyclic_cocycle_def} always exists for finite-propagation operators as only a compact subset of $\Xx^{d+1}$ contributes by the cochain condition \cite[Corollary 4.18]{Roe1993}.

\begin{example}\label{ex:cyclic_cocycle_discrete}
In the discrete case $\Hh_\Xx=\ell^2(\Lambda)$ we recover the explicit formula
$$\eta_{\varphi}(a_0,\dots,a_d)= \sum_{\substack{x_0,\dots,x_{d+1} \in \Lambda\\ x_0=x_{d+1}}} \varphi(x_0,\dots,x_d) \prod_{i=0}^d \langle \delta_{x_i}, a_i \delta_{x_{i+1}}\rangle.$$
If $\Xx=\RM^d$ and $\Hh_\Xx=L^2(\RM^d)$ the sum is replaced by an integral and the matrix elements by the integral kernels $a_i(x, y)$.
\end{example}

It is standard (see \cite{Roe1993}) that $\varphi\mapsto \eta_\varphi$ defines a map from coarse cohomology to the cyclic cohomology of $\CM_u^1(\Zz\subset\Xx,\Hh_\Xx)$. However, that algebra is not closed under holomorphic functional calculus, hence a cyclic cocycle on that algebra does not necessarily define a pairing with the K-groups $K_i(C^*_u(\Zz\subset \Xx))$ of the $C^*$-completion. This can be remedied under stronger assumptions:

\begin{definition}
We say that a coarse cochain $\varphi: \Xx^{d+1}\to \CM$ is $\Zz$-polynomially bounded if one of the seminorms
$$\|\varphi\|_{\Zz, k}:=\sum_{x_0,\dots,x_{d}\in \Lambda}  \|\varphi\|_{\ell^\infty(
Q_{x_0}\times\cdots\times Q_{x_d})} \prod_{i=0}^{d-1} \langle d(x_i, x_{i+1})\rangle^{-k}\prod_{i=0}^{d}  \langle d(x_i, \Zz)\rangle^{-k}$$
is finite for some $k>0$.
\end{definition}
Here, we consider the actual supremum norm in $\ell^{\infty}$, not the essential supremum.
\begin{proposition}
\label{prop:cont_ext}
If $\Xx$ has polynomial growth, then for any $\Zz$-polynomially bounded coarse cochain $\varphi$, the cyclic cochain $\eta_{\varphi}$ extends by continuity to a cyclic cochain on $\Ss_u^1(\Zz\subset \Xx)^{d+1}$.
\end{proposition}
\begin{proof}
Morally speaking, one wants to use $$|\Tr(f_0a_0\dots f_da_d)|\leq \prod_{i=0}^d \|f_i\|_\infty \prod_{i=0}^d
\|\chi_{x_i}a_i\chi_{x_{i+1}}\|_{S^1}= \|f_0\otimes\cdots \otimes f_d\|_\infty \prod_{i=0}^d
\|\chi_{x_i}a_i\chi_{x_{i+1}}\|_{S^1}$$ 
for bounded Borel functions $f_0,\ldots,f_d$ supported in $Q_{x_0},\ldots,Q_{x_d}$ to define a Radon measure, however, we require the estimate also for linear combinations of elementary tensors. This can be obtained with an algebraic trick:
Define $\Kk=\Hh_\Xx^{\otimes d+1}$ and the shift operator $S:\Kk\to \Kk$ defined by $S(\xi_0\otimes\cdots\otimes \xi_d)=\xi_1\otimes\cdots\otimes \xi_d\otimes \xi_0$.
Then
$$\sum_{j=1}^N \Tr(f_{0,j}a_0\dots f_{d,j}a_d)= \Tr(M_h (\chi_{x_0}a_0 \chi_{x_1}\otimes \cdots \otimes \chi_{x_d} a_d \chi_{x_0})S)$$
for $h=\sum_{j=1}^N f_{0,j}\otimes \cdots \otimes f_{d,j}$ with $M_h$ is the multiplication operator on $\Kk$. From this we conclude
$$\big|\sum_{j=1}^N \Tr(f_{0,j}a_0\dots f_{d,j}a_d)\big|= |\Tr(M_h (\chi_{x_0}a_0 \chi_{x_1}\otimes \cdots \otimes \chi_{x_d} a_d \chi_{x_0})S)|\leq \|h\|_\infty  \prod_{i=0}^d
\|\chi_{x_i}a_i\chi_{x_{i+1}}\|_{S^1}.$$
By the Riesz-Markov-Kakutani Theorem, $\mu_a$ therefore defines a locally finite Radon measure whose total variation in any cell satisfies
$$\big\|\mu_a\rvert_{Q_{x_0}\times \cdots \times Q_{x_d}}\|\leq \prod_{i=0}^d
\|\chi_{x_i}a_i\chi_{x_{i+1}}\|_{S^1}.$$
Consequently,
$$\big|\eta_{\varphi}(a_0,\ldots,a_d)\big| \leq \sum_{\substack{x_0,\dots,x_{d+1} \in \Lambda\\ x_0=x_{d+1}}}\|\varphi\|_{\ell^\infty(
Q_{x_0}\times\cdots\times Q_{x_d})}  \prod_{i=0}^d \|\chi_{x_i} a_i \chi_{x_{i+1}}\|_{S^1}$$
and therefore 
$$|\eta_{\varphi}(a_0,\dots,a_d)|\leq \|\varphi\|_{\Zz,k} \prod_{i=0}^d \|a_i\|_{\pdv,k}$$
which implies $\eta_{\varphi}$ is bounded with respect to the Fr\'echet topology of $\Ss_u^1(\Zz\subset \Xx)$. 
\end{proof} The continuous extension is a chain map from coarse cohomology to the cyclic cohomology of $\Ss_u^1(\Zz\subset \Xx)$, but a coarse coboundary $\varphi$ need not be a coboundary in the cyclic cohomology of $\Ss_u^1(\Zz\subset \Xx)$ unless it is the coboundary of polynomially bounded cochain. 

\begin{definition}
Let $A$ be a $C^*$-algebra and $\Aa \subset A$ a dense subalgebra which is closed under holomorphic functional calculus, hence the $K$-groups are identified. The pairing of cyclic cocycle $\eta: \Aa^{d+1}\to \CM$ with the $K$-theory of $\Aa$ is defined as follows:

For even $d$, a class $x\in K_0(A)$ admits representatives $x=[e]_0-[s(e)]_0$ with projections $e\in M_N(\Aa^+)$ that have scalar part $s(e)\in M_N(\CM)$. Then 
$$\langle x, [\eta] \rangle = (\eta\# \Tr)(e,\dots,e).$$

For odd $d$ the pairing with class $x\in K_1(A)$ is defined by representing $x=[u]_1$
by a unitary $u\in M_N(\Aa^+)$ with $s(u)=\one_N$, and
$$\langle x, [\eta] \rangle = (\eta\# \Tr)(u^*,u,u^*,\dots,u).$$
\end{definition}
Here, $\eta\#\Tr$ is the natural amplification of $\eta$ to matrix-valued operators and we tacitly use the normalized extension to the unitization, $\eta(a_0,\ldots,a_d)=\eta(a_0-s(a_0),\ldots,a_d-s(a_d))$ which eliminates all scalar parts.

We can now introduce the generalized Kubo cocycles. As explained in \cite{LudewigThiang25}, the Kubo formula corresponds to a specific kind of coarse cocycle. We generalize this by phrasing it in terms of Dirac-type functions:
\begin{definition}
\label{def:F}
Let $f:\Xx\to \RM^d$ be a $\Zz$-Dirac-type function and let $\Theta: \RM\to [-1,1]$ be a nondecreasing function with $\Theta(\lambda)=1$ for $\lambda > r$ and $\Theta(\lambda)=-1$ for $\lambda< -r$. We define $F=(F_1,\ldots,F_d)$ by $F_i(x)=\Theta(f_i(x))$, and the coarse cocycle $$\varphi_F = 1 \wedge F_1 \wedge \dots \wedge F_d.$$
\end{definition}
Among other properties, the $\Zz$-properness implies this is a cochain and it is closed by definition. We begin with a more general technical estimate:

\begin{lemma}
\label{lemma:poly_bounded_from_support}
Let $\nu$ be the growth exponent of $\Xx$ from Definition~\ref{def:poly_growth} and let $\varphi: \Xx^{d+1}\to \RM$ be a bounded function. Define the $Q$-support $\supp_{Q}(\varphi)$ as the set of all $\mathbf x\in \Lambda^{d+1}$ such that the block $Q_{\mathbf x}=Q_{x_0}\times ... \times Q_{x_d}$ satisfies $Q_{\mathbf x}\cap \supp(\varphi)\neq \emptyset$.

Suppose that there are $x_*\in\Xx$ and a function
$p:\RM_+\to\RM_+$ satisfying $p(R)\leq c_p\langle R\rangle^\alpha$ such that
\[ \operatorname{supp}_Q(\varphi) \cap(\Delta)_R\cap(\Zz)_R^{d+1} \subset B_{p(R)}(x_*)^{d+1}, \qquad R>0.\]
Then $\varphi$ is a $\Zz$-polynomially bounded relative coarse
cochain. More precisely, for every integer $k>(d+1)\nu\alpha+1$
one has
\[ \|\varphi\|_{\Zz,k} \leq C_{k,p}\|\varphi\|_\infty,\]
where $C_{k,p}$ only depends on $k$, $p$, $\Xx$ and the partition $Q$.
\end{lemma}
\begin{proof}
Define the shells
\[E_R= \{(x_0,\dots,x_d): d(x_i, x_{i+1})<R, \; d(x_j, \Zz)< R\}\]
for $i\in \{0,\dots, d-1\}$ and $j \in \{0,\dots, d\}.$
Note in particular that $E_R \subset (\Delta)_{dR}\cap (\Zz)_R^{d+1}$. For any $x\in E_N\setminus E_{N-1}$ one has
$$\prod_{i=0}^{d-1}
\langle d(x_i,x_{i+1})\rangle^{-k}
 \prod_{i=0}^{d}
 \langle d(x_i,\Zz)\rangle^{-k}
\leq C_k\langle N\rangle^{-k}.$$
as at least one factor is smaller than the right-hand side. 
Therefore, 
\begin{align*}
\|\varphi\|_{\Zz,k} &= \sum_{\substack{x_0,\dots,x_{d}\in \Lambda}} \|\varphi\|_{\ell^\infty(
Q_{x_0}\times\cdots\times Q_{x_d})}  \prod_{i=0}^{d-1} \langle d(x_i, x_{i+1})\rangle^{-k} \prod_{i=0}^{d}\langle d(x_i, \Zz)\rangle^{-k} \\
&\leq \|\varphi\|_\infty \sum_{N\geq 1} \sum_{\substack{\mathbf x \in E_N\setminus E_{N-1} \\ Q_{\mathbf x}\cap  \mathrm{supp}(\varphi)\neq \varnothing}}   C_k\langle N\rangle^{-k} 
\end{align*}

Using the bound on the $Q$-support and polynomial growth we can count those cells
\begin{align*}
\|\varphi\|_{\Zz,k}
&\leq
C_k\|\varphi\|_\infty
\sum_{N\geq1}
\#\Bigl\{
(x_0,\ldots,x_d)\in\Lambda^{d+1}: x_i \in B_{p(dN)}(x_*) \text{ for all }i
\Bigr\}
\langle N\rangle^{-k} \\
&\leq
C_k c^{d+1}\|\varphi\|_\infty
\sum_{N\geq1}
\langle p(dN)\rangle^{\nu(d+1)}
\langle N\rangle^{-k}.
\end{align*}
The final series converges under the stated condition for $k$.

\end{proof}

\begin{proposition}\label{prop:poly_proper_bounded}
Suppose $f:\Xx \to \RM^d$ is a polynomially proper $\Zz$-Dirac-type function with constants $c,C,L, s$ as in Definition~\ref{def:edge_dirac_type}. For $\varphi_F$ as in Definition~\ref{def:F} and for fixed $x_*$, the norms $\|\varphi_F\|_{\Zz,k}$ for sufficiently large $k$ can be bounded uniformly in terms of $r, L,d, c, C, s$ alone. In particular, the sets $B_R$ for Lemma~\ref{lemma:poly_bounded_from_support} can be chosen as balls $B_{p(R)}(x_*)$ with polynomial $p$ determined by as a continuous function in $r,L,d,c,C$ and $s$.
\end{proposition}
\begin{proof}
After writing the character $\varphi_F$ as a determinant and performing a row reduction, it can be written as
$$\varphi_{F}(x_0,\dots,x_d)= \sum_{\sigma\in S_d} (-1)^\sigma \prod_{i=1}^{d} \left(F_{\sigma(i)}(x_i)-F_{\sigma(i)}(x_{i-1})\right)=: \sum_{\sigma\in S_d} (-1)^\sigma \varphi_\sigma(x_0,\ldots,x_d).$$
We show that each term $\varphi_\sigma$ is polynomially bounded. As before, define
$$E_R=\{(x_0,\dots,x_d): d(x_i, x_{i+1})<R, \; d(x_i, \Zz)< R\}.$$
Then for $(x_0,\dots,x_d) \in E_R$, one has $(x_0,\dots,x_d)\in (\Delta)_{dR}\cap (\Zz)_R^{d+1}$.

Define $H_i^\pm = F_i^{-1}(\pm 1)$. If $\varphi_\sigma(x_0,\ldots,x_d)\neq 0$ then $x_i$ and $x_{i-1}$ cannot both be in $H_{\sigma(i)}^+$ or both be in $H_{\sigma(i)}^-$. For every $i$ there are therefore indices $a,b$ such that
$$f_i(x_a) \geq -r, \qquad f_i(x_b) \leq r.$$
For any $i$ and $k$ we can therefore estimate
$$f_i(x_k)\geq -r - L(1+\, \mathrm{dist}(x_k,x_a)), \qquad f_i(x_k)\leq r + L(1+\, \mathrm{dist}(x_k,x_b)).$$
Assuming now $x\in (\Delta)_{dR}\cap (\Zz)_R \cap \supp(\varphi_\sigma)$ the mutual distances are bounded by $2dR$, hence
$$|f_i(x_k)| \leq r + L + 2dRL.$$
On the other hand, by the $\Zz$-polynomial properness of $f$ we have
\[ \sqrt{C} d(x_*,x_k)^s \leq |f(x_k)| + R \leq \sqrt{d}(r+L+2dRL)+R,\]
whenever $d(x_*,x_k) > c$. For every $x_k$, including those with $d(x_*,x_k)\leq c$, we therefore have
$$d(x_*,x_k)\leq c+ \left(\frac{(\sqrt{d}(r+L+2dRL) + R)}{\sqrt{C}}\right)^{1/s} =: p(R).$$
Since each $x_k$ is in the ball $B_{p(R)}(x_*)$ and the blocks $Q_\lambda$ have bounded diameter the $Q$-support is bounded as required for Lemma~\ref{lemma:poly_bounded_from_support}, which finishes the proof.
\end{proof}

The pairing with these Kubo cocycles reproduces the familiar trace formulas:
\begin{lemma}\label{lem:cyclic_convergence}
The cocycle $\eta_F$ corresponding to a $\Zz$-polynomially proper Dirac-type function can be written as
$$\eta_F(a_0,\dots,a_d) = \sum_{\sigma\in S_{d}}(-1)^\sigma \Tr(a_{0}\prod_{i=1}^{d} [F_{\sigma{(i)}}, a_i]).$$
\end{lemma}
\begin{proof}
We insert partitions of unity and use the identity
\begin{equation}
\label{eq:commutater_expanded}
\chi_{x_1} [F_i, a] \chi_{x_2} = \chi_{x_1} F_i a \chi_{x_2} -  \chi_{x_1} a F_i \chi_{x_2}.
\end{equation}
For $a_0,\ldots,a_d\in \CM_u^1(\Zz\subset\Xx)$, we can formally expand by definition of $\mu_a$.
Since all sums are finite,
\begin{align*}
\Tr\left(a_{0}\prod_{i=1}^{d} [F_{\sigma{(i)}}, a_i]\right) &=  \sum_{\substack{x_0,x_1,\dots,x_{d+1} \in \Lambda\\ x_0=x_{d+1}}}  \Tr\left(\chi_{x_0} a_0 \prod_{i=1}^{d} \left(\chi_{x_{i}}[F_{\sigma{(i)}}, a_i]\chi_{x_{i+1}}\right)\right)\\
&= \sum_{x_0,x_1,\dots, x_d \in \Lambda} \int_{Q_{x_0}\times \dots\times Q_{x_d}} \varphi_\sigma(x) d\mu_a(x)=\int_{\Xx^{d+1}} \varphi_\sigma(x) d\mu_a(x)
\end{align*}
with a polynomially bounded cochain $\varphi_\sigma$, and up to a cyclic permutation it is the same as defined in the proof of  Proposition~\ref{prop:poly_proper_bounded}.

It remains to show that $\prod_{i=1}^{d} [F_{\sigma{(i)}},a_i]$ is trace-class such that expansion actually makes sense and extends by continuity. We expand as blocks
\begin{equation}
\label{eq:product_expansion}
\prod_{i=1}^{d} [F_{\sigma{(i)}},a_i]= \sum_{\substack{x_0,\ldots,x_{d+1} \in \Lambda\\x_0=x_{d+1}}} \prod_{i=1}^{d} B_i(x)
\end{equation}
with $B_i(x)=\chi_{x_{i}}[F_{\sigma{(i)}},a_i] \chi_{x_{i+1}}$. If $F_{\sigma(i)}$ is constant on both $Q_{x_{i-1}}$ and $Q_{x_{i}}$ and takes the same value on both, then $B_i(x)=0$ by \eqref{eq:commutater_expanded}.

For one fixed permutation $\sigma\in S_d$ we can therefore consider a slightly more coarse version of the support argument that we used in Lemma~\ref{lemma:poly_bounded_from_support} and Proposition~\ref{prop:poly_proper_bounded}.
Writing \(H_j^\pm=F_j^{-1}(\{\pm1\})\), define 
$$\Omega_\sigma =\left\{(x_1,\ldots,x_{d+1})\in\Lambda^{d+1}: Q_{x_{i}}\cup Q_{x_{i+1}}\not\subset H_{\sigma(i)}^\pm \text{ for every }i \in \{1,\dots, d\}\right\}$$
and note that for $x\notin \Omega_\sigma$ at least one of $B_1(x),\ldots, B_d(x)$ vanishes.

Moreover, the argument of Proposition~\ref{prop:poly_proper_bounded} applies in the same way to \(\Omega_\sigma\): For $x\in E_R \cap \Omega_\sigma$, we have $x_0,\ldots,x_d\in B_{p(R)}(x_*)$. 
Hence, the shell-counting argument of Lemma~\ref{lemma:poly_bounded_from_support} gives, for sufficiently large \(k\),
$$C_{\sigma,k} \coloneqq  \sum_{x\in\Omega_\sigma} \prod_{i=1}^d \langle d(x_i,x_{i+1})\rangle^{-k} \prod_{i=0}^d \langle d(x_i,\Zz)\rangle^{-k} <\infty. $$
Therefore, \eqref{eq:product_expansion} converges absolutely in trace norm with
\begin{align*}
\left\|\prod_{i=1}^d[F_{\sigma(i)},a_i] \right\|_{S^1}\leq \sum_{x\in\Lambda^{d+1}} \|B_1(\mathbf x)\cdots B_d(\mathbf x)\|_{S^1} 
&\leq 2^d C_{\sigma,k} \prod_{i=1}^d\|a_i\|_{\pdv,k}.
\end{align*}
\end{proof}

We also need sufficient conditions for when the cohomology class of this cocycle is invariant under perturbations of $f$.

\begin{proposition}\label{prop:cocycle_homotopy}
Let $t \in [0,1] \mapsto F^t$ be a differentiable path of Borel functions $F_1^t,\dots, F_d^t: \Xx\to \RM$ in the sense that $t \mapsto F^t(x)$ is continuously differentiable for every $x$.

Suppose that $\varphi_{F^t}$ is a $\Zz$-polynomially bounded coarse cochain for all $t$ and that 
$$A^t_j(x_0,\dots,x_{d-1}) = \dot{F}^t_j(x_0) (1\wedge F^t_1 \wedge \dots \widehat{F}^t_j \wedge \dots \wedge F^t_d)(x_0,\dots,x_{d-1})$$
is a $\Zz$-polynomially bounded coarse cochain for each $t$. Then $\dot{\varphi}_{F^t}$
is the coboundary of a $\Zz$-polynomially bounded cochain for every $t$.

Suppose in addition that the path is uniformly polynomially bounded in the sense that
\begin{equation}
\label{eq:unif_bounded_Atj}
\sup_{t\in [0,1]}\|\varphi_{F^t}\|_{\Zz, k} < \infty \hbox{ and } \sup_{t\in [0,1]} \|A^t_j\|_{\Zz, k} <\infty, 
\end{equation}
for each $j=1,\dots,d$ and some $k>0$ and that for every $R>0$ there exists some bounded set $B_R\subset \Xx$ with
\begin{equation}
\label{eq:supp_homotopy}
\mathrm{supp}(A_j^t)\cap (\Delta)_R \cap (\Zz^{d})_R \subset B_R^{d}
\end{equation}
for every $t\in [0,1]$ and $j=1,\dots,d$.
Then $\varphi_{F^t}-\varphi_{F^s}$ and $\eta_{F^t} - \eta_{F^s}$ are the coboundaries of polynomially bounded coarse and cyclic cochains respectively for every $s,t\in [0,1]$.
\end{proposition}
\begin{proof}
We can write the derivative of the character as
$$\frac{d}{dt}(1 \wedge F^t_1 \wedge \dots \wedge F^t_d)=\sum_{j=1}^d 1 \wedge F^t_1 \wedge \dots \wedge \dot{F}^t_j \wedge \dots \wedge F^t_d$$
by applying the Leibniz rule. Since none of the terms
$F^t_1 \wedge \dots \wedge \dot{F}^t_j \wedge \dots \wedge F^t_d$
are generally coarse cochains, we need to take advantage of cancellations.
Let $\hat{A}^t_j$ be the complete antisymmetrization of $A^t_j$.
Then
$$\delta \hat{A}^t_j= \frac{1}{d} 1\wedge \dot{F}^t_j \wedge F^t_1 \wedge \dots \wedge \widehat{F}^t_j \wedge \dots \wedge F^t_d.$$

Therefore,
$$\dot\varphi_{F^t} = \delta\left(d \sum_{j=1}^d (-1)^{j-1} \hat{A}_j^t\right),$$
which shows that $\dot\varphi_{F^t}$ is a coboundary, where the difference at time $t$ and $s$ is given by 
$$\varphi_{F^t}- \varphi_{F^s} = \delta \int_{s}^t d \sum_{j=1}^d (-1)^{j-1} \hat{A}^\lambda_j d\lambda.$$
Due to \eqref{eq:supp_homotopy}, the integral on the right-hand side is also a coarse cochain for every $s,t$, hence the left-hand side is a coboundary in coarse cohomology. Moreover, 
$$\bigg\lVert \int_{s}^t d \sum_{j=1}^d (-1)^{j-1} \hat{A}^\lambda_j d\lambda \bigg\rVert_{\Zz,k}\leq  \int_{s}^t d \sum_{j=1}^d \big\lVert\hat{A}^\lambda_j \big\rVert_{\Zz,k} d\lambda $$
where the right-hand side is finite if \eqref{eq:unif_bounded_Atj} holds.
\end{proof}

One sufficient condition which leaves the cohomology class invariant is as follows:
\begin{proposition}
\label{prop:coarse_bounded_perturbation}
Suppose $f$ is a polynomially proper $\Zz$-Dirac-type function with constants $L,c,C,s$ as in Definition~\ref{def:edge_dirac_type}. If $g: \Xx\to \RM^d$ is a bounded Borel function, then $\tilde{f}=f+g$ is also a polynomially proper $\Zz$-Dirac-type function with constants $\tilde{L},\tilde{c},\tilde{C},s$ bounded in terms of $L,c,C,s$ and $\|g\|_\infty$. 

For $F$ and $\Theta$ as in Definition~\ref{def:F} with  $\Theta \in C^{\infty}(\RM)$ then
$$\varphi_{F}-\varphi_{\tilde{F}}$$
is the coboundary of a polynomially bounded coarse cochain.
\end{proposition}
\begin{proof}
It is easy to see that for $t\in [0,1]$ all of the functions
$$f^t = f + t g$$
are polynomially proper $\Zz$-Dirac-type functions with uniform constants $\tilde{c},\tilde{C},\tilde{L},s$ that can be bounded in terms of $L,c,C,s$ and $\|g\|_\infty$.  
We prove that the transgression cochains $$A^t_j(x_0,\dots,x_{d-1}) = \dot{F}^t_j(x_0) (1\wedge F^t_1 \wedge \dots \widehat{F}^t_j \wedge \dots \wedge F^t_d)(x_0,\dots,x_{d-1})$$
satisfy the conditions of Proposition~\ref{prop:cocycle_homotopy}.
Recall that the derivative is given by
\[\dot{F}^t_j = {\Theta'}(f^t_j(x))\dot{f}^t_j(x)= {\Theta'}(f_j(x)+tg_j(x))g_j(x), \]
and in particular the derivative is a bounded function.

For $(x_0,\dots,x_{d-1})\in E_R \cap \supp({A^t_j})$ we must have $|f^t_j(x_0)| \leq r$. 
As before, for every $i\neq j$ there is a pair $x_a,x_b$ such that $x_a\in (H^-_i)^c$ and $x_b \in (H^+_i)^c$. Here, the $t$-dependence in $H^{\pm}$ is implicit.
For each $k$ and $i \neq j$ we once again obtain
$$|f_i^t(x_k)|\leq r + \tilde{L} + 2dR\tilde{L}$$
and for the remaining index $j$ one has $|f^t_j(x_0)| \leq r$. As in the proof of Proposition~\ref{prop:poly_proper_bounded} we conclude that
$$(\Zz)_R^d \cap (\Delta)_R \cap \supp (A^t_j) \subset B_{p(R)}^{d}$$
with a ball around $x_*$ whose radius is given as polynomial expression in $\tilde{c},\tilde{C},\tilde{L}$ and $s$. Proposition~\ref{prop:cocycle_homotopy} and Lemma~\ref{lemma:poly_bounded_from_support} finish the proof.
\end{proof}

\begin{corollary}
\label{cor:F_independence}
For given polynomially proper $\Zz$-Dirac-type function $f$, the class $[\eta_F]$ in cyclic cohomology does not depend on the choice of switch function.
\end{corollary}
\begin{proof}
Pick a switch function $\Theta$ which is smooth and hence has range exactly $[-1,1]$. For any other switch function $\tilde{\Theta}$ one can find a bounded Borel function $g$ such that $\tilde{\Theta}(f)=\Theta(f+g)$: Choose $R$ larger than the thresholds of both switch functions and let $\phi:[-1,1]\to [-R,R]$ be a Borel right inverse to $\Theta$, then one can set $$g_i(x)=\begin{cases}
    0 & \text{if }|f_i(x)|\geq R\\
    \phi(\tilde{\Theta}(f_i(x)))-f_i(x) & \text{if }|f_i(x)|< R.
\end{cases}$$
\end{proof}

\subsection{Kubo formulas with Finite Multiplicity}\label{sec:kubo_uf}

In this section we prove the bulk and edge Kubo formulas in a simplified setting. Let $f$ be a polynomially proper $\Zz$-Dirac-type function and let $F=(F_1,\ldots,F_d)$ be as in Definition~\ref{def:F} with $\Theta$
chosen smooth without loss of generality by Corollary~\ref{cor:F_independence}. 
The generalized Kubo formula asserts that, up to normalization, the pairing of a K-theory class with the cyclic-cohomology class $[\eta_F]$ agrees with its pairing with the K-homology class $[D_f]$.

The idea is to reduce this to a special case using functoriality of both pairings. Since general coarse maps do not give rise to homomorphisms of the uniform Roe algebra, we must use a more restricted notion:
\begin{definition}
Let $\Lambda_1,\Lambda_2$ be discrete metric spaces and $\Zz_1\subset \Lambda_1$, $\Zz_2\subset \Lambda_2$ closed subspaces. 

Suppose that $\psi: \Lambda_1\to \Lambda_2$ is large-scale Lipschitz and satisfies
\begin{equation}
\label{eq:boundary_distance}
d_{\Lambda_2}(\psi(x),\Zz_2)\leq C (1+d_{\Lambda_1}(x,\Zz_1)), \qquad \forall x\in \Lambda_1,
\end{equation}
for some $C>0$. If $\psi$ moreover has uniformly finite multiplicity in the sense that 
\[\sup_{y \in \Lambda_2} \# \psi^{-1}(y) \leq N_{\psi} < \infty\]
for some $N_\psi\in\mathbb N$, then we will call $\psi$ uniformly finite-to-one (with respect to $\Zz_1$ and $\Zz_2$).
\end{definition}

The following is a slight generalization of a construction that was recently described by \cite[Lemma 2.4]{krutoy2026bijective} using the formalism of covering isometries. In what follows, we abbreviate $C^*_u(\Lambda, \ell^2(\Lambda))$ by $C^*_u(\Lambda)$.
\begin{proposition}
\label{prop:cover_homomorphism}
Let $\Lambda_1,\Lambda_2$ be discrete metric spaces with polynomial growth and let $\phi:\Lambda_1\longrightarrow \Lambda_2$
be uniformly finite-to-one with respect to closed subsets $\Zz_1\subset \Lambda_1$ and $\Zz_2\subset \Lambda_2$. Then $\phi$ lifts to a large-scale Lipschitz injection
\[\widetilde\phi:\Lambda_1\times\mathbb N\longrightarrow \Lambda_2\times\mathbb N\]
and corresponding isometry 
\[V_\phi:\ell^2(\Lambda_1)\otimes\ell^2(\mathbb N)
\longrightarrow
\ell^2(\Lambda_2)\otimes\ell^2(\mathbb N).\]
It induces an injective $*$-homomorphism
$$\Psi_\phi:C^*_{us}(\Zz_1\subset \Lambda_1)\longrightarrow C^*_{us}(\Zz_2\subset \Lambda_2)$$
which maps the subalgebra $\Ss^p_{us}(\Zz_1\subset \Lambda_1)$ to $\Ss^p_{us}(\Zz_2\subset \Lambda_2)$ for any $p \in [1,\infty]$.
\end{proposition}
\begin{proof}
We give the construction and refer to \cite[Lemma~2.4]{krutoy2026bijective} for more details. 

By the finite-multiplicity assumption, each exact fiber
$\phi^{-1}(y)$ has cardinality at most $N_\phi$. For each
$y\in \Lambda_2$, choose an injection $\ell_y:\phi^{-1}(y)\longrightarrow \{0,\ldots,N_\phi-1\}$ which associates the label $\ell(x):=\ell_{\phi(x)}(x)$ to each $x$. 
We then define the injection
\begin{align*}
\widetilde\phi:\Lambda_1\times\mathbb N&\longrightarrow \Lambda_2\times\mathbb N \\
(x,n) &\longmapsto \big(\phi(x),\,N_\phi n+\ell(x)\big).
\end{align*}
One represents
$C^*_{us}(\Zz_i \subset \Lambda_i) = C^*_u(\Zz_i\subset \Lambda_i)\otimes \KM(\ell^2(\mathbb N))$ on the Hilbert space $\Hh_i:=\ell^2(\Lambda_i\times\NM)$ with standard bases $(\delta_{(x,n)})_{x\in \Lambda_i,n\in \NM}$.
The corresponding isometry $V_\phi:\Hh_1\longrightarrow \Hh_2$ is the injection on the standard bases
$V_\phi\delta_{(x,n)}= \delta_{\widetilde\phi(x,n)}.$  The lift $\tilde{\phi}$ is still large-scale Lipschitz and combined with \eqref{eq:boundary_distance} the adjoint action $\Psi_\phi: a\longmapsto V_\phi a V_\phi^*$ maps finite-propagation operators supported near $\Zz_1\times \{0\}$ to finite-propagation operators supported near $\Zz_2\times \{0\}$. Moreover, it is easily seen to be continuous with respect to the Schwartz seminorms, hence $\Psi_\phi$ defines a homomorphism from $\Ss^p_{us}(\Zz_1 \subset \Lambda_1)$ to $\Ss^p_{us}(\Zz_2 \subset \Lambda_2)$.
\end{proof}
\begin{remark}
{\rm Despite our notation, the construction of $\Psi_\phi$ is not canonical, since it depends on the choice of the injections $\ell_y$.
However, different choices are conjugate by a
unitary multiplier of $C^*_{us}(\Lambda_2)$, hence at least the induced maps $(\Psi_\phi)_*:K_i(\Ss^p_{us}(\Zz_1\subset \Lambda_1)) \to K_i(\Ss^p_{us}(\Zz_2\subset \Lambda_2))$ are all equal.} $\diamond$
\end{remark}

We now use a similar construction as in \cite{LudewigThiang25} to introduce a standardization map. In the case of signed distance functions, it maps a collection of
half-spaces to the standard coordinate half-spaces of $\ZM^d$.
\begin{definition}
For $x\in \RM^d$ define $[x]\in \ZM^d$ to be the component-wise closest integer point. For $\Xx$ a proper metric space with coarsely dense uniformly discrete subset $\Lambda$ define the standardization map
$$s_f: \Lambda \to \ZM^{d+1}, \qquad s_f(x)= \left[(f_1(x),\dots,f_d(x), d_{\Zz}(x))\right].$$
If the value is exactly halfway between two points, we adopt the convention that we round down. 
If $s_f$ is uniformly finite-to-one then it lifts to an injection $\tilde{s}_f \colon \Lambda \times \NM \to \ZM^{d+1} \times \NM$.
\end{definition}
In the bulk case $\Zz=\Xx$, the image of $s_f$ lies in the single slice $\ZM^{d}\times \{0\}\simeq \ZM^{d}$.

\begin{figure}[htbp]
\centering
\begin{tikzpicture}[scale=0.43]
\begin{scope}
\node at (0,6.7) {$\Xx$};
\fill[blue!20]
(-7,-3)
.. controls (-7,3) and (-9.5,7) ..
(2,6)
.. controls (1,2) and (-1,-1) ..
(-1,-5)
.. controls (-3,-6) and (-5,-6) ..
(-7,-5)
-- cycle;
\fill[red!20]
(-7,5)
.. controls (-3,-2) and (2,2) ..
(7,-4.8)
.. controls (5,-2.0) and (6,1.5) ..
(7,5.5)
.. controls (2,6.5) and (-3,6.2) ..
(-7,5)
-- cycle;
\begin{scope}
\clip
(-7,5)
.. controls (-3,-2) and (2,2) ..
(7,-4.8)
.. controls (5,-2.0) and (6,1.5) ..
(7,5.5)
.. controls (2,6.5) and (-3,6.2) ..
(-7,5)
-- cycle;
\fill[purple!30]
(-7,-3)
.. controls (-7,3) and (-9.5,7) ..
(2,6)
.. controls (1,2) and (-1,-1) ..
(-1,-5)
.. controls (-3,-6) and (-5,-6) ..
(-7,-5)
-- cycle;
\end{scope}
\foreach \x/\y in {
-6.2/4.6, -6.0/2.8, -5.8/0.9, -5.5/-1.2,
-5.2/-3.4, -4.8/5.0, -4.5/3.1, -4.2/1.0,
-3.9/-2.0, -3.6/-4.2, -3.2/4.3, -2.9/2.2,
-2.6/0.5, -2.3/-1.5, -2.0/-3.8, -1.7/5.2,
-1.4/3.0, -1.1/1.2, -0.8/-0.6, -0.5/-2.8,
-0.2/4.7, 0.1/2.5, 0.4/0.4, 0.7/-1.4,
1.0/-3.5, 1.3/5.5, 1.6/3.3, 1.9/1.1, 2.5/0.9,
2.2/-0.9, 2.5/-2.9, 2.8/4.2, 3.1/2.0,
3.4/0.3, 3.7/-1.6, 4.0/-3.6, 4.3/5.0,
4.6/2.7, 4.9/0.9, 5.2/-1.1, 5.5/-3.0, 5.3/-2.6
}{
\fill (\x,\y) circle (0.12);
}
\fill[black] (0.4,0.4) circle (0.18);
\node[blue] at (-4.1,4.8) {$\Yy_1$};
\node[red] at (4.5,4.0) {$\Yy_2$};
\node[purple!80!black] at (-1,4.0) {$\Yy_1 \cap \Yy_2$};
\node[black] at (1,0.0) {$x$};

\end{scope}
\draw[->,very thick] (8,0)--(16,0);
\node at (12,1.0)
{$s_f(x)=\bigl[\delta_{\Yy_1}(x),\delta_{\Yy_2}(x)\bigr]$};
\begin{scope}[shift={(23,0)}]
\node at (0,6) {$\mathbb{Z}^2$};
\fill[blue!20] (0,-5) rectangle (5,5);
\fill[red!20] (-5,0) rectangle (5,5);
\fill[purple!30] (0,0) rectangle (5,5);
\draw[thick,->] (-5,0)--(5,0);
\draw[thick,->] (0,-5)--(0,5);
\foreach \i in {-4,...,4}
{
  \foreach \j in {-4,...,4}
  {
    \fill (\i,\j) circle (0.08);
  }
}
\fill[black] (0,0) circle (0.18);
\node[blue] at (2.8,-3.5) {$\mathbb{H}_1$};
\node[red] at (-2.7,3.5) {$\mathbb{H}_2$};
\node[purple!80!black] at (2.8,2.8) {$\mathbb{H}_1\cap \mathbb{H}_2$};
\node[black] at (-1,0.5) {$s(x)$};
\end{scope}
\end{tikzpicture}
\caption{An example of the standardization map being used in the bulk case on half-spaces $\Yy_1,\Yy_2$ of $\Xx$ in two dimensions, after identifying the slice $\ZM^2 \times \{0\}$ with $\ZM^2$.
A point $x$ is depicted mapping to the origin in $\ZM^2$.}
\label{fig:coarse-sets}
\end{figure}
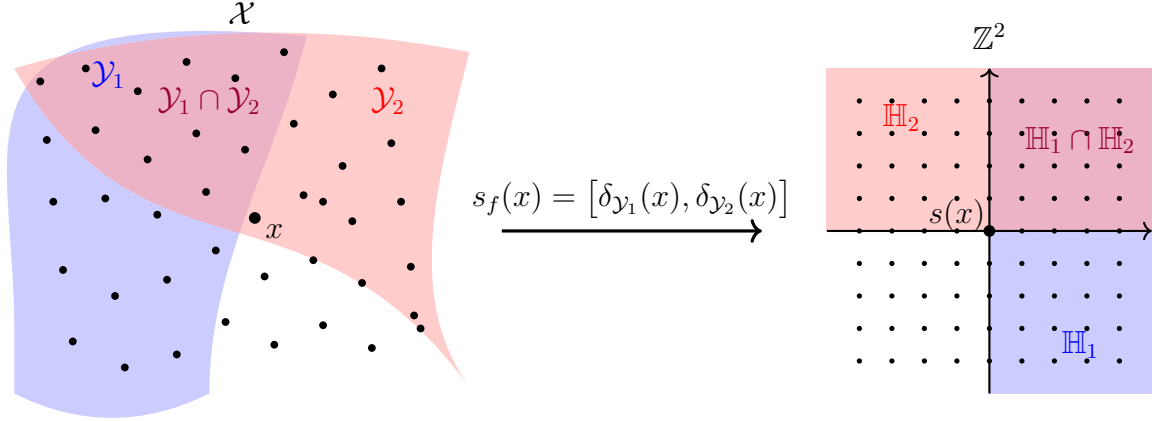

In the remainder of this subsection, we prove that the Kubo formula holds whenever $s_f$ is uniformly finite-to-one:
\begin{theorem}
\label{thm:kubo_uniformly_finite}
Let $\Xx$ be a proper metric space of polynomial
growth, $\Zz\subset \Xx$ a non-empty closed subset, and $\Hh_\Xx$ a coarsely faithful geometric module. Fix a coarsely dense uniformly discrete subset $\Lambda\subset\Xx$ and define the subset $\Lambda_\Zz \coloneqq \{\lambda \in \Lambda: Q_{\lambda} \cap \Zz \neq \varnothing \}$ which is coarsely dense in $\Zz$. 

For a polynomially proper $\Zz$-Dirac-type function $f$, if the standardization map $s_f$ is uniformly finite-to-one with respect to
\[\Lambda_\Zz \subset \Lambda, \qquad
\ZM^{d}\times\{0\}\subset\ZM^{d+1},\]
then 
$$\langle x, [D_{f}]\rangle = \Xi_{d}\, \langle x,[\eta_{F}]\rangle \in\ZM$$
for every $x\in K_{d\bmod 2}\bigl(C^*_u(\Zz\subset\Xx,\Hh_\Xx)\bigr)$ with $\Xi_d$ the normalization constant from
Appendix~\ref{sec:kubozd}.
\end{theorem}

\begin{proposition}
\label{prop:kubo_cts_discrete}
It suffices to prove
Theorem~\ref{thm:kubo_uniformly_finite}
in the special case $\Hh_\Xx=\ell^2(\Lambda)$. 
\end{proposition}
\begin{proof}
We consider the full-corner inclusion $$\psi: C^*_u(\Lambda_{\Zz} \subset \Lambda)\to C^*_u(\Zz \subset \Xx, \Hh_\Xx), \quad \psi(a)= \sum_{x,y\in \Lambda} |\xi_x\rangle \langle \delta_x, a\delta_y\rangle \langle \xi_y|$$
with any (fixed) unit vectors $\xi_\lambda \in \Hh_\lambda$. It induces an isomorphism in $K$-theory by Proposition~\ref{prop:morita_continuous}.

Also define the Dirac-type function
$${f}^\Lambda(x) = \sum_{\lambda\in \Lambda} f(\lambda)\chi(x\in Q_\lambda).$$
This is a bounded perturbation of $f$, hence $[D_{f^\Lambda}]=[D_f]$ and $[\eta_F]=[\eta_{F^\Lambda}]$ (by Proposition~\ref{prop:coarse_bounded_perturbation}). If we assume Theorem~\ref{thm:kubo_uniformly_finite} holds for $\Hh_\Xx=\ell^2(\Lambda)$ then
$$\langle \psi_*(x), [D_{f}]\rangle=\langle \psi_*(x), [D_{f^\Lambda}]\rangle = \Xi_d \langle \psi_*(x), [\eta_{F^\Lambda}]\rangle = \Xi_d \langle \psi_*(x), [\eta_F]\rangle.$$
Here we used (with some abuse of notation) that by the cellwise constant definition of $f^\Lambda$ one can equivalently consider it as a Dirac-type function on $\Xx$ or $\Lambda$. In that sense neither the index nor the cohomological pairing change under the corner inclusion since they evaluate to the same integer respectively complex number independent of the interpretation. Since $\psi_*$ is an isomorphism in $K$-theory every class in $K_{d\bmod 2}(C^*_u(\Zz\subset \Xx,\Hh_\Xx))$ can be represented in this form, which finishes the proof.
\end{proof}

The K-homology class of the standard Dirac operator transforms naturally under the standardization:

\begin{proposition}
\label{prop:standardization_dirac}
For $\Lambda$ a discrete metric space let $f:\Lambda\to \RM^d$ be a $\Lambda_\Zz$-Dirac-type function and let $\tilde{f}:\Lambda\times\NM\to \RM^d$ be the trivial extension $\tilde{f}(x,n)=f(x)$. If the standardization map $s_f$ is uniformly finite-to-one then
$$\langle (\Psi_{s_f})_*(x),[D_{\mathrm{std}}]\rangle =  \langle x,[D_{\tilde{f}}]\rangle, \qquad \forall x\in K_{d\bmod 2}(C^*_{us}(\Lambda_{\Zz}\subset \Lambda))$$
with $D_{\mathrm{std}}=D_{\mathbb{H}_1,\ldots\mathbb{H}_d}$ the standard edge Dirac operator on $\ell^2(\ZM^{d+1} \times \NM)$ given by the Dirac-type function $f_{\mathrm{std}}(x)=(x_1,\ldots,x_d)$.
\end{proposition}
\begin{proof}
It is enough to show 
$$
(\Psi_{s_f})^*[D_{\mathrm{std}}]= [D_{\tilde{f}}]
$$
for the pullback in KK-theory as the index pairing is a Kasparov product which behaves naturally with respect to pull-backs. By definition, $(\Psi_{s_f})^*[D_{\mathrm{std}}]$ 
is represented by the $KK_d(C_{us}^*(\Lambda_\Zz\subset \Lambda),\CM)$-cycle
$$  (\ell^2(\ZM^{d+1}\times \NM)\hat{\otimes}\CM_d,\pi_{\ZM^{d+1}} \circ \Psi_{s_f}, D_{\mathrm{std}})
$$
with $\pi_{\ZM^{d+1}}$ the defining representation of $C^*_u(\ZM^d\times\{0\} \subset \ZM^{d+1}\times \NM)$. Operators in the image of the representation $\pi_{\ZM^{d+1}} \circ \Psi_{s_f}$ only act non-trivially on the range of the isometry $V_{s_f}:\ell^2(\Lambda\times \NM)\to \ell^2(\ZM^{d+1}\times \NM)$, hence we can restrict the Hilbert space to that subspace. After unitary equivalence we therefore end up with the equivalent KK-cycle
$$(\ell^2(\Lambda\times \NM)\hat{\otimes}\CM_d,V^*_{s_f}(\pi_{\ZM^{d+1}} \circ \Psi_{s_f}) V_{s_f}, V^*_{s_f}D_{\mathrm{std}}V_{s_f}).$$
Moreover, $V_{s_f}^*(\pi_{\ZM^{d+1}} \circ \Psi_{s_f}) V_{s_f}(a)=a$ for every $C^*_{us}(\Lambda_{\Zz}\subset \Lambda)$.
Since $V^*_{s_f}D_{\mathrm{std}}V_{s_f} = D_{f^*_{\mathrm{std}}}$ and one can replace $f^*_{\mathrm{std}}$ by $\tilde{f}$ since they are bounded perturbations of each other
$$|f^*_i(x,n) - \tilde{f}_i(x,n)| =\| [f_i(x)]- f_i(x) \| \leq 1, \quad  \forall n\in \NM, i=1,\dots,d.$$  We conclude that the pullback KK-cycle is equal to $[D_{\tilde{f}}]$.
\end{proof}

The same is true for the Kubo cocycles:
\begin{proposition}
\label{prop:standardization_cocycle}
In the setting of Proposition~\ref{prop:standardization_dirac}, suppose further $f$ is $\Lambda_\Zz$-polynomially proper. 
Then
$$\langle (\Psi_{s_f})_*(x), [\tilde{\eta}_{\mathbb{H}_1,\dots,\mathbb{H}_d}] \rangle = \langle x,[\tilde{\eta}_{F}]\rangle, \qquad \forall x\in K_{d\bmod 2}(C^*_{us}(\Lambda_{\Zz}\subset \Lambda))$$
with the standard half-spaces $\mathbb{H}_1,\dots,\mathbb{H}_d$, $\mathbb{H}_i=\{x: x_i\geq 0\}$ of $\ZM^{d+1}$ and Kubo cocycles $\tilde{\eta}_{\mathbb{H}_1,..,\mathbb{H}_d}$ and $\tilde{\eta}_{F}$ induced by $f_{\mathrm{std}}$ and $f$ respectively on $\Ss^1_{us}((\ZM^d\times \{0\})\subset \ZM^{d+1})$ and $\Ss^1_{us}(\Zz\subset \Lambda)$ respectively. 

\end{proposition}
\begin{proof}
The pullback $f^*_{\mathrm{std}}: \Lambda\times \NM\to \RM^d$ is a polynomially proper $\Lambda_{\Zz}\times \{0\}$-Dirac-type function. It is easy to see that
$$\langle (\Psi_{s_f})_*(x), [\tilde{\eta}_{\mathbb{H}_1,\dots,\mathbb{H}_d}] \rangle = \langle x, [\tilde{\eta}_{F_{\mathrm{std}}^*}] \rangle$$
as one has the identity
$$\tilde{\eta}_{\mathbb{H}_1,\dots,\mathbb{H}_d}(\Psi_{s_f}(a_0),\cdots,\Psi_{s_f}(a_d))= \tilde{\eta}_{F_{\mathrm{std}}^*}(a_0,\cdots,a_d).$$
Finally, $f^*_{\mathrm{std}}$ and $\tilde{f}$ differ only by a bounded perturbation, so they define the same class in cyclic cohomology.
\end{proof}

\begin{proof}[Proof of Theorem~\ref{thm:kubo_uniformly_finite}]
By Proposition~\ref{prop:kubo_cts_discrete}, it suffices to consider the discrete case $\Xx=\Lambda$.
Then by Proposition~\ref{prop:standardization_dirac} and Proposition~\ref{prop:standardization_cocycle} it is enough to prove the Kubo formula for case of $\Lambda=\ZM^{d+1}$ and $\Zz=\ZM^{d}\times \{0\}$, since then
\[
\langle x,[D_f]\rangle=
\langle(\Psi_{s_f})_*(x),[D_{\mathrm{std}}]\rangle=
\Xi_d\left\langle (\Psi_{s_f})_*(x),[\widetilde\eta_{\mathbb H_1,\ldots,\mathbb H_d}]\right\rangle=
\Xi_d\langle x,[\widetilde\eta_F]\rangle.
\] Moreover, since $C^*_{u}((\ZM^{d}\times \{0\})\subset \ZM^{d+1})\simeq C^*_{us}(\ZM^{d})$ is just stabilization  it is now sufficient to prove the Kubo formula for the algebra $C^*_u(\ZM^d, \ell^2(\ZM^d))$ and the standard half-spaces of $\ZM^d$. We defer this to Appendix~\ref{sec:kubozd} where Theorem~\ref{th:kubozd} finishes the proof.
\end{proof}

\subsection{Polynomially Proper Kubo Formulas}

In this subsection we show that for $\Xx=\RM^d$ the condition that the standardization map is uniformly finite-to-one is much stronger than necessary. Instead, the Kubo formula holds for any polynomially proper Dirac-type function. The idea is that, using the homotopy invariance of both K-homology and cyclic cohomology, it is enough to prove the Kubo formula for one representative of each degree.

Throughout this subsection, we retain the convention that $F_i = \Theta \circ f_i$ and $\underline{F} = (F_1,\dots, F_{d-1})$ with $\Theta$ some fixed smooth switch function and $\supp \Theta' \subset [-r,r]$ for some $r >0$. The smoothness is again without loss of generality by Corollary~\ref{cor:F_independence}.

We then have a bulk cocycle $\eta_{F}$ on $\Ss_u^1(\Xx,\Hh_\Xx)$ and an edge cocycle $\eta_{\underline{F}}$ on $\Ss_u^1(\Zz_d\subset \Xx,\Hh_\Xx)$ which are defined by
$$\eta_{F}(a_0,\dots,a_d)= \sum_{\sigma\in S_{d}}(-1)^\sigma \Tr(a_{0}\prod_{i=1}^{d} [F_{\sigma(i)}, a_i])$$
and
$$\eta_{\underline{F}}(a_0,\dots,a_{d-1})= \sum_{\sigma\in S_{d-1}}(-1)^\sigma \Tr(a_{0}\prod_{i=1}^{d-1} [F_{\sigma(i)}, a_i]).$$

\begin{theorem}
\label{th:kubo_poly_proper}
Fix $d\geq 2$ and let $\Xx$ be coarsely equivalent to $\RM^d$ and $\Hh_\Xx$ a coarsely faithful geomeric module.
Let $f:\Xx\to \RM^d$ be a polynomially proper Dirac-type function.
\begin{enumerate}
    \item[(i)] The Kubo formula 
$$\langle x, [D_{f}]\rangle = \Xi_d \langle x, \eta_{F}\rangle$$ 
holds for every class $x\in K_{d\bmod 2}(C^*_u(\Xx, \Hh_\Xx))$.
    
    \item[(ii)] If $f$ is $\Yy_d$-adapted to some half-space $\Yy_d\subset \Xx$ with coarse boundary $\Zz_d$, then the edge Kubo formula
$$\langle \pdv_{\Yy_d}(x), [D_{\underline{f}}]\rangle=\Xi_{d-1} \langle  \pdv_{\Yy_d}(x) , \eta_{\underline{F}}\rangle$$
holds for every $x\in K_{d\bmod 2}(C^*_u(\Xx, \Hh_\Xx))$ with the boundary map \[\pdv_{\Yy_d}:  K_{d\bmod 2}(C^*_u(\Xx, \Hh_\Xx))\to  K_{(d-1) \bmod 2}(C^*_u(\Zz_d\subset \Xx, \Hh_\Xx)).\]
\end{enumerate}
\end{theorem}
Note that as a consequence, by Theorem~\ref{thm:degree_multiplicity} we have the analogous degree expression for the Kubo formulas for non-trivial half-spaces, given by 
\[\eta_{\Yy_1,\dots,\Yy_d}(e,\dots,e)=\deg(\delta_{\Yy_1},\dots,\delta_{\Yy_d} )\,\eta_{\mathbb{H}_1,\dots,\mathbb{H}_d}(e,\dots,e)\]
and 
\[\eta_{\Yy_1,\dots,\Yy_d}(u^*, u,\dots, u^*, u)=\deg(\delta_{\Yy_1},\dots,\delta_{\Yy_d} )\,\eta_{\mathbb{H}_1,\dots,\mathbb{H}_d}(u^*, u,\dots, u^*, u)\]
in even and odd dimensions respectively. We also note again the following simplification:

\begin{lemma}
\label{lemma:discrete_reduction}
It is enough to prove Theorem~\ref{th:kubo_poly_proper} for $\Xx=\Lambda\subset \RM^d$ any Delone set and $\Hh_\Xx=\ell^2(\Lambda)$.
\end{lemma}
\begin{proof}
This is the same argument as Proposition~\ref{prop:kubo_cts_discrete}: With Delone subsets $\Lambda$ of $\Xx$ and $\Lambda_{\Zz_d}\subset \Zz_d$ one can represent using Morita equivalence any K-theory class by a representatives in $\Ss^1_u(\Lambda)$ embedded into $C^*_u(\Xx,\Hh_\Xx)$. This is compatible with the boundary maps as it is a corner inclusion. If $f$ is cellwise constant the corner inclusion also preserves both sides of the Kubo formula. Finally, any two Delone sets are coarsely equivalent, so it is enough to prove it for one example.
\end{proof}
Our main tool in the proof of this theorem will be to deform both the half-space $\Yy_d$ and its boundary $\Zz_d$ via homotopy some standard example for which the standardization map has uniform finite multiplicity.
Although the spherical caps in Example~\ref{ex:sph_caps} are not always uniformly finite-to-one, there are Dirac-type functions whose underlying half-spaces are spherical caps which are:
\begin{lemma}
\label{lemma:finite_mult}
For any $n\in \ZM \setminus \{0\}, d \geq 2$, the Dirac-type function $f:\RM^d\to \RM^d$ defined by 
\[f(r  \cos\theta,r\sin\theta,x_3,\ldots, x_d)=(r\cos(n\theta), r\sin(n\theta), x_3, \ldots, x_d)\] 
has degree $n$ and the standardization maps $s_f:\Lambda \to \ZM^{d+1}$ or $s_{\underline{f}}:\Lambda \to \ZM^{d}$ have uniformly finite multiplicity.

Such a Dirac-type function with degree $0$ is similarly given by setting
\[f(r  \cos\theta,r\sin\theta,x_3,\dots, x_d)=\begin{cases}
    (r\cos(2\theta), r\sin(2\theta), x_3, \dots, x_d), &, \theta \in [0,\pi)\\
    (r\cos(2\theta), -r\sin(2\theta), x_3, \dots, x_d), & \theta \in [\pi,2\pi)
\end{cases}\]

\end{lemma}
\begin{proof}
Away from the origin, the map $(r  \cos\theta,r\sin\theta)\mapsto (r\cos(n\theta), r\sin(n\theta))$ is an $|n|$-fold cover of $\SM^{1}$, and for each branch choice for the inverse it becomes bi-Lipschitz. Hence, the pre-image $s_f^{-1}(\{x\})$ of any point $x\neq 0$ is the union of $|n|$ uniformly bounded sets. 
Since the function $f$ is the $(d-2)$-fold suspension of the circle map which is well-known to have degree $\pm n$, the degree is therefore also $\pm n$. For $n>0$ the map is orientation-preserving and for $n<0$ it is orientation-reversing, hence the degree is $n$. 

The degree $0$ example similarly is $2$-fold covering if one excludes a set of measure zero, and is orientation-preserving at one of the two pre-images of $e_1$ and is orientation-reversing at the other.
\end{proof}

We need a path as in Lemma~\ref{lem:homogeneous} which interpolates between any Dirac-type function and one of the above examples. It also needs to establish differentiability and polynomial properness:
\begin{lemma}
\label{lem:homogeneous2}
Suppose that $f^0:\RM^d \to \RM^d$ is a smooth polynomially proper Lipschitz Dirac-type function and $f^\infty$ a homogeneous Lipschitz Dirac-type function with the same degree, then one can connect them with a continuous path a $(f^t)_{t\in [0,\infty]}$ of Dirac-type functions with the following properties:
\begin{enumerate}
    \item[(i)] The Lipschitz constants of $f^{t}$ are uniformly bounded.
    \item[(ii)] All $f^{t}$ are polynomially proper
    $$|f^{t}(x)|\geq C |x|^s$$
    for all $|x|>c$ with constants $c,C,s$ independent of $t$.
    \item[(iii)] Suppose that $f^\infty$ is smooth outside some set of rays $\RM_+\Omega$, $\Omega\subset \SM^{d-1}$, of measure $0$. For each $x\in \RM^d \setminus (\RM_+\Omega)$ the path $t\mapsto f^{t}(x)$ is differentiable at all $t\notin \{1,\infty\}$ and one has $$\sup_{t\in [a,b]} \sup_{x\in \RM^d \setminus (\RM+_\Omega)}\|\dot{f}^{t}(x)\| < \infty$$
    for every finite interval $[a,b] \subset (0,1) \cup (1,\infty)$.
\end{enumerate}
\end{lemma}
\begin{proof}
The first part of the homotopy is the linear interpolation $(f^t)_{t\in [0,1]}$ as constructed in the proof of Lemma~\ref{lem:homogeneous}. One obtains thus a function $f^1$ which coincides with $f^\infty$ on some ball $B_R$ and with $f^0$ on $B_{2R}^c$. The prototype for the second part is again based on rescaling, $$h^t(x)=(2-t)^{-1}f^1((2-t)x), \qquad x\in \RM^d, t\in [1,2].$$ Using polynomial properness of $f^1$, pick some $R_0>R$ and $C_0>0$ such that $|f^1(y)|\geq C_0 |y|^s$ for $|y|\geq R_0$ and set $m=\inf_{R\leq |y|\leq R_0} |f^1(y)|>0$. We may assume $0<s\leq 1$ as $f^1$ has at most linear growth. One has
$$|h^t(x)|\geq\begin{cases}
|x|, & (2-t)|x| \leq R\\
\frac{m}{R_0}|x|, & R \leq (2-t)|x|\leq R_0\\
C_0(2-t)^{s-1} |x|^s, & (2-t)|x|\geq R_0.
\end{cases}$$
where the first and last case use that $x$ is in the homogeneous part and the part controlled by polynomial properness respectively. 
The second case used just the definition of $m$. Since $(2-t)^{s-1}\geq 1$, in all cases $|h^t(x)|\geq \min(1,\frac{m}{R_0},C_0)|x|^s$ for all $|x|\geq 1$.

The concatenation of $(f^t)_{t\in [0,1]}$ and $(h^t)_{t\in [1,2]}$ already defines a continuous homotopy between $f^0$ and $f^\infty$ which is uniformly polynomially proper, but it does not have the required bounded $t$-derivatives. To compensate for this and prove (iii) one can slow the homotopy using a new coordinate-dependent time and define the path 
$$f^{t}(x) = h^{\tau(t,x)}(x), \qquad \tau(t,x)=1 + \frac{t-1}{t-1+\langle x\rangle}, \qquad t\in [1,\infty).$$ The $t$-derivative exists everywhere outside the rays $\Omega$ and one has by the chain rule
$$|\dot{f}^{t}(x)| = | \dot{h}^{\tau(t,x)}(x) \dot{\tau}(t,x)|\leq \frac{2L |x| |\dot{\tau}(t,x)|}{2-\tau(t,x)}= \frac{2L |x|}{t-1 + \langle x\rangle} \leq 2 L$$
with the Lipschitz constant of $f^1$, where the first inequality also used $f^1(0)=0$. 

Furthermore, all $f^{t}$ also have uniform Lipschitz constant for fixed $t$, since for almost every $x$ we have
$$\bigg|\frac{d}{dx_i}\left(h^{\tau(t,x)}\right)(x)\bigg|
\leq \bigg|\dot{h}^{\tau(t,x)}(x)\pdv_{x_i} \tau(t,x)+ (\pdv_{x_i} h)^{\tau(t,x)}(x)\bigg|\leq \frac{2L|x|}{2-\tau}|\pdv_{x_i} \tau(t,x)|  + L \leq 3L.$$
Finally, since we only rescale the variable $t$ but not $x$, $f^t$ still has the same uniform polynomial lower bound as $h^t$, completing the proof of (ii).
\end{proof}

\begin{remark}
{\rm We need $f^{\infty}$ to be non-smooth on some rays to avoid the topological obstruction that in dimension $d\geq 3$ any smooth function with $|\deg f^{\infty}| \neq 1$ has critical points and therefore will not satisfy the uniform finite multiplicity condition. }$\diamond$
\end{remark}

\begin{proposition}\label{prop:polynomially_bounded_F}
Fix a path $f: [0,\infty] \times \RM^d \to \RM^d$ as in Lemma~\ref{lem:homogeneous2} and assume that $\Lambda$ is a Delone set which does not intersect the rays where $f^\infty$ is not smooth.

Let $\Theta:\RM \to [-1,1]$ be a switch function with threshold $r>0$ chosen such that for each $t \in [0,\infty]$, the set
$$\Zz^{t}_d := \{x\in \Lambda: f_d^{t}(x) \in [-r,r]\}$$
is non-empty for every $t$. Observe that this choice of $r$ is possible by continuity and the compactness of $[0,\infty]$. Let $f^t, F^t: \Lambda\to \RM^d$ be defined by restriction to $\Lambda$ and also consider in the following (relative) coarse cochains $\varphi_{F^{t}}, \varphi_{\underline{F}^{t}}$ on $\Lambda$ (not $\RM^d$). 
\begin{enumerate}
    \item[(i)] $\varphi_{F^{t}}$ is polynomially bounded and $\varphi_{\underline{F}^{t}}$ for $\underline{F}^{t}$ is $\Zz_d^{t}$-polynomially bounded with
    $$\sup_{t\in [0,\infty]} \|\varphi_{F^{t}}\|_{\Xx, k} <\infty, \qquad \sup_{t\in [0,\infty]} \|\varphi_{\underline{F}^{t}}\|_{\Zz^{t}_d, k} <\infty $$ 
    for large enough $k$.
    \item[(ii)] For each compact interval $I$ such that $I \cap \{1,\infty\}= \varnothing$ one has $$\sup_{s,t\in I}\|f^s-f^t\|_{\ell^\infty(\Lambda)} < \infty.$$ 
    \item[(iii)] For each compact interval $I$ such that $I\cap \{1,\infty\}= \varnothing$ the path $t\in I \mapsto \varphi_{\underline{F}^{t}}$ is differentiable with
    $$\sup_{t\in I} \|\dot{\varphi}_{\underline{F}^{t}}\|_{\Zz^{t}_d, k} <\infty $$ 
    for some large enough $k$ and each $\dot{\varphi}_{\underline{F}^t}$ is the coboundary of a $\Zz_d^t$-polynomially bounded cochain.
    \end{enumerate}   
\end{proposition}
\begin{proof}
As $\sup_{z\in\Zz_d^{t}}|f_d^{t}(z)|\leq r$ by assumption on $r$, one has $|f_d^{t}(x)|\leq r+Ld_{\Zz_d^t}(x)$ by the Lipschitz bound and therefore
\[|f^{t}(x)|
 \leq
 |\underline f^{t}(x)|+r+L d_{\Zz_d^t}(x).\]
Since $f^t$ is uniformly polynomially proper, after possibly increasing the threshold radius $c$ so that $r\leq \frac C2|x|^s$ for all $|x|>c$, one obtains
\[|\underline f^{t}(x)|+ d_{\Zz_d^t}(x)
 \geq
 \frac{C}{2\max\{1,L\}}|x|^s.\]
Therefore the functions $\underline{f}^t$ are also $\Zz_d^t$-proper and the $\Zz_d^t$-properness constants $c,C,s$ can be chosen uniformly for all $t\in [0,\infty]$.
The uniform boundedness (i) is then an immediate consequence of Proposition~\ref{prop:poly_proper_bounded}.

For (ii) the bounded derivatives of Lemma~\ref{lem:homogeneous2}(iii) imply
$$\|f^t-f^s\|_{\ell^\infty(\Lambda)} \leq \int_s^t \| \dot{f}^\tau\|_{\ell^\infty(\Lambda)} \difd\tau <\infty.$$

One can then prove (iii) exactly as Proposition~\ref{prop:coarse_bounded_perturbation}: the derivatives of the cocycles are given in terms of the same explicit transgression cochains $A_j^t$ and even though the boundaries are $t$-dependent, all bounds are uniform.
\end{proof}

We can now prove the Kubo formula in the bulk:
\begin{proof}[Proof of Theorem~\ref{th:kubo_poly_proper}(i)]
Since the argument for odd and even $d$ are identical, assume that $d$ is even. As we said above, we assume without loss of generality $\Xx=\Lambda$ is some specific Delone set and $\Hh_\Xx=\ell^2(\Lambda)$. Any polynomially proper Dirac-type function $f:\Lambda\to \RM^d$ can be extended to a smooth polynomially proper Dirac-type function on $\RM^d$ and we fix some such extension $f^0:\RM^d\to \RM^d$ as the starting point for the homotopy. For $f^\infty$ we choose the appropriate function from Lemma~\ref{lemma:finite_mult} with the same degree as $f^0$. The Kubo formula holds for $f^\infty\rvert_\Lambda$ by Theorem~\ref{thm:kubo_uniformly_finite}. The specific examples are smooth outside the sets $\{x_1=x_2=0\}$ for non-zero degree respectively $\{x_2=0\}$ for degree zero. We may therefore assume that $\Lambda$ does not intersect the rays on which $f^\infty$ is not smooth, e.g. by taking $\Lambda=\ZM^d+\frac12 e_2$.

Along the path of Proposition~\ref{prop:polynomially_bounded_F}, we claim that the map
$$t\in [0,\infty] \mapsto \eta_{F^t}(a_0,\dots,a_d)$$
is continuous for any $a_0,\dots,a_d\in \Ss^1_u(\Lambda)$: Suppose first $a_0,\dots,a_d$ are finite-propagation operators with propagation less than some $\rho>0$. 
Then 
\begin{align*}
\eta_{F^t}(a_0,\dots,a_d)= \sum_{\substack{x_0,x_1,\dots,x_{d+1} \in \Lambda\\ x_0=x_{d+1}}}\varphi_{F^t}(x) \, \Tr_{\CM^N}\left(\prod_{j=0}^{d} \langle\delta_{x_j}, a_j\delta_{x_{j+1}}\rangle\right)
\end{align*}
is a finite sum for each $t$. In fact, one can find some radius $R$ independent of $t$ such that only terms with $|x_k|\leq R$ for every $k=0,\ldots,d$ contribute: 
By Proposition~\ref{prop:poly_proper_bounded}, the support of $\varphi_{F^t}$ intersected with a neighborhood of the multidiagonal $\Delta_\rho$ is contained in some ball $B_{p(\rho)}(0)^{d+1}$ with polynomial $p$ depending only on $\rho$, the Lipschitz constant, and the polynomial lower bounds of $f^t$, all of which are uniform in $t$. 
Therefore, $t\mapsto \eta_{F^t}(a_0,\dots,a_d)$ is continuous as each of the finitely many terms is continuous individually. 
As seen in Proposition~\ref{prop:cont_ext}, the uniform norm-bounds of Proposition~\ref{prop:polynomially_bounded_F}(i) imply that the functionals $\eta_{F^t}$ are equicontinuous with respect to the Fr\'echet topology of $\Ss_u^1(\Lambda)$, hence an $\frac{\epsilon}{3}$ argument shows continuity in general.

Proposition~\ref{prop:coarse_bounded_perturbation} applies to any interval $[a,b]\cap \{1,\infty\} = \varnothing$ by Proposition~\ref{prop:polynomially_bounded_F}(ii) and therefore shows each $[\eta_{F^t}]$ defines the same class in cyclic cohomology. For $e \in M_N(\Ss_u^1(\Lambda)^+)$ such that $x = [e]_0-[s(e)]_0$ this implies together with continuity at $t=1,\infty$ that
$$t\in [0,\infty] \mapsto \langle x, [\eta_{F^t}]\rangle = \eta_{F^t}(e,\dots,e).$$
is constant. We conclude
$$\langle x,[D_{f^0}]\rangle = \langle x, [D_{f^\infty}]\rangle = \Xi_d \langle x, [\eta_{F^\infty}]\rangle = \Xi_d \langle x, [\eta_{F^0}]\rangle.$$
\end{proof}

To apply the same argument to the edge case, we also need to deform the edge to end at a standard representative. Therefore one also needs to deform the edge unitary/projection such that it stays localized to a time-dependent boundary $\Zz_d^t$. 

\begin{lemma}\label{lem:technical_homotopy}
Let $H_+ \in \Ss_u^1(\Lambda) \otimes M_N(\CM)$ be a gapped, self-adjoint operator with finite propagation range $\mathrm{prop}(H_+)=\rho$ and let $\scalarmatrix \in  M_N(\CM)$ be self-adjoint matrix. Fix the open interval $J=(-2\|H_+ \|-2\|\scalarmatrix\|, 2\|H_+\|+ 2\|\scalarmatrix\|)$ and consider the space $C_c^\infty(J)$ of compactly supported smooth functions on $J$ with the monotone $C^k$-seminorms $$\|g\|_{C^k}=\sup_{\alpha\leq k, x\in J}|D^\alpha g(x)|.$$
For any path $f^{t}$ as in Proposition~\ref{prop:polynomially_bounded_F} and corresponding mollified half-space projections
\[\mathcal{P}_t= \frac12\left(1+\Theta(f_d^t)\right), \quad \mathcal{P}^\perp_t= \frac12\left(1-\Theta(f_d^t)\right)=1-\mathcal{P}_t\]
define 
$$\hat{H}_t = \mathcal{P}_t H_+ \mathcal{P}_t + \mathcal{P}_t^\perp \scalarmatrix \mathcal{P}_t^\perp.$$
For $g\in C^\infty_c(J)$ set
$$e_{g,t}= g(\hat{H}_t)-\mathcal{P}_t g(H_+) \mathcal{P}_t - \mathcal{P}_t^\perp g(\scalarmatrix)\mathcal{P}_t^\perp.$$
Then the following holds:
\begin{enumerate}
    \item[(i)] The maps $t\mapsto \langle \delta_x,e_{g,t}\delta_y\rangle$ are continuous for each $x,y\in \Lambda$ and differentiable at every $t \notin \{1,\infty\}$. The derivatives $\dot{e}_{g,t}$ are bounded operators when they exist.
    \item[(ii)] 
    The linear maps $g \in C^\infty_c(J)\mapsto e_{g,t} \in \Ss^1_u(\Zz_d^t\subset \Lambda)\otimes M_N(\CM)$ are equicontinuous: For each $k$ there are $C_k$ and $m_k$ such that
    $$\sup_{t\in [0,\infty]} \|e_{g,t}\|_{\pdv^t,k}\leq C_k \|g\|_{C^{m_k}}.$$
    \item[(iii)] On any compact interval $I\subset [0,\infty]$ with $I\cap \{1,\infty\}=\emptyset$ the linear maps $g\mapsto \dot{e}_{g,t}$ are equicontinuous: for each $k$ there are $C_{k,I}$ and $m_k$ such that
    $$\sup_{t\in I} \|\dot{e}_{g,t}\|_{\pdv^t,k}\leq C_{k,I}\|g\|_{C^{m_k}}.$$
\end{enumerate}

\end{lemma}
\begin{proof}
Since $H_+$ acts on $\ell^2(\Lambda)\otimes \CM^N$ it is automatically locally trace-class (the spaces $\Ss_u^p(\Lambda)$ for different $p$ all coincide). By the Combes-Thomas estimate in the discrete finite-propagation case (see e.g., \cite[Theorem 10.5]{AizenmanWarzel}) the resolvents of $H_+$ and $\hat{H}_t$ have exponential off-diagonal decay, in particular one has for every $k$ some constant $c_k$ such that
$$\biggl\|\frac{1}{\hat{H}_t - z}\biggr\|_{k} \leq \frac{c_k}{\min(1,|\mathrm{Im}(z)|^{k+1})}$$
with the $k$-th seminorm of $\Ss_u^1(\Xx)$.
The analogous estimate also holds for the resolvent of $H_+ $. 

The goal is now to establish for every $k$ an estimate
\begin{equation}
\label{eq:bulk_hs_comparison_tech}
\biggl\|\frac{1}{\hat{H}_t - z}- \mathcal{P}_t \frac{1}{H_+ -z}\mathcal{P}_t - \mathcal{P}^\perp_t \frac{1}{\scalarmatrix-z}\mathcal{P}^\perp_t \biggr\|_{\pdv^t,k}\leq \frac{C_k}{\min(1,|\mathrm{Im}(z)|^{n_k})}
\end{equation}
with some constants $C_k, n_k$ uniformly in $z\in \CM\setminus \RM$ and $t$, which finishes the proof using the smooth functional calculus (see \cite{Davies1995}): One can write as a norm-convergent integral over the resolvent
$$
e_{g,t}=\frac{1}{2\pi}\int_{\CM} \pdv_{\overline{z}}\tilde{g}(z) \left(\frac{1}{\hat{H}_t - z}- \mathcal{P}_t \frac{1}{H_+ -z}\mathcal{P}_t - \mathcal{P}^\perp_t \frac{1}{\scalarmatrix-z}\mathcal{P}^\perp_t\right) dz\wedge d\overline{z}
$$
with a suitable almost analytic extension $\tilde{g}$ of $g$ for which $\pdv_{\overline{z}}\tilde{g}$ vanishes faster than any prescribed power $|\mathrm{Im}(z)|^{n_k}$ at the real line and therefore compensates for the resolvents. From explicit constructions of the extensions one then concludes \cite[Lemma 2.2.1]{Davies1995}
$$\bigl\|e_{g,t}\bigr\|_{\pdv^t,k} \leq \tilde{C}_{k}\|g\|_{C^{m_k}}$$
with constants $\tilde{C}_k, m_k$ that depend only on $C_k$ and $n_k$. 

With the shorthands $P_+=\mathcal{P}_t$, $P_-=\mathcal{P}^\perp_t$, $\hat{R}_z = (\hat{H}_t-z)^{-1}$, $R^\pm_z = (H_\pm -z)^{-1}$ we have the geometric resolvent identity
$$\hat{R}_z - P_+ R^+_z P_+ - P_- R^-_z P_- = \hat{R}_z\big( 2P_+P_- + \sum_{\sigma\in\{+,-\}}(P_\sigma H_\sigma -\hat{H}_t P_\sigma) R_z^\sigma P_\sigma\big)$$
which used $1-P_+^2-P_-^2=2P_+P_-$. We can estimate 
\begin{align*}
\|P_\pm H_\pm -\hat{H}_t P_\pm\|_{\pdv^t,k} &= \|P_\pm H_\pm(1-P_\pm^2)- P_\mp H_\mp P_- P_+\|_{\pdv^t,k} \\
&\leq  c_k (1+\rho)^{m_k} (\|H_+\|_{k} +\| \scalarmatrix\|)
\end{align*}
for suitable powers $m_k$ since $P_\pm H_\pm (1-P_\pm^2)$ is supported in a $\rho$-thickening of the boundary $\Zz_d^{t}$ and $P_-P_+$ is supported on $\Zz_d^{t}$ itself. To conclude   \eqref{eq:bulk_hs_comparison_tech} note that by Lemma~\ref{lemma:boundary_estimate}, multiplication $\Ss^1_u(\Lambda)\times \Ss^1_u(\Zz_d^t\subset \Lambda)\to \Ss^1_u(\Zz_d^t\subset \Lambda)$ is equicontinuous with respect to the boundaries $\Zz_d^t$.

For the derivative we can differentiate the smooth functional calculus under the integral sign and need to show that
$$\frac{d}{dt}\frac{1}{\hat{H}_t - z} = - \frac{1}{\hat{H}_t - z} \left(\frac{\difd}{\difd t}\hat{H}_t\right) \frac{1}{\hat{H}_t - z}$$
is supported close to the boundary. Indeed,
\begin{equation}
\label{eq:Hderivative}
\bigg\|\frac{\difd}{\difd t}\hat{H}_t\bigg\|_{\pdv^t,k} = \| \dot{P_+} H P_+ + P_+ H_+ \dot{P_+} - \dot{P_+}\scalarmatrix P_- - P_-\scalarmatrix \dot{P_+} \, \|_{\pdv^{t},k} 
\end{equation}
can be uniformly bounded since
$$\dot{P_+}(x)= \frac12{\Theta'}(f_d^t(x))\dot{f}_d^t(x)$$
and the derivative $\dot{P_+}$ is uniformly bounded and supported on $\Zz_d^{t}$.
Thus, the entire expression in \eqref{eq:Hderivative} is supported on an $\rho$-thickening of the boundary $\Zz_d^t$. The ideal property again finishes the proof.

\end{proof}

\begin{proof}[Proof of Theorem~\ref{th:kubo_poly_proper}(ii)]
We again only discuss the case of even $d$ in detail. As in the proof of part (i) we use the homotopy between $f^0=f$ and a fixed $f^\infty$, but we assume in addition $f^\infty$ is $\Yy_d^\infty$-adapted for some half-space (this is again true for our functions in Lemma~\ref{lemma:finite_mult}). By making $r$ larger if necessary, we also assume $\Zz_d^0$ and $\Zz_d^\infty$ are coarse boundaries for their respective half-spaces $\Yy_d^0$ and $\Yy_d^\infty$. We now construct a path of unitaries $(u_t)_{t\in [0,\infty]}$ which follows along the deformation.

Recall again we use the specific geometric module $\Hh_\Xx=\ell^2(\Lambda)$. We can then write $[u_0]_1 = \pdv_{{\Yy_d}}(x)$ with a class $x=[e]_0-[s(e)]_0$ represented by a bulk projection $e \in C^*_u(\Lambda)^+\otimes M_N(\CM)$.

There exists a gapped self-adjoint operator $H_+\in C^*_u(\Lambda)\otimes M_N(\CM)$ with finite propagation range $\mathrm{prop}(H_+)<\infty$ whose negative spectral projection defines the same K-theory class as $e$. Indeed, one can construct it immediately as a truncation of $1-2e$. For $\scalarmatrix=2\one_N$ consider the operators $\hat{H}_t$ as in Lemma~\ref{lem:technical_homotopy} with bulk spectral gap $\Delta\subset J$. Defining then boundary unitaries $u_{t}=e^{2\pi\imath \rho_\Delta(\hat{H}_t)}$ we have by Proposition~\ref{prop:boundary_maps}
$$[u_0]_1 = \pdv_{\Yy_d^0}(x), \qquad [u_\infty]_1 = \pdv_{\Yy_d^\infty}(x)$$
with the boundary maps for the respective exact sequences $$0\to C_u^*(\Zz_d^t\subset \Lambda) \to E_{\Yy_d^t} \to C_u^*(\Lambda) \to 0$$ 
at $t = 0$ and $t = \infty$. By Theorem~\ref{th:bulk_edge} we have the bulk-edge correspondence
\begin{equation}
\label{eq:index_doesnt_change}
\langle [u_0]_1, [D_{\underline{f}^{0}}]\rangle= \langle x, [D_{f^{0}}]\rangle=\langle x, [D_{f^{\infty}}]\rangle= \langle [u_\infty]_1, [D_{\underline{f}^{\infty}}]\rangle.
\end{equation}
For the middle equality, we just used that $f^{0}$ and $f^{\infty}$ have the same degree.

We will argue that similarly as in the proof of part (i) the map
$$t \in [0,\infty] \mapsto \langle [u_t]_1, [\eta_{\underline{F}^{t}}]\rangle$$
is constant. Then the proof is finished by
$$\Xi_{d-1} \langle [u_0]_1, [\eta_{\underline{F}^{0}}]\rangle = \Xi_{d-1} \langle [u_\infty]_1, [\eta_{\underline{F}^{\infty}}]\rangle = \langle[u_\infty]_1, [D_{\underline{f}^{\infty}}]\rangle = \langle [u_0]_1, [D_{\underline{f}^{0}}]\rangle.$$

Abbreviating the respective inputs $u_t-\one_N$, $u_t^*-\one_N$ of the Kubo cocycle by $a_j^{t}$ we want to show that
\begin{equation}
\label{eq:tech_cocycle_matrixelements}
t\in [0,\infty]\mapsto \eta_{\underline{F}^t}(a^t_0,\dots,a^t_{d-1})= \sum_{\substack{x_0,x_1,\dots,x_{d} \in \Lambda\\ x_0=x_{d}}}\varphi_{\underline{F}^t}(x) \Tr_{\CM^N}\left(\prod_{j=0}^{d-1} \langle\delta_{x_j}, a^{t}_j\delta_{x_{j+1}}\rangle\right)
\end{equation}
is continuous and constant. Indeed, for $t \notin \{1,\infty\}$, the derivative is formally
\begin{equation}
\label{eq:cycle_derivative}\frac{d}{dt}\eta_{\underline{F}^{t}}(u^*_t, u_t,\dots) = \dot{ \eta}_{\underline{F}^t}(u^*_t,\dots) + \sum_{k=0}^{d-1} \eta_{\underline{F}^t}(u^*_t, \dots \dot{u}_t, \dots, u^*_t, u_t)=0.
\end{equation}
In fact, both terms on the right-hand side are well-defined and $0$; by Proposition~\ref{prop:polynomially_bounded_F} the former term is the pairing of a coboundary with a K-theory class, whereas the latter term is a pairing of $\eta_{F^t}$ with a boundary in cyclic homology by \cite[Proposition 3.3]{Getzler1993OddChern}. While the equality \eqref{eq:cycle_derivative} can be justified, it requires one to show that one can compute the derivative of \eqref{eq:tech_cocycle_matrixelements} term-wise, which is not evident.

We therefore finish the proof by an approximation argument. 
For any compact interval $I\subset [0,\infty]$ and $\epsilon>0$, we claim there exists some $\rho_\epsilon>0$ and approximations $a_{j, \epsilon}^t$ supported in a $\rho_\epsilon$-thickening of $\Zz_d^t$ and with propagation less than $\rho_\epsilon$ such that 
$$\|a_{j, \epsilon}^t-a_{j}^t\|_{\pdv^t,k}< \epsilon, \qquad \forall t\in I$$
for some fixed large enough $k$ and, if $I\cap \{1,\infty\}=\varnothing$, we also can arrange
$$\|\dot{a}_{j, \epsilon}^t-\dot{a}_{j}^t\|_{\pdv^t,k}< \epsilon, \qquad \forall t\in I.$$
By same argument as in the bulk case, the sum
$$\eta_{\underline{F}^t}(a^t_{0,\epsilon},\dots,a^t_{d-1,\epsilon}) = \sum_{\substack{x_0,x_1,\dots,x_{d} \in \Lambda\\ x_0=x_{d}}}\varphi_{\underline{F}^t}(x)  \Tr_{\CM^N}\left(\prod_{j=0}^{d-1}\langle\delta_{x_j}, a^t_{j,\epsilon}\delta_{x_{j+1}}\rangle\right)$$
for fixed $\epsilon$ is then over a finite index set depending only on $\rho_\epsilon$, hence is continuous and one can take term-wise derivatives. By the uniform upper bounds of Proposition~\ref{prop:polynomially_bounded_F}, the approximation errors are controlled:
$$\sup_{t\in I}|\eta_{\underline{F}^t}(u^*_t,u_t,\dots)-\eta_{\underline{F}^t}(a_{0,\epsilon}^t, a_{1,\epsilon}^t,\dots)|=O(\epsilon),$$
which implies continuity, and if $I\cap \{1,\infty\}= \varnothing$ then
\begin{align*}
\sup_{t\in I}|\dot{ \eta}_{\underline{F}^t}(u^*_t,u_t,\dots)-\dot{ \eta}_{\underline{F}^t}(a_{0,\epsilon}^t, a_{1,\epsilon}^t,\dots)|&=O(\epsilon), \\
\sup_{t\in I}|\sum_{k=0}^{d-1}\eta_{\underline{F}^t}(u^*_t, \dots, \dot{u}_t, \dots, u_t)-\sum_{k=0}^{d-1}\eta_{\underline{F}^t}(a_{0,\epsilon}^t, \dots, \dot{a_{k,\epsilon}^t}, \dots, a_{d-1,\epsilon}^t)|&=\,O(\epsilon),
\end{align*}
so the fact that both terms on the right-hand side of \eqref{eq:cycle_derivative} vanish implies $$\sup_{s,t\in I}\big|\eta_{\underline{F}^s}(a^s_{0,\epsilon},\dots,a^s_{d-1,\epsilon})-\eta_{\underline{F}^t}(a^t_{0,\epsilon},\dots,a^t_{d-1,\epsilon})\big|=|I| O(\epsilon).$$ 
Applying this for all compact intervals we conclude $t\in [0,\infty] \mapsto \eta_{\underline{F}^t}(a^t_{0},\dots,a^t_{d-1})$ is constant.

Let us finally give the construction of the $a^t_{j,\epsilon}$ which is immediate from Lemma~\ref{lem:technical_homotopy}: Consider some polynomials $g_K$ whose restrictions to $J$ converge in $C_c^\infty(J)$ to $\exp(2\pi \imath \rho_\Delta\rvert_J)$. Then the approximations $u_{t,K} = e_{g_K, t}$ of $u_t$ each have finite propagation less than some $\rho_K$ and are supported on a $\rho_K$-thickening of $\Zz_d^t$ where $\rho_K$ does not depend on $t$. By Lemma~\ref{lem:technical_homotopy} one has
$$\sup_{t\in I}\|u_{t,K} -u_{t}\|_{\pdv^t, k} \to 0$$
for any compact interval $I$ and if $I\cap \{1,\infty\}=\varnothing$ then also
$$\sup_{t\in I}\|\dot{u}_{t,K} -\dot{u}_{t}\|_{\pdv^t, k}\to 0.$$

In the odd case one argues similarly using the same path of Dirac-type functions $f^t$ but one needs a path of edge projections instead. For that, one defines $H$ as a truncation of $\begin{psmallmatrix}
    0 & U^*\\ U & 0
\end{psmallmatrix}$ for $[U]_1$ a preimage of $y\in K_0(C^*_u(\Zz_d^t\subset \Lambda))$ under the boundary map and sets $\scalarmatrix=\begin{psmallmatrix}
    0 & 1\\ 1 & 0
\end{psmallmatrix}$ since both $H$ and $\scalarmatrix$ need to be chirally symmetric. Then the boundary projections as defined in Definition~\ref{def:edge_invariants} give the required weakly differentiable path of edge projections.

\end{proof}

\appendix

\section{Conventions and index pairing}
\label{sec:index}
We fix an irreducible representation of the Clifford algebra $\CM_d$  (using the same convention as \cite{ProdanSchulzBaldes2016,SchulzBaldesStoiber2023}) by recursively constructing generators $\gamma_1,\dots,\gamma_d$, as follows:
For $d=1$ we set $\gamma_1=1 \in \Bb(\CM)$. For $d>1$ let $\hat{\gamma}_1, \dots, \hat{\gamma}_{d-1}$ be the representation for $d-1$. If $d$ is odd we set $\gamma_i=\hat{\gamma}_i$ for $i=1,\dots,d-1$ and $\gamma_d= \hat{\gamma}_0= (-\imath)^{(d-1)/2}\hat{\gamma}_1\dots\hat{\gamma}_{d-1}$. If $d$ is even then we set
$$\gamma_i= \begin{pmatrix}
    0 & \hat{\gamma}_i \\ \hat{\gamma}_i & 0
\end{pmatrix}, \qquad \gamma_d = \begin{pmatrix}
    0 & -\imath \\ \imath & 0
\end{pmatrix}, \qquad \gamma_0 = \begin{pmatrix}
    1 & 0 \\ 0 & -1
\end{pmatrix}.$$
Here, $\gamma_0$ represents the grading on the representation space $\CM^{N(d)}$, where the dimension of this representation is $N(d)=2^{\lceil (d-1)/2 \rceil}$.

Denote by $\rho_d:\CM_d \to \Bb(\CM^{N(d)})$ the (ungraded) Hilbert-space representation by setting 
$\rho_d(\sigma_i) = \gamma_i$. For even $d$ one obtains in the same way a graded representation $\hat{\rho}_d: \CM_d\to \Bb(\CM^{N(d)})$ 
where $\CM^{N(d)}$ is graded by $\gamma_0$. In odd dimension there is no graded representation on $\CM^{N(d)}$, instead we use the distinct representation
\begin{alignat*}{2}
&\hat{\rho}_d: \CM_{d} \to \Bb(\CM^{N(d)})\hat{\otimes} \CM_1, \qquad \hspace{1.7cm} \hat{\rho}_d(\sigma_i) &&= \gamma_i \hat{\otimes} \ep
\end{alignat*} 
where $\CM^{N(d)}$ is considered to be ungraded. We make the convention that the Bott periodicity isomorphism in $KK$ respectively $KK^{\mathrm{sep}}$ is given by Kasparov product with $1\hat{\otimes}[\hat{\rho}_d]$.

\begin{definition}
Let $A,B$ be two ungraded separable $C^*$-algebras.

A standard $KK_0(A,B)$-cycle is a $KK_0(A,B)$-cycle of the form $((\Hh\oplus \Hh^{\mathrm{op}})\hat{\otimes} B, \phi_2, F)$ with $\Hh, \Hh^{\mathrm{op}}$ two copies of the same Hilbert space and grading operator $\one\oplus (-\one)$, $\phi_2=\phi\oplus \phi$ with a homomorphism $\phi: A \to \Ll_B(\Hh\otimes B)$ and $F=\begin{psmallmatrix}
    0 & T^* \\ T & 0
\end{psmallmatrix}$ a bounded odd self-adjoint operator with off-diagonal part $T$.

A standard $KK_1(A,B)$ cycle is $KK_1(A,B)$-cycle of the form $(\Hh\hat{\otimes}(B\hat{\otimes}\CM_1), \phi, F)$ with $\Hh$ an ungraded Hilbert space, $\phi:A\to \Ll_B(\Hh\otimes B)$, and $F=T\hat{\otimes}\ep$ with $T$ is a bounded self-adjoint operator.
\end{definition}
Every $KK_0$-class and $KK_1$-class admits a representative in standard form \cite{blackadar1998k}.

We make the following conventions:
\begin{enumerate}
\item The group $K_0(A)$ for a $C^*$-algebra $A$ consists of stable norm-homotopy equivalence classes $[e]_0-[s(e)]_0$
with $e\in M_N(A^+)$ matrix-valued projections over the unitization with $s(e)\in M_N(\CM)$ the scalar part such that $e-s(e)\in M_N(A)$. 

\item A homomorphism $KK_0(\CM,A)\to K_0(A)$ is defined by mapping a standard $KK_0(\CM,A)$-cycle to its $K_0(A)$-valued Fredholm index $\Ind_A(T)\in K_0(A)$. It is an isomorphism and its inverse will be denoted $\kappa_0^A: K_0(A)\to KK(\CM,A)$.

\item The group $K_1(A)$ consists of stable norm-homotopy equivalence classes $[u]_1$ of unitaries $u\in M_N(A^+)$ with scalar part $s(u)=\one_N$.

\item A homomorphism $KK_1(\CM,A)\to K_1(A)$ is defined by mapping a standard $KK_1(\CM,A)$-cycle to the class $[-\exp(-\pi \imath T)]_1\in K_1(A)$. It is an isomorphism and its inverse will be denoted $\kappa_1^A: K_1(A)\to KK_1(\CM,A)$.
\end{enumerate}
Under the isomorphism $KK(\CM,\CM)\simeq K_0(\CM) \simeq \ZM$ a standard KK-class represented by an odd self-adjoint Fredholm operator is sent to the usual Hilbert-space index $\Ind(T)=\mathrm{dim}\ker(T)-\mathrm{dim}\Ker(T^*) \in \ZM$.

The following Kasparov products with K-groups are well-known:
\begin{proposition}
\label{prop:basic_products}
Let $[e]_0-[s(e)]_0\in K_0(A)$ and let $[u]_1\in K_1(A)$. 

For the product with a standard $KK_1(A,B)$-cycle,
$$\kappa^A_0([e]_0-[s(e)]_0) \otimes [(\Hh\hat{\otimes}(B\hat{\otimes}\CM_1), \phi, F)]= [e^{-\pi \imath (\phi(e) T \phi(e)+1-\phi(e))}e^{\pi \imath (s(e) T s(e)+1-s(e))}] \in K_1(B)$$
and
$$\kappa^A_1([u]_1) \otimes [(\Hh\hat{\otimes}(B\hat{\otimes}\CM_1), \phi, F)]= \Ind_B(P \phi(u)P+1-P) \in K_0(B)$$
with $T=2P-1$.

For the product with a standard $KK_0(A,B)$-cycle,
\begin{align*}
\kappa^A_0([e]_0-[s(e)]_0) &\otimes [((\Hh\oplus \Hh^{\mathrm{op}})\hat{\otimes} B, \phi_2, F)]\\
&= \Ind_B(\phi(e) T \phi(e) + 1-\phi(e))- \Ind_B(s(e) T s(e) + 1-s(e))  \in K_0(B)
\end{align*}

\end{proposition}
\begin{proof}
These can be found in e.g., \cite[Proposition 5.28-31]{ProdanSchulzBaldes2016B},  we translated them to our conventions and added scalar parts where necessary.
\end{proof}

Using these products and Bott periodicity, the index pairings of Definition~\ref{def:index} can now be expressed as Fredholm indices as follows:

\begin{lemma}\label{lemma:index_kk_product}
Let $A$ be an ungraded $C^*$-algebra represented on an ungraded Hilbert space $\Hh$ via a representation $\pi$. Consider an unbounded $KK_d(A, \CM)$-cycle $(\Hh\hat{\otimes} \CM_d, \pi\hat{\otimes}\one_{\CM_d}, D)$ with $\pi:A\to \Bb(\Hh)$, and an invertible operator $D=\sum_{j=1}^d D_j\otimes \sigma_j$ with commuting self-adjoint operators $D_1,\ldots,D_d$.

If $d$ is odd and $[u]_1\in K_1(A)$ then
\begin{equation}
\label{eq:oddpairing}
\langle [u]_1, [D] \rangle = \Ind(P \pi(u) P +1- P)
\end{equation}
for $P=\chi(\rho_d(D)>0)$ the positive spectral projection.

If $d$ is even and $[e]_0-[s(e)]_0\in K_0(A)$ is represented by a projection $e\in M_N(A^+)$ with scalar part $s(e)\in M_N(\CM)$ then
\begin{equation}
\label{eq:evenpairing}
\langle [e]_0-[s(e)]_0, [D]\rangle = \Ind(\pi(e) G\pi(e)+ 1-\pi(e))
\end{equation}
where $G$ is the off-diagonal part of $$\sgn(\rho_d(D))=\begin{pmatrix}
    0 & G^* \\ G & 0
\end{pmatrix}, \qquad \gamma_0=\begin{pmatrix}
    \one & 0 \\ 0 & -\one
\end{pmatrix}$$
with the matrix representation relative to spectral subspaces of the grading operator $\gamma_0$. 
\end{lemma}
\begin{proof}
It suffices to prove this in the case that $A$ is separable, as the product in $KK(\CM,\CM)\simeq \ZM$ is expressed by the index of a concrete Fredholm operator which does not depend on the separable subalgebras by compatibility of the inverse limit family. 

Since $D$ is invertible, we use the sign-function for the bounded transform $[(\Hh\hat{\otimes}\CM_d, \pi, D)]=[(\Hh\hat{\otimes}\CM_d, \pi, \sgn(D))]$.
After Bott periodicity, $[(\Hh\hat{\otimes}\CM_d, \pi, D)] \otimes [\hat{\rho}_d]$ therefore corresponds to the standard KK-cycle
$$((\Hh\otimes \CM^{N(d)/2})\oplus (\Hh\otimes \CM^{N(d)/2}), \pi, \begin{pmatrix}
    0 & G^*\\ G & 0
\end{pmatrix})$$
if $d$ is even, and corresponds to
$$((\Hh\otimes \CM^{N(d)})\hat{\otimes}\CM_1, \pi, (2P-1)\hat{\otimes}\ep)$$
if $d$ is odd. 

Therefore, this is an immediate consequence of Proposition~\ref{prop:basic_products}, where we observe that the scalar part $\Ind_\CM(s(e) G s(e) + 1-s(e))=0$ vanishes as it is the index of a unitary operator.
\end{proof}

In general, $D$ might not be invertible, in which case the pairing is computed by passing to the doubled Hilbert module and set \cite[Section~2.1]{Andersson2019}
$$\tilde{\pi}(a)=\begin{pmatrix}
    \pi(a) & 0 \\ 0 & 0
\end{pmatrix}, \qquad \tilde{D}_\mu = \begin{pmatrix}
    D & \mu \\ \mu & -D
\end{pmatrix}$$
for any real number $\mu\neq 0$ and define the index pairings that way.

\section{Schwartz algebras and commutators}\label{sec:schwartz}
Throughout this appendix, let $\Xx$ be a proper metric space with bounded geometry and fix a uniformly discrete coarsely equivalent subset $\Lambda$ and Borel partition $(Q_\lambda)_{\lambda\in \Lambda}$ as in Section~\ref{sec:roe}. For an operator $T$ on a geometric module $\Hh_\Xx$, we denote by $T_{\lambda\mu}:\Hh_\mu\to \Hh_\lambda$ its restriction to fixed blocks.

\begin{lemma}
\label{lemma:schur_test}
\begin{enumerate}
    \item[(i)] Suppose $\Xx$ has bounded geometry. Then any collection of operators $T_{\lambda\mu}\in \mathcal B(\mathcal H_\mu,\mathcal H_\lambda)$ 
with uniformly bounded norms and $T_{\lambda\mu}=0$ whenever $d(\lambda,\mu)>\rho$ defines a bounded finite-propagation operator $T=\sum_{\lambda,\mu\in \Lambda}\chi_\lambda T_{\lambda,\mu} \chi_\mu$ with
$$\|T\|\leq C_\rho \sup_{\lambda,\mu} \|T_{\lambda \mu}\|$$
for some constant that depends only on $\rho$.
\item[(ii)] Suppose $\Xx$ has polynomial growth with exponent $\nu$. Then for any $1 \leq p \leq \infty$ and $k > \nu + 1$, there exists some constant $C_k$ such that any family of operators $(T_{\lambda\mu})_{\lambda,\mu}$ with finite $\|\cdot\|_{k,p}$-norm satisfies
$$\|T\|\leq C_{k} \|T\|_{k,p}.$$
\end{enumerate}
\end{lemma}
\begin{proof}
Both statements follow from the general version of the Schur test
$$\|T\| \leq \left( \sup_{x\in \Lambda} \sum_{y\in \Lambda} \|\chi_x T \chi_y\| \right)^\frac12 \left( \sup_{y\in \Lambda} \sum_{x\in \Lambda}\|\chi_x T \chi_y\|\right)^\frac12$$
which is valid for any partition of unity by projections. Since the Schatten norms dominate the operator norm we can bound
\begin{align*}
\sup_{x\in \Xx} \sum_{y\in \Xx} \|\chi_x T\chi_y\|_{S^p}&\leq \norm{T}_{k,p} \sup_{x\in \Xx} \sum_{y\in \Xx} \langle d(x,y)\rangle^{-k},
\end{align*}
hence the right-hand side is finite for large enough for $k>\nu+1$. 
\end{proof}
In particular, (ii) implies that $\Ss^p_u(\Xx, \Hh_\Xx)$ is a closed subspace of $C^*_u(\Xx,\Hh_\Xx)$ in the topology generated by its seminorms, i.e., a Fr\'echet space.
It is moreover an m-convex Fr\'echet algebra, as the special case $\Zz=\Xx$ of the following Lemma shows, and the subalgebras $\Ss^p_u(\Zz\subset \Xx,\Hh_\Xx)$ are closed ideals: 

\begin{lemma}
\label{lemma:boundary_estimate}
Let $\Xx$ be a proper metric space with polynomial growth and fix $\varnothing \neq \Zz\subset \Xx$ closed.
Fix $1 \leq p \leq \infty$ and let $a\in \Ss_u^\infty(\Xx, \Hh_\Xx)$ and $b\in \Ss_u^p(\Zz\subset \Xx, \Hh_\Xx)$. 
Then for each $k>0$, there are constants $c, k_1,k_2$ such that
$$\|ab\|_{\pdv, k, p}\leq c\|a\|_{k_1, \infty}\, \|b\|_{\pdv,k_2, p}.$$
Moreover, $c,k_1,k_2$ can be chosen to be independent of $\Zz$.
\end{lemma}
\begin{proof}
The triangle inequality gives 
\[ \langle d(x,z)\rangle \leq 2\langle d(x,y)\rangle \langle d(y,z)\rangle; \qquad  \langle d(x,\Zz)\rangle \leq 2\langle d(x,y)\rangle \langle d(y,\Zz)\rangle\]

For $k_1>2k$ and $k_2>k+\nu +1$ we can estimate
\begin{align*}
\|\chi_{x} ab\chi_z\|_{S^p} \langle d(x,z)\rangle^k \langle d(x,\Zz)\rangle^k \langle d(z, \Zz)\rangle^k \leq   \|a\|_{k_1, \infty}\, \|b\|_{\pdv,k_2,p} \sum_{y\in \Lambda} \langle d(y,z)\rangle^{k-k_2}.
\end{align*} 
By the polynomial growth assumption, the sum on the right-hand side is finite and bounded independently of $z$.
\end{proof}

\begin{lemma}\label{lemma:lipschitz}
Let $f:\Xx\to\RM$ be a Borel function satisfying the large-scale Lipschitz condition $\|f(x)-f(y)\|\leq L (1+d(x,y))$ and identify it with the corresponding (possibly unbounded) multiplication operator on a geometric module $\Hh_\Xx$.

\begin{enumerate}
    \item[(i)] Suppose $\Xx$ has bounded geometry. Then any $a\in \CM_u(\Xx, \Hh_\Xx)$ maps $\Dom(f)\subset \Hh_\Xx$ into itself and the commutator $[f, a]$ extends to a bounded operator in $\CM_u(\Xx, \Hh_\Xx)$ with
\[\|[f,a]\|\leq C_\rho (1+L)\|a\|\]
where $C_\rho$ depends only on $\Xx$ and the propagation of $a$.
\item[(ii)] Suppose in addition that $\Xx$ has polynomial growth. Then for $1 \leq p \leq \infty$, any $a\in \Ss_u^p(\Xx, \Hh_\Xx)$ maps $\Dom(f)$ into itself and the commutator $[f, a]$ extends to a bounded operator in $\Ss^p_u(\Xx, \Hh_\Xx)$.
\end{enumerate}
\end{lemma}
\begin{proof}
For the bounded diameter $R=\sup_\lambda \mathrm{diam}(Q_\lambda)$ the matrix coefficients satisfy
$$\|\chi_x [f,a]\chi_y\|_{S^p}\, \leq\, (|f(x)-f(y)|+ 4(1+R)L)\leq 5L(1+R+d(x,y)) \| \chi_x  a \chi_y\|_{S^p}$$
which immediately implies the norm bounds by Lemma~\ref{lemma:schur_test}. 
As $a$ maps compactly supported functions to compactly supported rapidly decaying functions respectively, it maps a core of $f$ to $\Dom(f)$, which together with the bounded commutator implies $a \Dom(f)\subset \Dom(f)$ (see e.g., \cite[Proposition~2.1]{ForsythMeslandRennie2014}).
\end{proof}

\begin{lemma}\label{lem:normcont}
For any large-scale Lipschitz function $f: \Xx\to \RM$ the time evolution $t \mapsto \alpha_t(a):= e^{\imath ft} a e^{-\imath ft}$ is a strongly continuous action on $C^*_u(\Xx, \Hh_\Xx)$ which leaves $\CM_u(\Xx, \Hh_\Xx)$ invariant. The subalgebra $\Aa\subset C^*_u(\Xx, \Hh_\Xx)$ of elements which are norm-smooth with respect to $\alpha$ is norm-dense. Each $a\in \Aa$ preserves the domain of $f$ and $[f, a]$ extends to a bounded operator in $\Aa$. 
\end{lemma}
\begin{proof}
For $a\in \CM_u(\Xx, \Hh_\Xx)$ one has $a\Dom(f)\subset \Dom(f)$ and the commutator $[f, a]$ is bounded by Lemma~\ref{lemma:lipschitz}. For every $\psi\in \Dom(f)$ one has 
$$
\bigg|\bigg| \frac{d}{dt} \alpha_t(a) \psi   \bigg|\bigg| = \| \alpha_t([f,a])  \psi || \leq ||[f,a]\|\, \|\psi\|
$$
which implies
$$\|\alpha_t(a)-a\|\leq |t| \|[f,a]\|$$
and therefore norm-continuity of $\alpha$ by a density argument.  The subalgebra of elements of a Banach algebra which are norm-smooth with respect to a strongly continuous $\RM$-action is a dense Fr\'echet subalgebra (see e.g., \cite{Bratteli1986}). For each $a\in \Aa$ the derivative $\delta(a)=\lim_{t\to 0} \frac{\alpha_t(a)-a}{t}\in \Aa$ exists in operator norm coincides with the bounded extension of $\imath[f, a]$, hence the commutator extends to an element of $\Aa$.
\end{proof}

We will need that $C^*_u(\Lambda)$ has the same K-theory as $C^*_u(\Xx,\Hh_\Xx)$. This is based on a strong Morita-equivalence. For specific geometric modules this was proven in
\cite[Propositions~4.7-4.8]{SpakulaWillett2013}. We include an elementary proof since we require it for more general modules and also need the relative statement.

\begin{proposition}\label{prop:morita_continuous}
Let $\Xx$ be a proper metric space with bounded geomety. Let $\Hh_\Xx$ be a coarsely faithful geometric module with $\Hh_\lambda\neq \{0\}$ for all $\lambda\in \Lambda$. Choose unit vectors $\xi_\lambda\in\Hh_\lambda$ and set $p:=\bigoplus_{\lambda\in\Lambda} |\xi_\lambda\rangle\langle\xi_\lambda|$.
Then $p\in C_u^*(\Xx,\Hh_\Xx)$ is a full projection and compression induces an
isomorphism
\[ pC_u^*(\Xx,\Hh_\Xx)p\cong C_u^*(\Lambda).\]
Consequently, $C_u^*(\Xx,\Hh_\Xx)$ and $C_u^*(\Lambda)$ are strongly Morita
equivalent.

For a nonempty closed subset $\Zz \subset \Xx$, for $\Lambda_\Zz \coloneqq \{\lambda \in \Lambda: Q_{\lambda} \cap \Zz \neq \varnothing \}$, $p$ acts as a full multiplier projection on $C^*_u(\Zz \subset \Xx, \Hh_\Xx)$ and
\[pC_u^*(\Zz\subset\Xx,\Hh_\Xx)p
 \cong C_u^*(\Lambda_\Zz\subset\Lambda, \ell^2(\Lambda)).\]
The full-corner inclusions thus induce isomorphisms of $K$-theory groups
\[K_*\bigl(C_u^*(\Lambda, \ell^2(\Lambda))\bigr)
 \cong K_*\bigl(C_u^*(\Xx,\Hh_\Xx)\bigr); \qquad  K_*\bigl(C_u^*(\Lambda_\Zz\subset\Lambda)\bigr)
 \cong K_*\bigl(C_u^*(\Zz\subset\Xx,\Hh_\Xx)\bigr).\]
\end{proposition}

\begin{proof}
We abbreviate $A=C_u^*(\Xx,\Hh_\Xx)$ and, for $T\in\Bb(\Hh_\Xx)$, write
\[T_{\lambda\mu}:=\chi_\lambda T\chi_\mu
 \in\Bb(\Hh_\mu,\Hh_\lambda).\]

Clearly, $p$ has finite propagation and is uniformly locally compact, hence it is in $C^*_u(\Xx,\Hh_\Xx)$.

Let $U:\ell^2(\Lambda)\to p\Hh_\Xx$ be the unitary
$U\delta_\lambda=\xi_\lambda$.  On finite-propagation elements of $C^*_u(\Lambda,\ell^2(\Lambda))$ define as a strongly convergent sum
\[\Phi(a):=UaU^* =\sum_{\lambda,\mu\in\Lambda} a_{\lambda\mu}|\xi_\lambda\rangle\langle\xi_\mu|.\]
The propagation changes by at most $2R$ for $R=\sup_{\lambda}\diam(Q_\lambda)$, and every block of $\Phi(a)$ has rank
at most one, hence $\Phi(a)\in A$. Likewise, if $T\in A$ has finite propagation then $pTp$ is in the image of $\Phi$. Taking the closure one obtains an isomorphism $pAp\cong C_u^*(\Lambda,\ell^2(\Lambda)).$

It remains to prove that $p$ is full. 
By bounded geometry, we get for $s\geq0$ that
\[N_s:=\sup_{\lambda\in\Lambda} \#\bigl(B_s(\lambda)\cap\Lambda\bigr)<\infty.\]
 Let $\mathcal A_{\mathrm{fr}}\subset \CM_u(\Xx, \Hh_\Xx)$ be the set of finite-propagation operators $T$ for which $\sup_{\lambda,\mu}\operatorname{rank}(T_{\lambda\mu})<\infty$. 
We claim that $\mathcal A_{\mathrm{fr}}$ is dense in $A$.  It is enough to
consider a finite-propagation uniformly locally compact operator $T$.  Put
$S=\operatorname{prop}(T)$ and $D=N_{S+2R}$.  If $T_{\lambda\mu}\neq0$, then
$d(\lambda,\mu)\leq S+2R$, so every row and every column of the block matrix of
$T$ contains at most $D$ nonzero blocks.

Fix any $\epsilon>0$ and put $\delta=\epsilon/D$. Uniform local compactness, applied to the family
$\{\chi_{B_R(\lambda)}T:\lambda\in\Lambda\}$, gives an integer $n$ and
finite-rank operators $G_\lambda=\chi_\lambda G_\lambda$ of rank at most $n$ such that
\[\|\chi_{\lambda}T-G_\lambda\|<\delta
 \qquad \forall\lambda\in\Lambda.\]
Define a block matrix $F$ by
\[F_{\lambda\mu}:=
 \begin{cases}
  \chi_\lambda G_\lambda\chi_\mu,
     &d(\lambda,\mu)\leq S+2R,\\
  0,&d(\lambda,\mu)>S+2R.
 \end{cases}\]
Then $\operatorname{rank}(F_{\lambda\mu})\leq n$ and
$\|T_{\lambda\mu}-F_{\lambda\mu}\|<\delta$ for every $\lambda,\mu$.
The Schur test gives
\[\|T-F\|\leq D\delta=\epsilon.\]
The operator $F$ is in $\mathcal A_{\mathrm{fr}}$, hence this proves the density claim.

We now show that every $F\in\mathcal A_{\mathrm{fr}}$ is a linear combination of elements in $ApA$.  Let
$n$ bound the ranks of its blocks and let $D$ bound the number of nonzero
blocks in every row and column.  We can think of
\[\mathcal E_F:=\{(\lambda,\mu):F_{\lambda\mu}\neq0\}\]
as the edge set of a bipartite graph with two copies of $\Lambda$.  This graph
has degree at most $D$.  Since it is countable, a greedy edge-colouring with
$2D-1$ colours decomposes $\mathcal E_F$ into finitely many matchings
$\mathcal M_1,\ldots,\mathcal M_{2D-1}$.

For each
$(\lambda,\mu)\in\mathcal E_F$, choose a singular-value decomposition
\[
F_{\lambda\mu} =\sum_{j=1}^n s_{\lambda\mu,j} |\eta_{\lambda\mu,j}\rangle \langle \zeta_{\lambda\mu,j}|,\]
where $\eta_{\lambda\mu,j}\in\Hh_\lambda$ and
$\zeta_{\lambda\mu,j}\in\Hh_\mu$ are unit vectors, and the singular values satisfy
$0\leq s_{\lambda\mu,j}\leq\|F_{\lambda\mu}\|$.  For a matching
$\mathcal M_c$ and $1\leq j\leq n$, set
\begin{align*}
 x_{c,j}:=\sum_{(\lambda,\mu)\in\mathcal M_c}
     s_{\lambda\mu,j}
    |\eta_{\lambda\mu,j}\rangle \langle \xi_\lambda|, \qquad
 y_{c,j}:=\sum_{(\lambda,\mu)\in\mathcal M_c}
     |\xi_\lambda\rangle \langle \zeta_{\lambda\mu,j}|.
\end{align*}
Because $\mathcal M_c$ is a matching, these sums
are orthogonal direct sums; in particular,
\[\|x_{c,j}\| \leq \|F\|, \qquad \|y_{c,j}\| \leq 1.\]
They have finite propagation and are locally rank-one, so
$x_{c,j},y_{c,j}\in A$.  Moreover, $x_{c,j}p=x_{c,j}$ and
$py_{c,j}=y_{c,j}$, and blockwise one has
\begin{equation}
\label{eq:Fdecomp}F=\sum_{c=1}^{2D-1}\sum_{j=1}^n x_{c,j}p y_{c,j}.
\end{equation}
Thus $\mathcal A_{\mathrm{fr}}\subset \operatorname{span}(ApA)$.  Since
$\mathcal A_{\mathrm{fr}}$ is dense in $A$, the ideal generated by $p$ is all
of $A$. Equivalently stated, $p$ is a full projection.

For the relative case, let $I_\Zz=C_u^*(\Zz\subset\Xx,\Hh_\Xx)$. If $F\in\mathcal A_{\mathrm{fr}}$ is supported near $\Zz$, then the factorization \eqref{eq:Fdecomp} is already term-wise in $I_\Zz p I_\Zz$ as all non-zero $x_{c,j}$ or $y_{c,j}$ are supported in some thickening of $\Lambda_\Zz$. Hence $I_\Zz$ is contained in the closed linear span of $I_\Zz p I_\Zz$.
\end{proof}

Finally, we show that in the polynomial growth case the Schwartz algebras are closed under holomorphic functional calculus and therefore K-theory classes admit rapidly decaying representatives. 

\begin{proposition}\label{prop:spectral_invariance}
If $\Xx$ is a proper metric space with polynomial growth, then for $1 \leq p \leq \infty$, $\Ss^p_u(\Xx, \Hh_\Xx)$ and $\Ss^p_u(\Zz\subset\Xx, \Hh_\Xx)$ are closed under the holomorphic functional calculus of $C_u^*(\Xx, \Hh_\Xx)$.
\end{proposition}
\begin{proof}
Since $\Ss^p_u(\Xx, \Hh_\Xx)$ and $\Ss^p_u(\Zz\subset\Xx, \Hh_\Xx)$ are Fr\'echet algebras that the unitizations $\Ss^p_u(\Xx, \Hh_\Xx)^+$ and $\Ss^p_u(\Zz\subset\Xx, \Hh_\Xx)^+$ are inverse-closed in $C^*_u(\Xx, \Hh_\Xx)^+$ \cite{Schweitzer1992}.

Following \cite{JiYu2020UniformRoeRD}, for $z\in\Lambda$ define the unbounded self-adjoint operator $D_z=\sum_{\lambda\in\Lambda}d(\lambda,z)\chi_\lambda$
on $\Hh_\Xx=\bigoplus_{\lambda\in\Lambda}\Hh_\lambda$ and consider the commutator $\delta_z(T)=\imath[D_z,T]$, which is well-defined and a $*$-derivation for $T\in \Ss^\infty_u(\Xx,\Hh_\Xx)$. 
For $n\in\NM$ set
\[L_0(T)=\|T\|,  \qquad L_n(T)=\sup_{z\in\Lambda}\|\delta_z^n(T)\|, \quad n>0.\]

We will show that the family of seminorms $(L_n)_{n\in \NM}$ is equivalent to the defining seminorms of $\Ss^\infty_u(\Xx, \Hh_\Xx)$. For one direction one uses $\chi_\lambda\delta_z^n(T)\chi_\mu= \imath^n \bigl(d(\lambda,z)-d(\mu,z)\bigr)^n\chi_\lambda T\chi_\mu$ and notes that taking $z=\mu$ gives
\[d(\lambda,\mu)^n \|\chi_\lambda T\chi_\mu\| \leq L_n(T).\]
Since $\langle t\rangle^n \leq C_n(1+t^n),$
we obtain
\[\langle d(\lambda,\mu)\rangle^n
\|\chi_\lambda T\chi_\mu\|\leq C_n\bigl(\|\chi_\lambda T\chi_\mu\| + d(\lambda,\mu)^n \|\chi_\lambda T\chi_\mu\| \bigr) \leq C_n\bigl(\|T\|+L_n(T)\bigr),\]
and hence we get the bound
\[\|T\|_{n,\infty} \leq C_n \bigl(\|T\|+L_n(T)\bigr).\]
Conversely, by the reverse triangle inequality, for any $r>\nu+1$, 
\begin{align*}
\|\chi_\lambda\delta_z^n(T)\chi_\mu\| \leq
d(\lambda,\mu)^n
\|\chi_\lambda T\chi_\mu\|\leq
\|T\|_{n+r,\infty}
\langle d(\lambda,\mu)\rangle^{-r}.
\end{align*}
Lemma~\ref{lemma:schur_test} consequently gives $\|\delta_z^n(T)\| \leq C_r\|T\|_{n+r,\infty}$, with $C_r$ independent of $z$.
Taking supremums yields the reverse bound.

Consequently, the seminorms $(L_n)_{n\in\NM}$ and
$(q_{n,\infty})_{n\in\NM}$ define equivalent Fr\'echet topologies on
$\Ss_u^{\infty}(\Xx,\Hh_\Xx)$. By the Leibniz rule the equivalent norms satisfy the Blackadar-Cuntz differential seminorm condition
$$L_n(ab)\leq \sum_{j=0}^n \binom{n}{j} L_{n-j}(a)\, L_j(b),$$ so we can conclude (see e.g., \cite{Schweitzer1993SpectralInvariance}) that $\Ss_u^\infty(\Xx,\Hh_\Xx)$ is spectral invariant in $C^*_u(\Xx,\Hh_\Xx)$.

Spectral invariance of $\Ss^p_u(\Xx,\Hh_\Xx)$ then follows from the fact that $\Ss^p_u(\Xx,\Hh_\Xx)\subset \Ss^\infty_u(\Xx,\Hh_\Xx)$ is a two-sided Fr\'echet ideal: If $b \in \Ss_u^p(\Xx, \Hh_\Xx)$ and $1+b$ is invertible in $C^*_u(\Xx, \Hh_\Xx)$, then spectral invariance of $\Ss^\infty_u(\Xx, \Hh_\Xx)$ implies $(1+b)^{-1}\in \Ss^\infty_u(\Xx, \Hh_\Xx)^+$, and the ideal property gives
\[(1+b)^{-1} - 1 = -(1+b)^{-1} b \in \Ss^p_u(\Xx, \Hh_\Xx).\]

Similarly, spectral invariance for $\Ss^p_u(\Zz\subset\Xx, \Hh_\Xx)$ follows from the fact that it is a two-sided Fr\'echet ideal in $\Ss^\infty_u(\Xx, \Hh_\Xx)$, which is true by Lemma~\ref{lemma:boundary_estimate}. 
\end{proof}

\section{Kubo formula on \texorpdfstring{$\ZM^d$}{Z\textasciicircum d}}
\label{sec:kubozd}
In this section we prove the Kubo formula for the index for the space $\ZM^d$ with the standard Dirac operators. Define two cocycles
$$\eta_{\mathrm{Connes}}(a_0,\dots,a_d)= \begin{cases}  \frac12 \mathrm{Tr}(F_{\mathrm{rad}}[F_{\mathrm{rad}},a_0] [F_{\mathrm{rad}}, a_1]\dots[F_{\mathrm{rad}},a_d])  & d\text{ odd}\\
\frac12\mathrm{Tr}(\gamma_0 F_{\mathrm{rad}} [F_{\mathrm{rad}},a_0] [F_{\mathrm{rad}}, a_1]\dots[F_{\mathrm{rad}},a_d])  & d\text{ even}
\end{cases}
$$
and
$$\eta_{\mathrm{Kubo}}(a_0,\dots,a_d)= \sum_{\sigma\in S_d} (-1)^\sigma \mathrm{Tr}(a_0 [F_{\sigma(1)}, a_1]\dots[F_{\sigma(d)},a_d])$$
with the multiplication operators
$$F_{\mathrm{rad}}(x)= \sum_{i=1}^d \frac{x_i \gamma_i}{|x|}, \qquad F_i(x) = \sgn(x_i).$$
Here, the matrices $\gamma$ are the representation of $\CM_d$ defined in Appendix~\ref{sec:index}. We handle the special case $|x|=0$ by setting $F_{\mathrm{rad}}(0)=\gamma_1$.
\begin{theorem}
\label{th:kubozd}
Let $[D_{\mathrm{std}}]\in KK_{d}(C^*_u(\ZM^d),\CM)$ be the class defined by the standard Dirac operator $D_{\mathrm{std}} = \sum_i X_i \otimes \sigma_i$.

In even dimension $d$, let $e \in \Ss_u^1(\ZM^d) \otimes M_N(\CM)$ be a projection, then
$$ \Xi_d \eta_{\mathrm{Kubo}}(e,\dots,e)=\langle [e]_0, [D_{\mathrm{std}}]\rangle = \Gamma_d \eta_{\mathrm{Connes}}(e,\dots,e).$$

In odd dimension $d$, let $u \in \Ss_u^1(\ZM^d) \otimes M_N(\CM)$ be a unitary, then
$$ \Xi_d \eta_{\mathrm{Kubo}}(u^*,\dots,u)= \langle [u]_1, [D_{\mathrm{std}}]\rangle= \Gamma_d \eta_{\mathrm{Connes}}(u^*,\dots,u).$$ 

The normalization constants are
\begin{equation}\label{eqn:dirac_phase_index}
\Gamma_d = \begin{cases}
    (-1)^{d/2}  & d\text{ even}, \\
    (-1)^{(d+1)/2}2^{-d} & d\text{ odd}
\end{cases} \qquad \text{and} \qquad 
\Xi_d = \begin{cases}
    \frac{(-i\pi)^{d/2}}{d!!} & d\text{ even}, \\
    -\frac{ (-\imath \pi)^{(d-1)/2}}{d!!\; 2^d} & d\text{ odd}
\end{cases}
\end{equation}
\end{theorem}
The relation between the index pairing and the Connes cocycle is completely general: it follows from the Calderón–Fedosov index formula for sufficiently summable Fredholm modules, see e.g., \cite{ProdanSchulzBaldes2016B}. Using crossed-product algebras, it was shown in \cite[Chapter 6]{ProdanSchulzBaldes2016} that for magnetically covariant ergodic families, the same index pairing is computed by another cocycle
\begin{equation}
\label{eq:chern_cocycle}
\eta_{\mathrm{Chern}}(a_0,\ldots,a_d)= \sum_{\sigma \in S_d} (-1)^{\sigma} \Tt(a_0 \prod_{i=1}^d [X_{\sigma(i)}, a_i])
\end{equation}
with $\Tt$ the trace-per-unit volume (which is well-defined ergodic families of operators). It was pointed out in \cite{Kubota17} that the proof also works for non-ergodic operators, since the uniform Roe algebra can be written as a crossed-product algebra $C^*_u(\ZM^d)\simeq \ell^\infty(\ZM^d)\rtimes \ZM^d$ and a trace-per-unit volume can be constructed using ultrafilters. We will use that same argument to compute the pairing with $\eta_{\mathrm{Kubo}}$ as it (after some averaging) agrees with $\eta_{\mathrm{Chern}}$ up to normalization constants. We also note that, up to signs, the normalization constants coincide with those derived in \cite{LudewigThiang25} after accounting for the factor of $2^d$ which comes from replacing projections by sign functions.

An important ingredient in the proof is the fact that elements of the uniform Roe algebra admit generalized limits and can be averaged over. 
The easiest way to define limit operators is via the isomorphisms $$C^*_u(\ZM^d)\simeq \ell^\infty(\ZM^d)\rtimes \ZM^d \simeq  C(\beta \ZM^d)\rtimes \ZM^d$$
where one uses $\ell^\infty(\ZM^d)\simeq C(\beta \ZM^d)$ for $\beta\ZM^d$ the Stone-Čech compactification of $\ZM^d$. By the universal property of crossed products, every free ultrafilter $\omega \in \beta\ZM^d\setminus \ZM^d$ therefore generates a covariant representation $L_\omega: C^*_u(\ZM^d)\to \Bb(\ell^2(\ZM^d))$. 
For $\omega\in \ZM^d$ this so-called limit operator $L_\omega(a)$ is just a translate of $a$, whereas if $\omega\in \beta\ZM^d\setminus \ZM^d$ is a free ultrafilter then the limit can be non-trivial. More explicitly, the matrix elements of $L_\omega(a)$ are obtained by taking the ultrafilter limit
\begin{equation}
\label{eq:limit_op}
\langle \delta_y, L_\omega(a)\delta_z\rangle =\lim_{x\to \omega} \langle \delta_{y+x}, a\delta_{z+x}\rangle.
\end{equation}
This definition is a special case of the groupoid definition \cite[Appendix C]{SpakulaWillett2017Metric} and also agrees with other notions of limit operator. 

\begin{lemma}
\label{lemma:weaklimit}
With $(u^x)_{x\in \ZM^d}$ the shift operators on $\ell^2(\ZM^d)$ there exists for every $a \in C^*_u(\ZM^d)$ and every $\omega \in \beta\ZM^d$ some sequence $(x_n)_{n\in \NM}$ such that
$$L_\omega(a) = \wlim_{n\to \infty}u^{-x_n} a u^{x_n}.$$
\end{lemma}
\begin{proof}
We can write \eqref{eq:limit_op} as 
$$L_\omega(a) =\wlim_{x\to \omega} (u^{-x}a u^x)$$
with the ultrafilter limit taken with respect to the weak operator topology. Consider the orbit $(u^{-x} a u^{x})_{x\in \ZM^d}$ of some fixed $a\in C^*_u(\ZM^d)$. The limit with respect to some ultrafilter is the limit of a suitable net and therefore lies in the WOT-closure of the orbit. As norm-bounded sets are metrizable in the weak-operator topology, there also exists a sequence which converges to the limit point.
\end{proof}

The collection $(L_\omega(a))_{\omega\in \beta\ZM^d}$ of all limit operators of an element $a\in C^*_u(\ZM^d)$ is a strong-operator continuous family of operators on $\ell^2(\ZM^d)$. The space $\beta\ZM^d$ also carries a $\ZM^d$-action which corresponds to the translations on $\ell^\infty(\ZM^d)$, hence it is a covariant family. Since $\ZM^d$ is amenable one can construct translation-invariant averages:
\begin{lemma}
There exists a translation-invariant probability measure $\PM$ on $\beta\ZM^d$.
\end{lemma}
\begin{proof}
Choose a Følner sequence $\Lambda_n$ converging to $\ZM^d$ and a free ultrafilter $\overline{\omega}\in \beta\NM\setminus \NM$. Then
$$\Tt(f) = \lim_{n\to \overline{\omega}} \frac{1}{|\Lambda_n|} \sum_{x\in \Lambda_n} f(x)$$
defines a translation-invariant state on $\ell^\infty(\ZM^d)\simeq C(\beta \ZM^d)$ or equivalently a probability measure $\PM$ on $\beta\ZM^d$ such that
$$\Tt(f)=\int_{\beta\ZM^d} f(\omega) \difd\PM(\omega).$$ 
\end{proof}

To prove equality between the index and a Kubo-style formula one can now directly adapt the proof of \cite[Theorem~6.4.1]{ProdanSchulzBaldes2016}:

\begin{proposition}
For any $x\in \RM^d$, define $\eta^{(x)}_{\mathrm{Connes}}$ and $\eta_{\mathrm{Kubo}}^{(x)}$ by the same expressions as $\eta_{\mathrm{Connes}}$ respectively $\eta_{\mathrm{Kubo}}$ but with the shifted operators $F_{\mathrm{rad}}^{(x)}=F_{\mathrm{rad}}(\cdot-x)$ respectively $F_{i}^{(x)}=F_{i}(\cdot-x)$. 
Then 
$$\eta^{\mathrm{avg}}_{\mathrm{Connes}}(a_0,\dots,a_d) = \int_{\beta\ZM^d} \int_{\Cc^d} \eta^{(x)}_{\mathrm{Connes}}(L_\omega(a_0),\dots,L_{\omega}(a_d))\difd{x}\,\difd{\PM}(\omega)$$
and
$$\eta^{\mathrm{avg}}_{\mathrm{Kubo}}(a_0,\dots,a_d) = \int_{\beta\ZM^d} \int_{\Cc^d} \eta^{(x)}_{\mathrm{Kubo}}(L_\omega(a_0),\dots,L_{\omega}(a_d))\difd{x}\,\difd{\PM}(\omega)$$
averaged over the cube $\Cc^d=[0,1]^d$ are also cyclic cocycles and
\begin{equation}
\label{eq:norm_consts2}
\eta^{\mathrm{avg}}_{\mathrm{Connes}} = A_d \eta^{\mathrm{avg}}_{\mathrm{Kubo}}, \qquad A_d=\begin{cases}
    \frac{(\imath \pi)^{d/2}}{d!!} \qquad \text{if d is even}, \\
    \frac{ (\imath \pi)^{(d-1)/2}}{d!!} \qquad \text{if d is odd}
    \end{cases}
\end{equation}
\end{proposition}
\begin{proof}
Assume $a_0,\ldots,a_d$ are finite propagation. Writing it in terms of the symbol $\varphi_{F^{x}}$ (cf. Definition~\ref{def:F}) we get the finite sums
\begin{align*}
&\int_{\beta\ZM^d} \eta^{(x)}_{\mathrm{Kubo}}(L_\omega(a_0),\dots,L_\omega(a_d))\difd\PM(\omega) \\
&=\int_{\beta\ZM^d} \sum_{\substack{x_0,x_1,\dots,x_{d+1}\\ x_0=x_{d+1}}} \varphi_{F^{(x)}}(x_0,..,x_d) \prod_{i=0}^d\langle\delta_{x_i}, L_\omega(a_i)  \delta_{x_{i+1}}\rangle\difd\PM(\omega).
\end{align*}
By covariance with respect to translations in $\beta\ZM^d$ one has $$\langle\delta_{x_i}, L_\omega(a_i)  \delta_{x_{i+1}}\rangle=\langle\delta_{x_i-x_0}, u^{-x_0}L_\omega(a_i) u^{x_0} \delta_{x_{i+1}-x_0}\rangle=\langle\delta_{x_i-x_0}, L_{x_0\triangleright \omega}(a_i) \delta_{x_{i+1}-x_0},$$
using translation-invariance of $\PM$ the expression therefore becomes after reindexing
$$
\int_{\beta\ZM^d} \sum_{\substack{x_0,x_1,\dots,x_{d+1}\\ x_0=0=x_{d+1}}} \sum_{y\in \ZM^d}\varphi_{F^{(x)}}(y,x_1+y,..,x_d+y) \prod_{i=0}^d\langle\delta_{x_i}, L_\omega(a_i) \delta_{x_{i+1}}\rangle\difd\PM(\omega)$$
When we additionally average over the unit cube $\Cc^d$, then the sum over $y$ combined with the offset $x$ becomes an integral over $\RM^d$. One computes
\begin{align*}
&\int_{\RM^d}\varphi_{F}(y,x_1+y,..,x_d+y)\mathrm{d}y\\
&=\sum_{\sigma \in S_d} (-1)^{\sigma} \int_{\RM^d} \prod_{i=1}^d (\sgn((x_{i}+y)\cdot e_{\sigma(i)}) - \sgn((x_{i-1}+y)\cdot e_{\sigma(i)}))\difd{y} \\
&=\sum_{\sigma \in S_d} (-1)^{\sigma} \prod_{i=1}^d \int_{\RM} (\sgn((x_{i}+z)\cdot e_{\sigma(i)}) - \sgn((x_{i-1}+z)\cdot e_{\sigma(i)}))\difd{z}\\
&= 2^d \, \sum_{\sigma \in S_d} (-1)^{\sigma} \prod_{i=1}^d  (x_{i} \cdot e_{\sigma(i)} - x_{i-1} \cdot e_{\sigma(i)})= 2^d\, \mathrm{det}(x_1,\dots,x_d).
\end{align*}
Here we used that each factor depends only on one independent direction $e_{\sigma(i)}$, so we split into one-dimensional integrals that were computed using the simple observation 
$$\int_\RM (\sgn(\lambda_1+t)-\sgn(\lambda_2+t))\difd t = 2 (\lambda_1-\lambda_2)$$
for every $\lambda_1,\lambda_2 \in \RM$. We conclude
$$\eta^{\mathrm{avg}}_{\mathrm{Kubo}}(a_0,\dots,a_d) = \sum_{\substack{x_0,x_1,\dots,x_{d+1}\\ x_0=0=x_{d+1}}} 2^d \det(x_1,\ldots,x_d) \int_{\beta\ZM^d}\prod_{i=1}^d\langle\delta_{x_i}, L_\omega(a_i)\delta_{x_{i+1}}\rangle\difd\PM(\omega).$$

Similar manipulations for $F_{\mathrm{rad}}$ allow one to pull out the Clifford factors in $\eta_{\mathrm{Connes}}$ and shift using translation-invariance (see the proof of \cite[Theorem 6.3.1]{ProdanSchulzBaldes2016}) to get
\begin{align*}
\eta^{\mathrm{avg}}_{\mathrm{Connes}}(a_0,\dots,a_d) 
=\sum_{\substack{x_0,x_1,\dots,x_{d+1}\\ x_0=0=x_{d+1}}}  C_d \mathrm{det}(x_1,\dots,x_d)\int_{\beta\ZM^d}\prod_{j=0}^d\langle\delta_{x_j}, L_\omega(a_j)\delta_{x_{j+1}}\rangle\difd\PM(\omega) 
\end{align*}
with the constants given by \cite[Lemma~6.4.1]{ProdanSchulzBaldes2016}, 
$$C_d =\int_{\RM^d} \Tr\biggl(\gamma_0 \prod_{i=1}^d (F_{\mathrm{rad}}(x_{i}+y)- F_{\mathrm{rad}}(x_{i+1}+y))\cdot \gamma\biggr)=\begin{cases}
\frac{2^{d}(\imath \pi)^k}{d!!} & d=2k+1,\\
\frac{(2\imath \pi)^k}{k!} & d=2k.
\end{cases}$$
Comparing the normalization constants finishes the proof.
\end{proof}

\begin{proof}[Proof of Theorem~\ref{th:kubozd}]
We will write the proof for the even case, with the odd case very similar. Let $e\in \Ss_u^1(\ZM^d)\otimes M_N(\CM)$ be a projection.

The shifted Connes cocycle computes the pairing with the standard KK-cycle on $\ell^2(\ZM^d)\hat{\otimes}\CM_d$ given by the shifted Dirac operator $D^{(x)}_{\mathrm{std}}= \sum_{i=1}^d (X_i-x_i)\otimes \sigma_i$. One has \cite[Section 6.3, 6.4]{ProdanSchulzBaldes2016B}
$$\langle [e]_0, [D^{(x)}_{\mathrm{std}}]\rangle = \Gamma_d \eta^{(x)}_{\mathrm{Connes}}(e,\dots,e)$$
with the normalization constants from Theorem~\ref{th:kubozd}. Since $D_{\mathrm{std}}-D^{(x)}_{\mathrm{std}}$ is relatively compact w.r.t. $D_{\mathrm{std}}$ the left-hand side does not depend on it and hence one can average
$$\langle [e]_0 , [D_{\mathrm{std}}] \rangle= \Gamma_d \int_{\Cc^d}\eta^{(x)}_{\mathrm{Connes}}(e,\dots,e) \mathrm{d}x.$$
Similarly, $\eta^{(x)}_{\mathrm{Kubo}}-\eta_{\mathrm{Kubo}}$ is a coboundary by Proposition~\ref{prop:coarse_bounded_perturbation} and hence the pairing between $\eta^{(x)}_{\mathrm{Kubo}}$ and $K_{d \bmod 2}(C^*_u(\ZM^d))$ does not depend on the offset $x$ either, so we can also average
$$\eta_{\mathrm{Kubo}}(e,\dots,e)= \int_{\Cc^d}\eta^{(x)}_{\mathrm{Kubo}}(e,\dots,e) \mathrm{d}x.$$

By \eqref{eq:norm_consts2} we therefore merely need to show that we can also drop the averages over $\beta\ZM^d$ in the equations
\begin{align*}
\int_{\beta\ZM^d} \langle [L_\omega(e)], [D_{\mathrm{std}}]\rangle \difd{\PM(\omega)}&=\Gamma_d \int_{\beta\ZM^d}  \eta_{\mathrm{Connes}}(L_\omega(e),\dots,L_\omega(e))\difd{\PM(\omega)} \\
&= \Gamma_d A_d \int_{\beta\ZM^d}  \eta_{\mathrm{Kubo}}(L_\omega(e),\dots,L_\omega(e))\difd{\PM(\omega)}
\end{align*}
for $\Gamma_d, A_d$ as in Equation \eqref{eqn:dirac_phase_index} and \eqref{eq:norm_consts2} respectively. Note also that $\Xi_d = \Gamma_d A_d$, which is the origin of that normalization constant.

Suppose now we have a sequence $e = \wlim_{n\to \infty} e_n$ in $\Ss_u^1(\ZM^d)$ with a uniform bound 
$$\sup_{n\in \NM} \sup_{x,y\in \ZM^d} |\langle \delta_x, e_n \delta_y\rangle| \langle x-y\rangle^k < \infty$$
for large enough $k$. Let $\eta$ be either $\eta_{\mathrm{Kubo}}$ or $\eta_{\mathrm{Connes}}$. Since both of them are given by sums over matrix elements in the form $$\eta(a_0,\dots,a_d)=\sum_{\substack{x_0,x_1,\dots,x_{d+1}\\ x_0=x_{d+1}}} \varphi(x_0,..,x_d) \prod_{i=0}^d\langle\delta_{x_i}, a_i  \delta_{x_{i+1}}\rangle,$$
with a summable majorant, we can conclude by dominated convergence that
$$\lim_{n\to \infty} \eta(e_n,\dots,e_n)= \eta(e,\dots,e).$$

By Lemma~\ref{lemma:weaklimit}, we can approximate $e$ in weak operator topology by $e_n=u^{-x_n} e u^{x_n}$ which has such uniform bounds on the matrix elements by assumption $e\in \Ss_u^1(\ZM^d)\otimes M_N(\CM)$. One therefore has
$$\eta(L_\omega(e),\dots,L_\omega(e))=\lim_{n\to \infty} \eta(e_n,\dots,e_n)= \eta(e,\dots,e)$$
since each $e_n$ is a projection and the pairing of $\eta_{\mathrm{Kubo}}$ respectively $\eta_{\mathrm{Connes}}$ with any projection is translation-invariant.
\end{proof}

\begin{remark}{\rm 
The relation to \eqref{eq:chern_cocycle} can be derived with the same proof, as the generalized trace-per-unit-volume can be written in the form
$$\Tt(a)=\int_{\beta\ZM^d}\langle \delta_0, L_\omega(a)\delta_0\rangle \difd\PM(\omega).$$
Hence, if one averages the Chern cocycle similarly, one gets $\eta^{\mathrm{avg}}_{\mathrm{Kubo}}= 2^d\eta^{\mathrm{avg}}_{\mathrm{Chern}}$.}$\diamond$
\end{remark}

\textbf{Acknowledgements}
This research was supported by NSF DMS-2052899, DMS-2155211, and Simons
896624. AM is supported by NSF through GRFP grant DGE-2146752.

\textbf{Data Availability Statement}
No datasets were created or analyzed during the present work, so data sharing is not relevant to this article.

\printbibliography

\end{document}